\documentclass[acmsmall,screen]{acmart}
\AtBeginDocument{%
  \providecommand\BibTeX{{%
    \normalfont B\kern-0.5em{\scshape i\kern-0.25em b}\kern-0.8em\TeX}}}

\setcopyright{none}

\acmJournal{PACMPL}
\acmVolume{1}
\acmNumber{POPL} 
\acmArticle{1}
\acmYear{2021}
\acmMonth{1}
\acmDOI{} 
\startPage{1}

\usepackage{stmaryrd} 
\usepackage{proof}
\usepackage{amsmath}
\usepackage{mathtools}
\usepackage{amsthm}
\usepackage{mathpartir}
\usepackage{color}
\usepackage{xstring}
\usepackage{ntabbing}
\usepackage{listings}
\usepackage{varwidth}
\usepackage[shortlabels]{enumitem}
\usepackage{multicol}
\usepackage{lipsum}
\usepackage{booktabs}
\usepackage{subcaption}
\usepackage{fancyvrb}
\usepackage[cal=cm]{mathalfa}
\usepackage{turnstile}
\usepackage{cleveref}
\crefname{theorem}{Theorem}{Theorems}
\crefname{lemma}{Lemma}{Lemmas}
\crefname{corollary}{Corollary}{Corollaries}
\crefname{figure}{Figure}{Figures}
\crefname{definition}{Definition}{Definitions}
\crefname{table}{Table}{Tables}
\crefname{section}{Section}{Sections}
\crefname{example}{Example}{Examples}
\crefname{proposition}{Proposition}{Propositions}

\setlist{nosep,leftmargin=\parindent}
\AtBeginDocument{%
  \setlength\abovedisplayskip{0.5\abovedisplayskip}%
  \setlength\belowdisplayskip{0.5\belowdisplayskip}%
  \setlength\abovedisplayshortskip{0.5\abovedisplayshortskip}%
  \setlength\belowdisplayshortskip{0.5\belowdisplayshortskip}%
  \setlength\floatsep{0.5\floatsep}%
  \setlength\abovecaptionskip{0.5\abovecaptionskip}%
}

\newlength{\rWidth}

\newcommand{\lang}{PRast}
\newcommand{\heads}{\m{H}}
\newcommand{\tails}{\m{T}}
\newcommand{\eflip}[3]{\m{flip} \; #1 \; (\heads \Rightarrow #2 \mid \tails \Rightarrow #3)}
\newcommand{\pichoiceop}{\oplus_\m{P}}
\newcommand{\pichoice}[1]{\pichoiceop \{ #1 \}}
\newcommand{\pechoiceop}{\with_\m{P}}
\newcommand{\pechoice}[1]{\pechoiceop \{ #1 \}}
\newcommand{\esendlp}[2]{#1..#2}
\newcommand{\epcase}[3]{\m{pcase} \; #1 \; (#2 \Rightarrow #3)}

\newcommand{\m}[1]{\mathsf{#1}}
\newcommand{\mb}[1]{\mathbf{#1}}
\newcommand{\dist}{\ensuremath{\mathbb{P}}}

\renewcommand{\l}{\m{L}}

\newcommand{\Sg}{\Sigma}
\newcommand{\xvdash}[1]{%
  \vdash^{\mkern-8mu\scriptstyle\rule[-.9ex]{0pt}{0pt}#1}%
}

\newcommand{\potconf}[1]{\overset{#1}{\Vdash}}

\newcommand{\D}{\Delta}

\newcommand{\proc}[2]{\m{proc}(#1, #2)}
\newcommand{\msg}[2]{\m{msg}(#1, #2)}

\newcommand{\step}{\mapsto}

\newcommand{\ecase}[3]{\m{case} \; #1 \; (#2 \Rightarrow #3)}

\newcommand{\erecvch}[2]{#2 \leftarrow \m{recv} \; #1}

\newcommand{\esendch}[2]{\m{send} \; #1 \; #2}

\newcommand{\ewait}[1]{\m{wait} \; #1}

\newcommand{\eclose}[1]{\m{close} \; #1}

\newcommand{\fwd}[2]{#1 \leftrightarrow #2}

\newcommand{\esendl}[2]{#1.#2}

\newcommand{\ecut}[4]{#1 \leftarrow #2 \; #3 \semi #4}

\newcommand{\ework}[1]{\m{work} \, \{#1\}}
\newcommand{\eget}[2]{\m{get} \, #1 \, \{#2\}}
\newcommand{\epay}[2]{\m{pay} \, #1 \, \{#2\}}
\newcommand{\procdef}[3]{#3 \leftarrow #1 \; #2}

\newcommand{\lolli}{\multimap}
\newcommand{\tensor}{\otimes}
\newcommand{\with}{\mathbin{\binampersand}}

\newcommand{\one}{\mathbf{1}}

\newcommand{\semi}{\; ; \;}
\newcommand{\ichoiceop}{\oplus}
\newcommand{\echoiceop}{\with}
\newcommand{\ichoice}[1]{\ichoiceop \{ #1 \}}
\newcommand{\echoice}[1]{\echoiceop \{ #1 \}}

\newcommand{\mi}[1]{\mbox{\it #1}}

\newcommand{\entailpot}[1]{\xvdash{#1}}

\newcommand{\paypot}{\triangleright}
\newcommand{\getpot}{\triangleleft}
\newcommand{\tgetpot}[2]{\getpot^{#2} #1}
\newcommand{\tpaypot}[2]{\paypot^{#2} #1}

\newcommand{\Next}{\raisebox{0.3ex}{$\scriptstyle\bigcirc$}}
\renewcommand{\next}[1]{\Next #1}
\newcommand{\tdelay}[2]{
    \IfEqCase{#2}{%
        {1}{\next{#1}}%
    }[{\Next^{#2} (#1)}]%
}%

\newcommand{\st}[1]{\m{store}_{#1}}

\newcommand{\bool}{\m{bool}}

\newcommand{\dom}[1]{\m{dom}(#1)}

\newcommand{\down}{\downarrow^{\m{S}}_{\m{L}}}
\newcommand{\up}{\uparrow^{\m{S}}_{\m{L}}}

\newtheorem{theorem}{Theorem}

\newtheorem{lemma}{Lemma}

\definecolor{verylightgray}{rgb}{.97,.97,.97}

\lstdefinelanguage{Solidity}{
	keywords=[1]{anonymous, assembly, assert, balance, break, call, callcode, case, catch, class, constant, continue, constructor, contract, debugger, default, delegatecall, delete, do, else, emit, event, experimental, export, external, false, finally, for, function, gas, if, implements, import, in, indexed, instanceof, interface, internal, is, length, library, log0, log1, log2, log3, log4, memory, modifier, new, payable, pragma, private, protected, public, pure, push, require, return, returns, revert, selfdestruct, send, solidity, storage, struct, suicide, super, switch, then, this, throw, transfer, true, try, typeof, using, value, view, while, with, addmod, ecrecover, keccak256, mulmod, ripemd160, sha256, sha3}, 
	keywordstyle=[1]\color{blue}\bfseries,
	keywords=[2]{address, bool, byte, bytes, bytes1, bytes2, bytes3, bytes4, bytes5, bytes6, bytes7, bytes8, bytes9, bytes10, bytes11, bytes12, bytes13, bytes14, bytes15, bytes16, bytes17, bytes18, bytes19, bytes20, bytes21, bytes22, bytes23, bytes24, bytes25, bytes26, bytes27, bytes28, bytes29, bytes30, bytes31, bytes32, enum, int, int8, int16, int24, int32, int40, int48, int56, int64, int72, int80, int88, int96, int104, int112, int120, int128, int136, int144, int152, int160, int168, int176, int184, int192, int200, int208, int216, int224, int232, int240, int248, int256, mapping, string, uint, uint8, uint16, uint24, uint32, uint40, uint48, uint56, uint64, uint72, uint80, uint88, uint96, uint104, uint112, uint120, uint128, uint136, uint144, uint152, uint160, uint168, uint176, uint184, uint192, uint200, uint208, uint216, uint224, uint232, uint240, uint248, uint256, var, void, ether, finney, szabo, wei, days, hours, minutes, seconds, weeks, years},	
	keywordstyle=[2]\color{teal}\bfseries,
	keywords=[3]{block, blockhash, coinbase, difficulty, gaslimit, number, timestamp, msg, data, gas, sender, sig, value, now, tx, gasprice, origin},	
	keywordstyle=[3]\color{violet}\bfseries,
	identifierstyle=\color{black},
	sensitive=false,
	comment=[l]{//},
	morecomment=[s]{/*}{*/},
	commentstyle=\color{gray}\ttfamily,
	stringstyle=\color{red}\ttfamily,
	morestring=[b]',
	morestring=[b]"
}

\newcommand{\Rule}[4][]{\ensuremath{\inferrule*[lab={\small(#2)},#1]{#3}{#4}}}

\makeatletter
\newcommand{\xMapsto}[2][]{\ext@arrow 0599{\Mapstofill@}{#1}{#2}}
\def\Mapstofill@{\arrowfill@{\Mapstochar\Relbar}\Relbar\Rightarrow}
\makeatother

\newcommand{\Dl}{\Delta}
\newcommand{\Gm}{\Gamma}

\newcommand{\calC}{\mathcal{C}}
\newcommand{\calD}{\mathcal{D}}
\newcommand{\calE}{\mathcal{E}}
\newcommand{\calI}{\mathcal{I}}
\newcommand{\calJ}{\mathcal{J}}

\newcommand{\calO}{\mathcal{O}}

\newcommand{\bbN}{\vvmathbb{N}}
\newcommand{\bbE}{\vvmathbb{E}}

\newcommand{\bind}{\mathbin{\gg\!=}}
\newcommand{\dplus}{\mathbin{{+}\!\!{+}}}
\newcommand{\defeq}{\coloneqq}

\newcommand{\many}[1]{\overline{#1}}
\newcommand{\tsum}{\smallsum\nolimits}
\newcommand{\tuple}[1]{\langle #1 \rangle}
\newcommand{\Forall}[1]{\ensuremath{\forall {#1}\!:}}
\newcommand{\Exists}[1]{\ensuremath{\exists {#1}\!:}}

\newcommand{\poised}{\;\m{poised}}
\newcommand{\live}{\;\m{live}}
\newcommand{\cpoised}[1]{\;{#1}\m{\textsf{-}poised}}
\newcommand{\cblocked}[1]{\;{#1}\m{\textsf{-}blocked}}
\newcommand{\comm}[1]{\;{#1}\m{\textsf{-}comm}}

\newcommand{\ostep}{\xmapsto{\mathrm{det}}}
\newcommand{\pstep}{\xmapsto{\mathrm{prob}}}
\newcommand{\dstep}{\Mapsto}
\newcommand{\sstep}{\mapsto}
\newcommand{\cstep}[1]{\xMapsto{#1}}

\begin{document}

\title{Probabilistic Resource-Aware Session Types}

\author{Ankush Das}
\affiliation{
  \institution{Carnegie Mellon University}
  \country{USA}
}
\email{ankushd@cs.cmu.edu}

\author{Di Wang}
\affiliation{
  \institution{Carnegie Mellon University}
  \country{USA}
}
\email{diw3@cs.cmu.edu}

\author{Jan Hoffmann}
\affiliation{
  \institution{Carnegie Mellon University}
  \country{USA}
}
\email{jhoffmann@cmu.edu}


\renewcommand{\shortauthors}{Das et al.}

\begin{abstract}
  
Session types guarantee that message-passing
processes adhere to predefined communication protocols.
Prior work on session types has focused on deterministic languages
but many message-passing systems, such as Markov chains and randomized
distributed algorithms, are probabilistic.
To model and analyze such systems, this
article introduces probabilistic session types and explores
their application in automatic expected resource analysis.
Probabilistic session types describe probability distributions
over messages and are a conservative extension of intuitionistic
(binary) session types.
To send on a probabilistic channel, processes have to utilize internal
randomness from a probabilistic branching expression or external
randomness from receiving on a probabilistic channel.
The analysis for expected resource bounds is integrated with the type
system and is a variant of automatic amortized resource analysis.
It can automatically derive symbolic bounds for different cost metrics
by reducing type inference to linear constraint solving.
The technical contributions include the meta theory that is based on a
novel nested multiverse semantics and a
type-reconstruction algorithm that allows flexible mixing of different
sources of randomness without burdening the programmer with type
annotations.
The type system has been implemented in the language \lang{}.
Experiments demonstrate that \lang{} is applicable in different
domains such as resource analysis of randomized distributed
algorithms, verification of limiting distributions in Markov
chains, and analysis of probabilistic digital contracts.


\end{abstract}



\maketitle

\section{Introduction}\label{sec:intro}

Session types statically describe communication protocols for message-passing
processes and well-typedness ensures adherence to these protocols at runtime.
Session types were introduced by Honda~\cite{Honda93CONCUR,Honda98esop}
for defining binary communication between processes,
which can be interpreted classically~\cite{Wadler12icfp} or
intuitionistically~\cite{Caires10concur}.
%
In this work, we follow the intuitionistic interpretation~\cite{Caires16mscs,Pfenning15fossacs}
of session types that is in Curry-Howard correspondence with intuitionistic
linear logic~\cite{Girard87tcs} and guarantees the \emph{absence of deadlocks}
(global progress) and \emph{session fidelity} (type preservation).
%

Existing work on session types has focused on deterministic (i.e.,
non-probabilistic) languages but many message-passing systems are
naturally probabilistic.
Often, distributed algorithms internally use randomization as a tool to
overcome limitations of deterministic algorithms. Examples of such
algorithms include~\citet{Lehmann81POPL} randomized dining philosophers
protocol or~\citet{Itai90IC} distributed leader election.
In other systems, probability distributions are used to model
uncertainty of external events such as incoming jobs in a data center.
More generally, Markov chains can be viewed as systems of
probabilistic message-passing processes.

This article presents \lang{}, a concurrent probabilistic programming
language with novel \emph{probabilistic session types} that can be used to
model and analyze probabilistic message-passing systems.
The key novelty is an additional type former that assigns a
prob.\ distribution to a choice of labels in a session type.
Such a probabilistic internal choice would, for instance, prescribe
that the label $\m{H}$ has to be sent on a channel with probability
$0.25$ and the label $\m{T}$ has to be sent with probability $0.75$.
To adhere to this type, a process can utilize two sources of
randomness: A new term former for probabilistic branching provides an
internal source of randomness and receiving messages on a probabilistic
channel according to some distribution provides an external source of
randomness.
Probabilistic session types are carefully designed to be a \emph{conservative
extension} of intuitionistic (binary) session types. In particular,
we support both probabilistic and standard choice types, which is
technically challenging.

In addition to the design of the type system, a major technical
contribution is an efficient type-reconstruction algorithm that
ensures that prob.\ distributions of labels sent on a probabilistic
channel matches the specification of the session type.
The type rules provide a high degree of flexibility and the correct
prob.\ distribution of probabilistic communication can be achieved by
nesting internal and external sources of randomness.
This flexibility complicates reconstruction of type derivations as we
do not want to burden programmers with providing channel probabilities for different
branches of a probabilistic split.
We solve this challenge by automatically reconstructing probabilities
in type derivations by generating and solving linear constraints.

A distinguishing feature of \lang{} is that it can be employed
for \emph{lightweight verification} of probabilistic systems. Probabilistic
session types statically guarantee the prob.\ distribution of
messages on a certain channel. These can naturally be composed
to produce a desired output prob.\ distribution given an input
prob.\ distribution. We have used this feature to verify correctness of
dice programs~\cite{Knuth76} that model dice from coin flips,
prove that a biased coin can be converted to an unbiased coin,
and verify limiting distributions of standard Markov chains.
In some situations, \lang{} can also \emph{infer the prob. distribution}
on a given channel. For instance, if the input probabilities
for a process are given, the output probabilities can be inferred
automatically.

As an application of \lang{}, we focus on automatically
deriving \emph{expected cost bounds for concurrent message-passing systems}.
Probabilistic models are often used as the basis for a quantitative
analysis such as the expected number of messages exchanged (in
randomized distributed protocols), the expected response time (quality
of service in a data center), or expected number of state transitions
in a Markov chain.
To derive bounds on such quantities, \lang{}'s type system combines
two recently introduced techniques: session-types for work analysis of
(deterministic) concurrent systems~\cite{Das18RAST} and expected cost
analysis for functional probabilistic programs~\cite{WangKH20}.
Both of these techniques can be seen as a variant of automatic
amortized resource analysis (AARA)~\cite{Jost03,HoffmannW15} and inference
can be reduced to standard linear programming (LP).
The newly designed AARA for \lang{} has interesting and
non-trivial interactions with probabilistic channels that enable a
more compositional analysis resulting in precise bounds.
While there are several techniques for automatic expected cost
analysis of sequential probabilistic programs~\cite{PLDI:NCH18,PLDI:WFG19,CAV:CFG16},
we are only aware of manual rule-based reasoning systems~\cite{Tassarotti19POPL,JTCS:MRS16}
or probabilistic model checking techniques~\cite{Kwiatkowska11CAV} for
analyzing cost of concurrent probabilistic programs. 


The meta theory of \lang{}---another contribution of this
work---is based on a novel variant of progress and type preservation.
The preservation proof is challenging as state-of-the-art techniques
for probabilistic programming languages~\cite{ICFP:BLG16,LICS:ALG19} do not
directly apply.
From an operational point of view, the problem is that different
universes that result from a probabilistic split cannot be considered
in isolation since the prob.\ distributions of messages on the channels
diverge.
Our solution is to use a \emph{nested-multiverse} semantics that manage the divergence of
message distributions to control their impact on other processes.

We have implemented \lang{} and performed experiments with
probabilistic message-passing systems from different domains.
The implementation of \lang{} infers, for instance, bounds on the
expected cost of randomized distributed protocols, such as Lehmann
and Rabin's randomized dining philosophers algorithm~\cite{Lehmann81POPL}, Itai and
Rodeh's synchronous leader election protocol~\cite{Itai90IC}, and Chaum's
dining cryptographers protocol~\cite{Chaum88JC}.
We have also implemented and verified several standard Markov chain-models
such as Google's PageRank algorithm~\cite{Page99TR}, random walks, and
dice programs.
%
Finally, we implemented some case studies that showcase how probabilistic
session types can be used as specifications of probabilistic digital contracts, such as
lotteries and slot machines.  \lang{} can verify winning probabilities
and automatically compute bounds on the expected financial gain/loss.

In summary, this article contains the following contributions. 
\begin{itemize}
\item The design of \lang{}, a language with probabilistic session
  types and a flexible type system for probabilistic send and receive.
  (Section~\ref{sec:prob})
\item An AARA for deriving symbolic bounds on the expected cost of
  \lang{} programs. (Section~\ref{sec:prob})
\item The soundness proof of the type system with a novel
  probabilistic nested-multiverse semantics establishing global progress, session
  fidelity and probability consistency. (Section~\ref{sec:metatheory})
\item An efficient type reconstruction algorithm that derives probabilities based
    on linear constraint solving. (Section~\ref{sec:impl})
\item An implementation of \lang{} and an experimental evaluation with
  distributed algorithms, digital contracts, and Markov chains. (\cref{sec:impl,sec:examples})
\end{itemize}

We start with an informal overview of \lang{} (Section~\ref{sec:overview})
and a recapitulation of session types for work analysis, which \lang{}
extends (Section~\ref{sec:formal}).


\section{Overview of Probabilistic Session Types}\label{sec:overview}
We briefly overview probabilistic session types with a series of
illustrative examples.
We follow the approach and syntax of Rast~\cite{Das20FSCD,Das20CONCUR,Das20PPDP},
which is based on a Curry-Howard isomorphism between intuitionistic
linear logic and session types, extended by recursively defined types
and processes. In this intuitionistic approach, every channel has a
unique provider and a client. We view the
session type as describing the communication from the provider's
point of view, with the client having to perform dual actions.

As a first example, consider the session type $\m{bool}$ defined as
\begin{center}
  \begin{minipage}{0cm}
  \begin{tabbing}
  $\m{bool} \triangleq \ichoice{\mb{true} : \one, \mb{false} : \one}$
  \end{tabbing}
  \end{minipage}
\end{center}
Here, the \emph{internal choice} type constructor $\ichoiceop$
dictates that the provider must send either $\mb{true}$ or $\mb{false}$.
In either case, the continuation type (after the colon) is $\one$,
indicating the end of the communication and requiring the provider to terminate
after sending a $\m{close}$ message.

As a first example, we define a simple process $\m{TT}$ that outputs $\mb{true}$ and
terminates. 
\begin{lstlisting}
  decl TT : . |- (b : bool)
  proc b <- TT = b.true ; close b
\end{lstlisting}
The first line declares the $\m{TT}$ process showing that it uses
an empty context (dot before the turnstile) and offers the channel $b$
of type $\m{bool}$. The second line shows the process definition.
The term $b \leftarrow \m{TT}$
is the syntax for defining (or spawning) process $\m{TT}$ offering on $b$
and using no channels.
The term $\esendl{b}{\mb{true}}$ denotes sending the
label $\mb{true}$ on $b$ and $\eclose{b}$ denotes closing the
channel $b$ and terminating.  A similar $\m{FF}$ process can be defined
that outputs $\mb{false}$.

\paragraph{\textbf{Probabilistic Processes}}
Suppose we wish to define a process $\m{TF}$ that outputs $\mb{true}$
with probability $p \in [0,1]$ and $\mb{false}$ with probability $1-p$.
We introduce a new probabilistic term $\eflip{p}{P_1}{P_2}$,
operationally interpreted as flipping a coin that outputs
heads with prob. $p$ and tails with prob. $1-p$. If the
coin outputs heads ($\heads$), we execute $P_1$, otherwise
we execute $P_2$. 
We employ this term to define process $\m{TF}$
with $p = 0.6$ (we only allow constant probabilities).
\begin{lstlisting}
  decl TF : . |- (b : bool)
  proc b <- TF = flip 0.6 (H => b.true ; close b  |  T => b.false ; close b)
\end{lstlisting}
The $\m{TF}$ process first flips a coin with prob. of $\heads$ being $0.6$.
If the coin flips to $\heads$, the process sends the label $\mb{true}$
and terminates. If the coin flips to $\tails$, the
process sends $\mb{false}$ and terminates. Since the prob. of
$\heads$ is $0.6$, the process $\m{TF}$ outputs $\mb{true}$ with prob.
$0.6$, and $\mb{false}$ with prob. $1-0.6=0.4$.

\paragraph{\textbf{Negation}}
Suppose we consider a negation process $\m{neg}$ that takes
a channel $b : \m{bool}$ as input and negates it (output
$\mb{false}$ if input is $\mb{true}$ and vice-versa).
\begin{lstlisting}
  decl neg : (b : bool) |- (c : bool)
  proc c <- neg b = case b ( true => c.false ; wait b ; close c
                           | false => c.true ; wait b ; close c )
\end{lstlisting}
The declaration describes that the $\m{neg}$ process uses channel
$b : \m{bool}$ and provides $c : \m{bool}$. This is similarly denoted
in the definition as $c \leftarrow \m{neg} \; b$. The definition
branches on the label received on channel $b$: if the process
receives $\mb{true}$, it sends $\mb{false}$ on $c$ and vice-versa.
Then, in either case, the process waits for the channel $b$ to
close using the term $\ewait{b}$ and then closes channel $c$.

\paragraph{\textbf{Probabilistic Session Types}}
Although processes can exhibit probabilistic behavior, this
information is not visible in their session types. In particular, the $\m{neg}$
process is unaware of the probability of $\mb{true}$ or $\mb{false}$
along channel $b$. Therefore, in this article, we introduce novel \emph{probabilistic
session types} that assign probabilities to the labels in a
session type. We introduce a \emph{probabilistic internal choice
type operator} $\pichoice{\ell^{p_\ell} : A_\ell}$: the provider
sends label $\ell$ with prob. $p_\ell$ and continues to provide type
$A_\ell$. 
The dual type is $\pechoice{\ell^{p_\ell} : A_\ell}$ where the provider
is guaranteed to receive label $\ell$ with probability $p_\ell$.
As illustrations, we describe several prob. session types.
\begin{align*}
  \m{tbool} \triangleq \pichoice{\mb{true}^{1.0} : \one, \mb{false}^{0.0} : \one} & \hspace{5em} &
  \m{fbool} \triangleq \pichoice{\mb{true}^{0.0} : \one, \mb{false}^{1.0} : \one} \\
  \m{pbool} \triangleq \pichoice{\mb{true}^{0.6} : \one, \mb{false}^{0.4} : \one} & \hspace{5em} &
  \m{npbool} \triangleq \pichoice{\mb{true}^{0.4} : \one, \mb{false}^{0.6} : \one}
\end{align*}
The type $\m{tbool}$ always outputs $\mb{true}$,
i.e. with prob. $1$. Similarly, the type $\m{fbool}$ always outputs
$\mb{false}$. The type $\m{pbool}$ outputs $\mb{true}$ with prob. $0.6$
and $\mb{false}$ otherwise. Its negation type $\mb{npbool}$ outputs
$\mb{true}$ with prob. $0.4$ and $\mb{false}$ otherwise.  With these
types, \emph{without changing the process definitions}%
\footnote{%
  In this article, we actually distinguish between standard and
  probabilistic send and case analysis for clarity. However, it is not
  necessary to make this distinction in the surface syntax.}%
, we obtain the following types for the
aforementioned processes.
Note that the previous typings with the type $\m{bool}$ are also valid
but provide less information.
\begin{lstlisting}
  decl TT : . |- (b : tbool)        decl FF : . |- (b : fbool)
  decl TF : . |- (b : pbool)        decl neg : (b : pbool) |- (c : npbool)
\end{lstlisting}
The soundness theorem ensures that the distribution of the labels sent on a
probabilistic channel at runtime does indeed match the distribution
on the labels in the choice types.
To send on a probabilistic channel, a process
can use two sources of randomness: a $\m{flip}$ (like in $\m{TF}$) or
$\m{case}$ on labels received on another probabilistic channel according
to a known distribution (like in $\m{neg}$).
These sources of randomness can be combined and nested as long as the
resulting distributions are valid. For instance, we can define an $\m{unbias}$
process that uses a biased coin ($b : \m{pbool}$) and produces an unbiased
coin ($c : \m{ubool} \triangleq \pichoice{\mb{true}^{0.5} : \one, \mb{false}^{0.5} : \one}$).
In each branch, the process flips a fair coin to decide whether to negate
the input or not. Since the input is copied or negated with \emph{equal probability},
we are able debias the input using a combination of $\m{flip}$
and $\m{case}$.
\begin{lstlisting}
decl unbias : (b : pbool) |- (c : ubool)
proc c <- unbias b = case b (
   true => flip 0.5 (H => c.false ; wait b ; close c | T => c.true ; wait b ; close c )
| false => flip 0.5 (H => c.true ; wait b ; close c | T => c.false ; wait b ; close c ))
\end{lstlisting}
A contribution of the article is an efficient type checking
algorithm that validates that implementations produce the distributions
defined in the types.

Probabilistic session types are naturally compositional. We can
define a $\m{negneg}$ process that calls $\m{neg}$ twice to obtain
an identity process.
\begin{lstlisting}
  decl negneg : (b : pbool) |- (d : pbool)
  proc d <- negneg b = c <- neg b ; d <- neg c
\end{lstlisting}
The term $c \leftarrow \m{neg} \; b$ corresponds to calling the $\m{neg}$
process passing $b$ as an input channel, and binding $c$ to the output
channel. We call $\m{neg}$ again using $d \leftarrow \m{neg} \; c$
and binding the return channel to $d$. Since $1-(1-p) = p$, we obtain
that the input and output types for $\m{negneg}$ are equal.

\paragraph{\textbf{Inference of Probabilities}}
If the input probabilities to a process are given, \lang{} internally
employs an LP solver to infer the output probabilities automatically.
For instance, recall the $\m{TF}$ process that outputs $\mb{true}$
with prob.\ $0.6$ and $\mb{false}$ otherwise.
We allow the programmer to define a \emph{starred} boolean type
as $\m{sbool} \triangleq \pichoice{\mb{true}^* : \one, \mb{false}^* : \one}$.
where $*$ denotes unknown prob. values that need to be inferred.
The programmer can then define
\begin{lstlisting}
  decl TF : . |- (b : sbool)
  proc b <- TF = flip 0.6 (H => b.true ; close b  |  T => b.false ; close b)
\end{lstlisting}
The type checker internally replaces $*$ with prob. variables, i.e.
substitutes $\m{sbool}$ with $\pichoice{\mb{true}^{p_1} : \one, \mb{false}^{p_2} : \one}$
with the constraint $p_1+p_2=1$.
Then, the typing rules of \lang{} are applied to the program, which
intuitively compute the prob.\ of outputting each label.
In the $\heads$ branch, the prob. of outputting $\mb{true}$ is $1$, while
in the $\tails$ branch, the prob. of outputting $\mb{true}$ is $0$.
Therefore, the total prob.\ of outputting $\mb{true}$ is $0.6 * 1 + (1-0.6) * 0$.
Similarly, the total prob.\ of outputting $\mb{false}$ is $0.6 * 0 + (1-0.6) * 1$.
Noting these observations, the type checker generates the following linear constraints
\begin{center}
  \begin{minipage}{0cm}
  \begin{tabbing}
  $p_1 = 0.6 * 1 + (1-0.6) * 0$ \qquad
  $p_2 = 0.6 * 0 + (1-0.6) * 1$ \qquad
  $p_1 + p_2 = 1$
  \end{tabbing}
  \end{minipage}
\end{center}
The LP solver then solves these constraints to produce the solution
$p_1 = 0.6, p_2 = 0.4$ which is then substituted back in the type annotations
for the programmer to view and verify.
In a similar fashion, \lang{} can infer the output probabilities for
the $\m{neg}$ and the $\m{unbias}$ processes if the input probabilities
are given.

\paragraph{\textbf{Application to Markov Chains}}
Probabilistic session types are adept for implementing
and verifying Markov chains. A practical application of
Markov chains are \emph{dice programs}~\cite{Knuth76} that use
a fair coin to model a die. The Markov chain for one such program
is described in Figure~\ref{fig:dice_chain}(a). For simplicity of
exposition, we consider a 3-faced die, although we have implemented
the complete 6-faced die program (see Section~\ref{sec:impl}).

The Markov chain initiates in state $1$, and transitions to states $2$
or $3$ with prob. $0.5$ each. In state $2$, with prob. $0.5$, the chain
outputs face 1 and with prob. $0.5$, it transitions to back to state $1$.
In state $3$, with prob. $0.5$, the chain outputs face 2 and with prob. $0.5$,
it outputs face 3.

We can prove the functional correctness of this die program using
probabilistic session types. To this end, we implement
the probabilistic program corresponding this Markov chain. First, we
need to define three different probabilistic types, one corresponding to
each state. We define $ T_1 \triangleq \pichoice{\mb{one}^{p_1} : \one,
\mb{two}^{p_2} : \one, \mb{three}^{p_3} : \one}$, $T_2 \triangleq
\pichoice{\mb{one}^{p_4} : \one, \mb{two}^{p_5} : \one, \mb{three}^{p_6} : \one}$
and $T_3 \triangleq \pichoice{\mb{one}^{p_7} : \one, \mb{two}^{p_8} : \one,
\mb{three}^{p_9} : \one}$.
Each of these types output $\mb{one}$, $\mb{two}$ or $\mb{three}$ with
different probabilities and terminate.

We define processes $P_i$ corresponding to each state $i$.
Each process $P_i$ offers type $T_i$. Figure~\ref{fig:dice_chain}(b)
outlines the declaration and definition of each process.
The process $P_1$ flips a coin with prob. $0.5$, and in the $\heads$ branch,
it calls $P_2$ (corresponding to transitioning to state $2$), and in the
$\tails$ branch, it calls $P_3$. The process $P_2$ outputs $\mb{one}$ in
the $\heads$ branch, and calls $P_1$ in the $\tails$ branch. Finally, process
$P_3$ outputs $\mb{two}$ in the $\heads$ branch and $\mb{three}$ in the $\tails$
branch. Since the Markov chain is mutually recursive, so are the processes
$P_1$ and $P_2$. This program exactly corresponds
to the Markov chain in Figure~\ref{fig:dice_chain}(a).

Since this Markov chain is mutually recursive, computing the conditional
prob.\ of sending each label from each state is challenging.
As a first illustration, consider process $P_3$.
The prob.\ of sending $\mb{one}$ for $P_3$ is $0.5 * 0 + (1-0.5) * 0 = 0$.
The prob.\ of sending $\mb{two}$ is $0.5 * 1 + (1-0.5) * 0 = 0.5$.
The prob.\ of sending $\mb{three}$ is $0.5 * 0 + (1-0.5) * 1 = 0.5$.

Next, consider the prob. of outputting $\mb{one}$ for process $P_1$.
In the $\heads$ branch, it calls process $P_2$ which outputs $\mb{one}$
with prob.\ $p_4$.
In the $\tails$ branch, it calls process $P_3$ which outputs $\mb{one}$
with prob.\ $p_7$.
Therfore, the total prob.\ that $P_1$ outputs $\mb{one}$ is
$0.5 * p_4 + (1-0.5) * p_7$.
We apply a similar argument for labels $\mb{two}$ and $\mb{three}$
to obtain the constraints
\begin{eqnarray*}
  & p_1 = 0.5 * p_4 + (1-0.5) * p_7 \\
  & p_2 = 0.5 * p_5 + (1-0.5) * p_8 \\
  & p_3 = 0.5 * p_6 + (1-0.5) * p_9 \\
  & p_1 + p_2 + p_3 = 1
\end{eqnarray*}
Using a similar argument for process $P_2$, we obtain the constraints
\begin{eqnarray*}
  & p_4 = 0.5 * 1 + (1-0.5) * p_1 \\
  & p_5 = 0.5 * 0 + (1-0.5) * p_2 \\
  & p_6 = 0.5 * 0 + (1-0.5) * p_3 \\
  & p_4 + p_5 + p_6 = 1
\end{eqnarray*}
Lines 1, 2, 3 equate the prob.\ of sending labels $\mb{one}$, $\mb{three}$
and $\mb{three}$ respecitvely. We use the LP solver to solve these constraints
and produce the following type annotations
\begin{eqnarray*}
  T_1 & \triangleq & \pichoice{\mb{one}^{1/3} : \one, \mb{two}^{1/3} : \one, \mb{three}^{1/3} : \one} \\
  T_2 & \triangleq & \pichoice{\mb{one}^{2/3} : \one, \mb{two}^{1/6} : \one, \mb{three}^{1/6} : \one} \\
  T_3 & \triangleq & \pichoice{\mb{one}^{0} : \one, \mb{two}^{1/2} : \one, \mb{three}^{1/2} : \one}
\end{eqnarray*}

\lang{} can automatically infer that state $i$ offers type $T_i$.
The programmer only needs to implement the program in Figure~\ref{fig:dice_chain}(b)
with $*$ annotations for types $T_1$, $T_2$, and $T_3$.
\lang{} infers the probabilities on each type automatically.
Moreover, the successful type-checking of this program indicates that
state $1$ truly outputs each of the labels with equal probability,
thus proving its functional correctness.

\begin{figure}
  \centering
  \includegraphics[width=\linewidth]{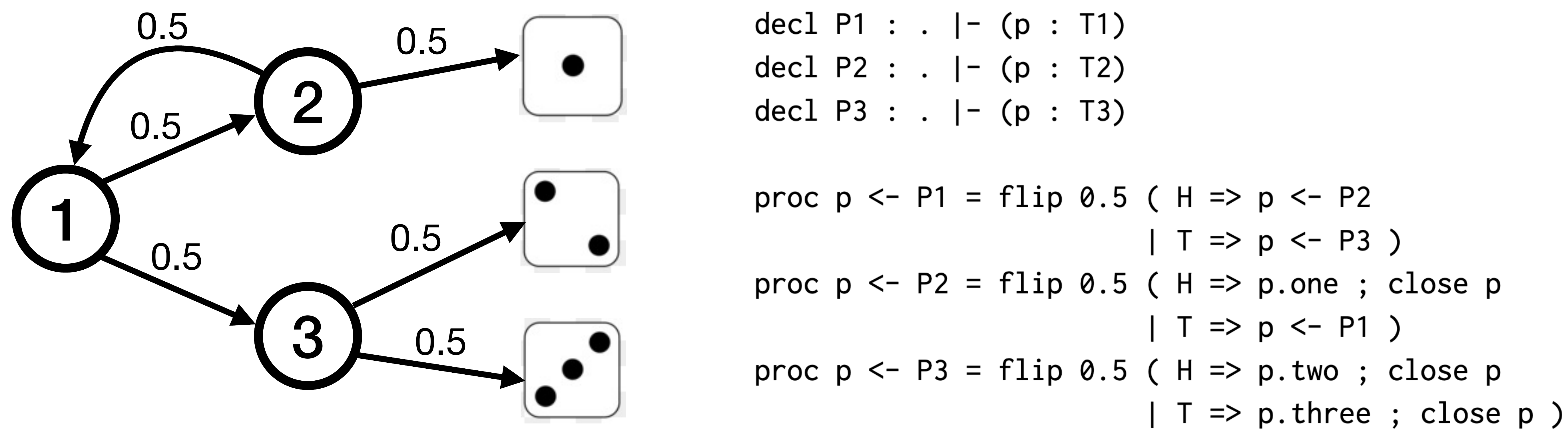}
  \vspace{-2em}
  \caption{(a) Markov chain associated with the 3-faced die program, and
  (b) corresponding program}
  \label{fig:dice_chain}
  \vspace{-1em}
  \Description{Dice program}
\end{figure}

\paragraph{\textbf{Automated Expected Cost Analysis}}

A natural application of probabilistic session types is the
type-guided analysis of the expected cost of distributed protocols and
Markov chains.
For instance, we would like automatically compute a bound on the
expected number of flips that are executed if we run the previously
defined process $\m{P_1}$, which is simulating a 3-faced die.

To perform the expected cost analysis we combine probabilistic session
types with existing techniques for automatic work analysis with
session types~\cite{Das18RAST}, which can be seen as an instantiation
of \emph{automatic amortized resource analysis}
(AARA)~\cite{Jost03,HoffmannW15}.
The idea is to statically associate a potential with each process
that is used to pay for the (expected) work that is performed
by this process.
A key feature is that this potential can also be transferred to other
processes to cover the work incurred by them, thus allowing amortization.
Importantly, the amount of potential transferred with a message or
associated with a process can be efficiently inferred by linear
constraint solving.
This technique is parametric in the cost model and can, for e.g.,
bound the expected number of messages or flips or other user-defined quantities.

We build intuition for the type system by revisiting the previously
discussed examples. Consider again the processes $\m{TF}$ and
$\m{neg}$ and assume a cost model in which the cost of sending the
label $\m{true}$ is $1$ and the cost of sending the label $\m{false}$
is $2$. Since $p = 0.6$, we can derive the typing
\begin{center}
  \begin{minipage}{0cm}
  \begin{tabbing}
  $\cdot \entailpot{1.4} \m{TF} :: (b : \m{bool})$
  \end{tabbing}
  \end{minipage}
\end{center}
where the number on the turnstile reflects the expected cost. Since
we incur cost $1$ with prob. $0.6$, and cost $2$ with prob. $0.4$,
we deduce that the expected cost of executing $\m{TF}$ is
$0.6 * 1 + 0.4 * 2 = 1.4$.
To infer the expected potential $q$ of a probabilistic expression $\eflip{p}{P_H}{P_T}$,
we use the second key idea expressed as the equation
\begin{center}
  \begin{minipage}{0cm}
  \begin{tabbing}
  $q = p \cdot q_H + (1-p) \cdot q_T$
  \end{tabbing}
  \end{minipage}
\end{center}
where $q_H$ and $q_T$ are the expected potentials for $P_H$ and $P_T$ respectively.

Reasoning about expected cost is not inextricably connected to
probabilistic session types and can be studied as an independent
problem.
However, there are interesting connections between the two concepts
that we explore.
For example, we are able to derive the following typing for the
process $\m{neg}$ (cost of $2$ for sending $\mb{false}$, and
$1$ for sending $\mb{true}$).
\begin{center}
  \begin{minipage}{0cm}
  \begin{tabbing}
  $(b : \m{pbool}) \entailpot{2 * 0.6 + (1{-}0.6)} \m{neg} :: (c : \m{npbool})$
  \end{tabbing}
  \end{minipage}
\end{center}
It states that the expected work cost of the process is bounded by
$2 * 0.6 + (1{-}0.6) = 1.6$. In the typing derivation, it is essential to have access
to the distribution of messages on the channel $b$. If this
information is not available then we have to assume the worst
case---the label $\mb{false}$ is sent on channel $c$---and derive
the bound $2$.
\begin{center}
  \begin{minipage}{0cm}
  \begin{tabbing}
  $(b : \m{bool}) \entailpot{2} \m{neg} :: (c : \m{bool})$
  \end{tabbing}
  \end{minipage}
\end{center}
Thus, probabilistic session types help us infer expected cost, instead
of worst-case cost.

For the 3-faced die program, we are interested in the expected cost
of the process $\m{P_1}$. For illustration purposes it is convenient
to consider a cost metric that counts the number of evaluated flips.
Then the expected cost of $\m{P1}$ is $\frac{8}{3}$. 
To infer this (tight) bound the type system assigns potential $q_i$ to process $P_i$.
Process $P_3$ requires only $1$ unit of
potential, since it performs \emph{only one} flip.
Taking the $p$-weighted sum of the expected cost in both branches,
we get
\begin{center}
  \begin{minipage}{0cm}
  \begin{tabbing}
  $q_1 = 1 + 0.5 \cdot q_2 + 0.5 \cdot q_3 = 1.5 + 0.5 \cdot q_2 \qquad
  q_2 = 1 + 0.5 \cdot 0 + 0.5 \cdot q_1 = 1 + 0.5 \cdot q_1$
  \end{tabbing}
  \end{minipage}
\end{center}
For each equation, the summand $1$ accounts for the flip at the start
of each process. For $q_1$, the potential in the $\heads$ branch is
$q_2$ since we call process $P_2$ and $q_3$ in the $\tails$ branch
since we call process $P_3$. For $q_2$, the potential in $\heads$
branch is $0$ since it does not involve any flips, and $q_1$ in the
$\tails$ branch since we call $P_1$. Solving these equations leads to
the solution $q_1 = \frac{8}{3}, q_2 = \frac{7}{3}, q_3 = 1$.

Our implementation (Section~\ref{sec:impl}) automatically generates
and solves these linear equations. 
In Section~\ref{sec:examples}, we show how we can automatically
derive expected cost bounds on randomized distributed protocols.
We can also derive symbolic bounds that depend on the numbers of
processes in the network by incorporating potential transfers in a
(recursive) session type as in previous work~\cite{Das18RAST} for
deterministic processes. 
For example, we infer that the expected number of messages in Lehmann
and Rabin's randomized dining philosophers algorithm~\cite{Lehmann81POPL} is
bounded by $6.25n$, where $n$ is the number of philosophers.


\section{\lang{} and Resource-Aware Session Types}\label{sec:formal}

\begin{figure}
  \centering
  \begin{tabbing}
  $\mbox{Proc} \;\; P, Q ::= \esendl{x}{k} \semi P \mid \ecase{x}{\ell}{P_\ell}_{\ell \in L}
  \mid \esendch{x}{y} \semi P \mid \erecvch{x}{y} \semi P_y \mid \eclose{x} \mid \ewait{x} \semi P$ \\
  \hspace{5em}\= $\mid \fwd{x}{y} \mid \ecut{y}{f}{\overline{x}}{P_y} \mid \eget{x}{r}
  \mid \epay{x}{r} \mid \ework{r} \semi P$ \\
  \> $\mid \color{blue}{\eflip{p}{P_H}{P_T}} \mid \color{blue}{\esendlp{x}{k} \semi P}
  \mid \color{blue}{\epcase{x}{\ell}{P_\ell}_{\ell \in L}}$ \\ \\
  $\mbox{Type} \;\; A, B ::= \ichoice{\ell : A_\ell}_{\ell \in L} \mid
  \echoice{\ell : A_\ell}_{\ell \in L}
  \mid A \tensor B \mid A \lolli B \mid \one \mid V \mid \tpaypot{A}{r} \mid \tgetpot{A}{r}$ \\
  \> $\mid \color{blue}{\pichoice{\ell^{r_\ell} : A_\ell}_{\ell \in L}} \mid
  \color{blue}{\pechoice{\ell^{r_\ell} : A_\ell}_{\ell \in L}}$
  \end{tabbing}
  \caption{Process expressions and session types of \lang{}}
  \label{fig:syntax}
  \Description{Grammar}
\end{figure}

This section describes the syntax and static semantics of \lang{}.
The types and expressions of \lang{} are defined by the grammars in Figure~\ref{fig:syntax}.
The symbol $\ell$ stands for a label (like in a sum type) and the symbols $x$ and $y$ stand
for variables, which range over channels.
The annotations $r$ and $r_\ell$ are non-negative rational numbers.
They denote potential annotations and probabilities.
The subscript $y$
on the process terms $P_y$
indicates that the variable $y$
is free in $P_y$ and bound in the respective syntactic form.
We focus the discussion on the (deterministic) fragment of the language,
which is similar to previous work on resource-aware session
types~\cite{Das18RAST}.
The novel probabilistic part of \lang{} is a
conservative extension and discussed in Section~\ref{sec:prob} (marked in
blue in Figure~\ref{fig:syntax}).

\paragraph{\textbf{Static Semantics}}

Type constructors in session types ($\oplus$,
$\with$,
$\tensor$,
$\lolli$,
$\one$)
are derived from assigning an operational interpretation to
connectives in intuitionistic linear logic. To express resource cost,
we use two type operators $\tpaypot{A}{r}$
and $\tgetpot{A}{r}$
to send and receive $r$
units of potential. Finally, we also have type identifiers $V$,
which can refer to type definitions to define recursive types.

The underlying base system of session types is derived from a Curry-Howard
interpretation~\cite{Caires10concur,Caires16mscs} of intuitionistic linear logic
\cite{Girard87tapsoft}. An intuitionistic linear sequent
$A_1, A_2, \ldots, A_n \vdash A$
is interpreted as the interface to a process expression $P$. We label each of the
antecedents with a channel name $x_i$ and the succedent with channel name $z$.
The $x_i$'s are \emph{channels used by} $P$ and $z$ is the \emph{channel provided by} $P$.
\begin{center}
  \begin{minipage}{0cm}
  \begin{tabbing}
  $x_1 : A_1, x_2 : A_2, \ldots, x_n : A_n \vdash P :: (z : C)$
  \end{tabbing}
  \end{minipage}
\end{center}
The resulting judgment states that process $P$ provides a service of
session type $C$ along channel $z$, while using the services of session types $A_1,
\ldots,A_n$ provided along channels $x_1, \ldots, x_n$, respectively. All these
channels must be distinct. We abbreviate the antecedent of the sequent by $\Delta$.

The typing judgment of \lang{} adds a non-negative rational number $q$
and a signature $\Sg$.
\begin{center}
  \begin{minipage}{0cm}
  \begin{tabbing}
  $\D \entailpot{q}_\Sg P :: (x : A)$
  \end{tabbing}
  \end{minipage}
\end{center}
The number $q$
is the potential of the process that can be used to cover (expected)
evaluation cost.
The signature $\Sg$
contains type and process definitions.  Because it is fixed, we elide
it from the presentation of the rules.


\begin{table*}[t]
\small
  \begin{tabular}{l l l l l}
  \textbf{Type} & \textbf{Cont.} & \textbf{Process Term} & \textbf{Cont.} & \multicolumn{1}{c}{\textbf{Description}} \\
  \toprule
  $c : \ichoice{\ell : A_\ell}_{\ell \in L}$ & $c : A_k$ & $\esendl{c}{k} \semi P$
  & $P$ & provider sends label $k$ along $c$ \\
  & & $\ecase{c}{\ell}{Q_\ell}_{\ell \in L}$ & $Q_k$ & client receives label $k$ along $c$ \\
  \addlinespace
  $c : \echoice{\ell : A_\ell}_{\ell \in L}$ & $c : A_k$ & $\ecase{c}{\ell}{P_\ell}_{\ell \in L}$
  & $P_k$ & provider receives label $k$ along $c$ \\
  & & $\esendl{c}{k} \semi Q$ & $Q$ & client sends label $k$ along $c$ \\
  \addlinespace
  $c : A \tensor B$ & $c : B$ & $\esendch{c}{w} \semi P$
  & $P$ & provider sends channel $w : A$ along $c$ \\
  & & $\erecvch{c}{y} \semi Q_y$ & $Q_y[w/y]$ & client receives channel $w : A$ along $c$ \\
  \addlinespace
  $c : A \lolli B$ & $c : B$ & $\erecvch{c}{y} \semi P_y$
  & $P_y[w/y]$ & provider receives channel $w : A$ along $c$ \\
  & & $\esendch{c}{w} \semi Q$ & $Q$ & client sends channel $w : A$ along $c$ \\
  \addlinespace
  $c : \one$ & --- & $\eclose{c}$
  & --- & provider sends $\mi{close}$ along $c$ \\
  & & $\ewait{c} \semi Q$ & $Q$ & client receives $\mi{close}$ along $c$ \\
  \addlinespace
  $c : \tpaypot{A}{r}$ & $c:A$ &   $\epay{c}{r}; P$
  & $P$ & provider sends potential $r$ along $c$ \\
  & & $\eget{c}{r};Q$ & $Q$ & client receives potential $r$ along $c$\\
  \addlinespace
  $c : \tgetpot{A}{r}$ & $c:A$ & $\eget{c}{r};P$ & $P$  & provider receives potential $r$ along $c$ \\
  & & $\epay{c}{r}; Q$
  & $Q$ & client sends potential $r$ along $c$\\
  \bottomrule
  \end{tabular}
  \caption{Deterministic session types with operational description}
  \label{tab:language}
  \vspace{-2em}
  \end{table*}
  

\paragraph{\textbf{Process and Type Definitions}}

Type definitions in the signature have the form $V = A$
and can be (mutually) recursive, departing from a strict Curry-Howard
interpretation of linear logic.
However, we require $A$ to be
\emph{contractive}~\cite{Gay2005} meaning $A$ should not itself be a
type name. Our type definitions are \emph{equirecursive} so we can
silently replace type names $V$ by $A$ during type checking,
and do not have explicit rules for recursive types.
Process definitions have the form $f = (\D,q,P,x,A)$,
where $f$ is the name of the process and $P$
its defining expression, with $\D$
being the channels used by $f$ and $x : A$ being the offered channel,
and $q$ its potential.
All definitions are collected in a fixed global signature $\Sg$.
For a \emph{well-formed signature}, we require that
$\D \entailpot{q}_{\Sg} P :: (x : A)$
for every process definition $f = (\D,q,P,x,A)$ in $\Sg$.
Like type definitions, process definitions are mutually
recursive.

\subsection{Basic Session Types}\label{subsec:base}

Table~\ref{tab:language} describes the deterministic session types,
their associated process expressions, their continuation (both in
types and expressions) and operational description.
Figure~\ref{fig:base-types} contains the type rules for basic session
types. The potential annotations $q$ present in these type rules are simply passed
around unchanged. They are discussed in detail in Section~\ref{subsec:ergo} .

\paragraph{\textbf{Internal and External Choice}}
The \emph{internal choice} type constructor
$\ichoice{\ell : A_{\ell}}_{\ell \in L}$ is an $n$-ary labeled
generalization of the additive disjunction $A \oplus B$. Operationally,
the provider of $x : \ichoice{\ell : A_{\ell}}_{\ell \in L}$ is required
to send a label $k \in L$ and then continue to provide $A_k$. The corresponding
process expression is $(\esendl{x}{k} \semi P)$ where $P$ is the
continuation. Dually, the client must branch based on the
label $k \in L$ received from the provider using the expression
$\ecase{x}{\ell}{Q}_{\ell \in L}$.
The corresponding typing rules are ${\oplus}R$ and ${\oplus}L$
in Figure~\ref{fig:base-types}. The process potential is unaffected and
will be equal in the premise and conclusion for all the structural rules.

The \emph{external choice} constructor
$\echoice{\ell : A_{\ell}}_{\ell \in L}$ is the dual of internal
choice requiring the provider to branch on one of the labels received
from the client (type rules ${\with}R$ and ${\with}L$). 
Dual constructors, like this one, reverse the role of the provider and client.

\begin{figure}
  \centering
\small
  \begin{mathpar}
\infer[{\with}R]
{\D \entailpot{q} \ecase{x}{\ell}{P_\ell}_{\ell \in L} ::
(x : \echoice{\ell : A_\ell}_{\ell \in L})}
{(\forall \ell \in L)
 & \D \entailpot{q} P_\ell :: (x : A_\ell)}
\and
\infer[{\with}L]
{\D, (x : \echoice{\ell : A_\ell}_{\ell \in L}) \entailpot{q}
(\esendl{x}{k} \semi Q) :: (z : C)}
{(k \in L) & \D, (x : A_k) \entailpot{q} Q :: (z : C)}
\and
  \infer[{\oplus}R]
    {\D \entailpot{q} (\esendl{x}{k} \semi P) :: (x : \ichoice{\ell : A_\ell}_{\ell \in L})}
    {(k \in L) & \D \entailpot{q} P :: (x : A_k)}
\and
  \infer[{\oplus}L]
    {\D, (x : \ichoice{\ell : A_\ell}_{\ell \in L}) \entailpot{q}
    \ecase{x}{\ell}{Q_\ell}_{\ell \in L} :: (z : C)}
    {(\forall \ell \in L) &
      \D, (x : A_\ell) \entailpot{q} Q_\ell :: (z : C)}
\and
  \infer[{\tensor}R]
    {\D, (y : A) \entailpot{q} (\esendch{x}{y} \semi P) :: (x : A \tensor B)}
    {\D \entailpot{q} P :: (x : B)}
  \and
  \infer[{\tensor}L]
    {\D, (x : A \tensor B) \entailpot{q} (\erecvch{x}{y} \semi Q) :: (z : C)}
    {\D, (y : A), (x : B) \entailpot{q} Q :: (z : C)}
\and
  \infer[{\lolli}R]
    {\D \entailpot{q} (\erecvch{x}{y} \semi P) :: (x : A \lolli B)}
    {\D, (y : A) \entailpot{q} P :: (x : B)}
  \and
  \infer[{\lolli}L]
    {\D, (x : A \lolli B), (y : A) \entailpot{q} (\esendch{x}{y} \semi Q) :: (z : C)}
    {\D, (x : B) \entailpot{q} Q :: (z : C)}
\and
  \infer[{\one}R]
    {\cdot \entailpot{q} (\eclose{x}) :: (x : \one)}
    {q = 0}
  \and
  \infer[{\one}L]
    {\D, (x : \one) \entailpot{q} (\ewait{x} \semi Q) :: (z : C)}
    {\D \entailpot{q} Q :: (z : C)}
\and
  \infer[\m{id}]
    {y : A \entailpot{q} (\fwd{x}{y}) :: (x : A)}
    {q = 0}
\and
  \inferrule*[right=$\m{spawn}$]
  {(\overline{y':B},p,P,x',A) \in \Sg \and
  r = p+q \and
  \D' = \overline{(y:B)} \and
  \D, (x : A) \entailpot{q} Q :: (z : C)}
  {\D, \D' \entailpot{r} (\ecut{x}{f}{\overline{y}}{Q}) :: (z : C)}
  \end{mathpar}
  \caption{Type Rules for Basic Session Types}
  \label{fig:base-types}
  \Description{Type Rules Basic}
\end{figure}

\paragraph{\textbf{Channel Passing}}
The \emph{tensor} operator $A \tensor B$ prescribes that the provider of
$x : A \tensor B$
sends a channel $y$ of type $A$ and continues to provide type $B$. The
corresponding process expression is $\esendch{x}{y} \semi P$ where $P$ is
the continuation.  Correspondingly, its client must receives a channel
using the expression $\erecvch{x}{y} \semi Q$, binding it to variable $y$
and continuing to execute $Q$. This if formalized by the rules ${\tensor}R$
and ${\tensor}L$.

The dual operator $A \lolli B$
(type rules ${\lolli}R$
and ${\lolli}L$)
allows the provider to receive a channel of type $A$
and continue to provide type $B$.
The client of $A \lolli B$,
on the other hand, sends the channel of type $A$
and continues to use $B$.

\paragraph{\textbf{Termination and Forwarding}}

The type $\one$, the multiplicative unit of linear logic,
indicates \emph{termination} requiring that the provider send a
\emph{close} message followed by terminating the communication.
In the corresponding type rules ${\one}R$ and ${\one}L$,
linearity enforces that the provider does not use any channels.
Since the potential is a linear quantity, we require that the
potential of a terminating process is $0$.

The rule $\m{id}$ handles forwarding of a channel.
A process $\fwd{x}{y}$ identifies the channels $x$ and $y$ so that any
further communication along either $x$ or $y$ will be along the unified
channel. Its typing rule corresponds to the logical rule of identity.
Since a forwarding process will terminate on interaction with
a message, we require its potential to be $0$ to preserve its linearity.

\paragraph{\textbf{Spawning and Tail Calls}}
A new instance of a defined process $f$ can be spawned with the expression
$\ecut{x}{f}{\overline{y}}{Q}$ (rule $\m{spawn}$) where $\overline{y}$ is a
sequence of channels matching the antecedents $\D'$.
The newly spawned process will use all variables in
$\overline{y}$ and provide $x$ to the continuation $Q$.
The declaration of $f$ is looked up in the signature $\Sg$ (first premise),
matching the types in $\D'$ and $\overline{y}$ (third premise). Similarly,
the freshly created channel $x$ has type $A$ from the signature.
The potential $r$ of the parent process must be equal to the sum of
the potential $p$ of the spawned process and $q$ of the continuation.
Sometimes a process invocation is a tail call,
written without a continuation as $\procdef{f}{\overline{y}}{x}$. This is a
short-hand for $\procdef{f}{\overline{y}}{x'} \semi \fwd{x}{x'}$ for a fresh
variable $x'$, that is, we create a fresh channel
and immediately identify it with x.


\begin{figure}
  \centering
\small
  \begin{mathpar}
  \infer[{\paypot}R]
  {\D \entailpot{q} \epay{x}{r} \semi P :: (x : \tpaypot{A}{r})}
  {q \geq r \and
  \D \entailpot{q-r} P :: (x : A)}
  \and
 \infer[{\getpot}R]
  {\D \entailpot{q} \eget{x}{r} \semi P :: (x : \tgetpot{A}{r})}
  {\D \entailpot{q+r} P :: (x : A)}
\and
  \infer[\m{work}]
  {\D \entailpot{q} \ework{r} \semi P :: (x : A)}
  {q \geq r \and
  \D \entailpot{q-r} P :: (x : A)}
\and
  \infer[\!\!{\paypot}L]
  {\D, (x : \tpaypot{A}{r}) \entailpot{q} \eget{x}{r} \semi Q :: (z : C)}
  {\D, (x : A) \entailpot{q+r} Q :: (z : C)}
\and
\hspace{-.35em}
  \infer[\!\!{\getpot}L]
  {\D, (x : \tgetpot{A}{r}) \entailpot{q} \epay{x}{r} \semi Q :: (z : C)}
  {q \geq r \and
  \D, (x : A) \entailpot{q-r} Q :: (z : C)}
\hspace{-.35em}
\and
  \infer[\!\!\m{weak}]
  {\D \entailpot{q} P :: (x : A)}
  {\D \entailpot{q+r} P :: (x : A)}
  \end{mathpar}
  \caption{Type Rules for Resource-Aware Session Types}
  \label{fig:res-types}
  \Description{Resource Type Rules}
\end{figure}

\subsection{Resource-Aware Types}\label{subsec:ergo}

To describe the resource contracts for inter-process communication,
the type system further supports amortized resource analysis
\cite{Tarjan85AARA}. The key idea is that \emph{processes store
potential} and \emph{messages carry potential}. This potential can
either be consumed to perform \emph{work} or exchanged using special
messages.
The resource-aware type rules are presented in Figure~\ref{fig:res-types}.

The type system provides the programmer with the flexibility to
specify what constitutes work. Thus, the programmer can choose to
count the resource they are interested in, and the type system
provides the corresponding upper bound. 
We use the expression $\ework{r} \semi P$ to define cost $r$.
In this article, we only consider monotone resource like time where $r \geq 0$.
The type rule $\m{work}$ requires that the potential $q$ is sufficient
to pay for the cost $r$ and the remaining potential $q-r \geq 0$.
Since the amount of potential consumed to type check this expression
is equal to the amount of work performed by it, the type safety theorem
expresses that the total work done by a system can never exceed its
initial potential.
Note that it is not necessary to expose the syntactic form $\ework{r}$
in the surface syntax of the language. It can be easily inserted
automatically to reflect a cost metric such as number of evaluation
steps.
For example,
to count the total number of messages sent, we insert $\ework{1}$
just before sending every message.

Two dual type constructors $\tpaypot{A}{r}$ and $\tgetpot{A}{r}$
are used to exchange potential. The provider of $x : \tpaypot{A}{r}$
must \emph{pay} $r$ units of potential along $x$ using process
expression $(\epay{x}{r} \semi P)$, and continue to provide $A$ by
executing $P$. These $r$ units are deducted from the potential
stored inside the sender. Dually, the client must receive the
$r$ units of potential using the expression $(\eget{x}{r} \semi Q)$
and add this to its internal stored potential. 
This is reflected in the type rules ${\paypot}R$ and ${\paypot}L$.
When sending potential, we ensure that the sender has sufficient potential to pay
($q \geq r$), which is then deducted from the internal process potential.
Dually, while gaining potential, it is added to the internal process potential.
The dual type $\tgetpot{A}{r}$ enables the provider to receive potential
that is sent by its client. This is formalized in the rules ${\getpot}R$ and ${\getpot}L$.

\paragraph{\textbf{Affine Potential}}

The previously-discussed rules treat potential as a linear
resource. As a result, the potential reflects the exact cost of
programs. However, we are most often interested in upper bounds on the
resource usage. For instance, if we treat potential linearly we cannot
type a process that has different work cost in different branches.

To treat potential in an affine way, we have to provide the ability to
throw away potential. This can be achieved by the rule $\m{weak}$
in Figure~\ref{fig:res-types}. Alternatively, we can drop the premise
$q=0$ from the rules ${\one}R$ and $\m{id}$ in Figure~\ref{fig:base-types}.



\begin{table*}[t]
\small
  \begin{tabular}{l l l l l}
  \textbf{Type} & \textbf{Cont.} & \textbf{Process Term} & \textbf{Cont.} \hspace{-1em} & \multicolumn{1}{c}{\textbf{Description}} \\
  \toprule
  $c : \pichoice{\ell^{p_\ell} : A_\ell}_{\ell \in L}$ & $c : A_k$ & $\esendlp{c}{k} \semi P$
  & $P$ & provider sends label $k$ on chan.\ $c$ ($p_k = 1$) \hspace{-.5em} \\
  & & $\epcase{c}{\ell}{Q_\ell}_{\ell \in L}$ & $Q_k$ & client receives distribution $\dist(\ell) = p_\ell$ on $c$ \\
  \addlinespace
  $c : \pechoice{\ell^{p_\ell} : A_\ell}_{\ell \in L}$ & $c : A_k$ & $\epcase{c}{\ell}{P_\ell}_{\ell \in L}$
  & $P_k$ & provider receives dist.\ $\dist(\ell) = p_\ell$ along $c$ \\
  & & $\esendlp{c}{k} \semi Q$ & $Q$ & client sends label $k$ along chan.\ $c$ ($p_k = 1$) \\
  \bottomrule
  \end{tabular}
  \caption{Probabilistic session types with operational description}
  \label{tab:prob-lang}
  \vspace{-2em}
  \end{table*}

\begin{figure}
  \centering
\small
  \begin{mathpar}
  \inferrule*[right=$\m{flip}$]
  {
    {\color{blue}
    \D = p \cdot \D_H +^{\m{L}} (1-p) \cdot \D_T} \\
    {\color{blue} A = p \cdot A_H +^{\m{R}} (1-p) \cdot A_T }\\\\
    q = p \cdot q_H + (1-p) \cdot q_T \\
    \D_H \entailpot{q_H} P_H :: (x : A_H) \\
    \D_T \entailpot{q_T} P_T :: (x : A_T)
  }
  {\D \entailpot{q} \eflip{p}{P_H}{P_T} :: (x : A)}
\and
  \inferrule*[right=${\pichoiceop}L$]
  {
    (\forall \ell \in L) \enskip \D_\ell, (x : A_\ell) \entailpot{q_\ell} Q_\ell :: (z : C_\ell) \\
    q = \textstyle \tsum_{\ell \in L} p_\ell  \cdot q_\ell \\
    \D = \textstyle \tsum^{\m{L}}_{\ell \in L} p_\ell \cdot \D_\ell \\
    C = \textstyle \tsum^{\m{R}}_{\ell \in L} p_\ell \cdot C_\ell
  }
  {\D, (x : \pichoice{\ell^{p_\ell} : A_\ell}_{\ell \in L}) \entailpot{q}
  \epcase{x}{\ell}{Q_\ell}_{\ell \in L} :: (z : C)}
\and
  \inferrule*[right=${\pichoiceop}R$]
  {
    p_k = 1 \and p_j = 0 \; (j \neq k) \\\\
    \D \entailpot{q} P :: (x : A_k)
  }
  {\D \entailpot{q} \esendlp{x}{k} \semi P :: (x : \pichoice{\ell^{p_\ell} :
  A_\ell}_{\ell \in L})}
\and
\inferrule*[right=${\pechoiceop}L$]
  {
    p_k = 1 \and p_j = 0 \; (j \neq k) \\\\
    \D,(x:A_k) \entailpot{q} P :: (z: C)
  }
  {\D,(x:\pechoice{\ell^{p_\ell}:A_\ell}_{\ell \in L}) \entailpot{q} \esendlp{x}{k} \semi P :: (z : C)}
\and
  \inferrule*[right=${\pechoiceop}R$]
  {
    (\forall \ell \in L) \enskip \D_\ell \entailpot{q_\ell} Q_\ell :: (x : A_\ell) \\
    q = \textstyle \tsum_{\ell \in L} p_\ell  \cdot q_\ell \\
    \D = \textstyle \tsum^{\m{L}}_{\ell \in L} p_\ell \cdot \D_\ell \\
  }
  {\D \entailpot{q}
  \epcase{x}{\ell}{Q_\ell}_{\ell \in L} :: (x : \pechoice{\ell^{p_\ell} : A_\ell}_{\ell \in L} )}
  \end{mathpar}
  \caption{Type Rules for Probabilistic Session Types}
  \label{fig:prob-types}
  \Description{Prob. Type Rules}
\end{figure}

\section{Probabilistic Session Types}\label{sec:prob}


In this section, we discuss the static semantics of the novel
probabilistic aspects of \lang{}. First, we introduce an expression
$\eflip{p}{P_H}{P_T}$
for probabilistic branching.
This expression together with the deterministic fragment of \lang{}
from Section~\ref{sec:formal} results in a probabilistic session-typed
language in which distributions are not reflected in the types.
We then add probabilistic choice types and syntactic forms (see
Table~\ref{tab:prob-lang}) for probabilistic send and receive, as well
as their interactions with probabilistic branching.
The additional type rules are given in Figure~\ref{fig:prob-types}.

\paragraph{\textbf{Probabilistic Flip}}
The expression $\eflip{p}{P_H}{P_T}$
operationally corresponds to flipping a coin with prob. $p$
(of outputting $\heads$, and $\tails$ otherwise) and executing $P_H$
if the coin flips to $\heads$ and executing $P_T$
otherwise. The corresponding typing rule is $\m{flip}$
in Figure~\ref{fig:prob-types}. In the fragment without probabilistic
choice types, we can ignore the probabilistic split of the types (blue
parts) of the rule and instead consider the following rule
$\m{simple{-}flip}$
that is a special case and identical to $\m{flip}$
in the deterministically-typed fragment of \lang{}.
\begin{mathpar}
  \inferrule*[right=$\m{simple{-}flip}$]
  {
    q = p \cdot q_H + (1-p) \cdot q_T \and
    \D \entailpot{q_H} P_H :: (x : A) \and
    \D \entailpot{q_T} P_T :: (x : A)
  }
  {\D \entailpot{q} \eflip{p}{P_H}{P_T} :: (x : A)}
\end{mathpar}
Both branches $P_H$
and $P_T$
of the probabilistic branching, are typed with the initial context
$\D$ and have to offer on the same channel $x$ of type $A$.
Notably, the probabilistic behavior of a process is not visible in its
type.
The interesting aspect of rule is the treatment of potential.
The initial potential $q$
is not identical to $q_H$
and $q_T$ (unlike rule ${\oplus}L$ in Figure~\ref{fig:base-types}).
Instead, $q$ is the weighted sum $p \cdot q_H + (1-p) \cdot q_T$
which corresponds to the \emph{expected potential} needed
to cover the probabilistic branch.
The rule $\m{simple{-}flip}$
can already be used in conjunction with the deterministic rules to
derive interesting and non-trivial bounds on the expected cost.

\paragraph{\textbf{Probabilistic Choices}}

Using probabilistic branching, processes can send labels according to
certain probability distribution.
%
For example, consider again the process $\m{TF}$
from Section~\ref{sec:overview}. The prob.\ distribution on $\m{true}$
and $\m{false}$
labels implemented by the process is $\dist{}$
where $\dist(\m{true}) = 0.6$ and $\dist(\m{false}) = 0.4$ (we use \verb|b..true|
to denote a probabilistic send of label \verb|true| on channel \verb|b|).
\begin{verbatim}
  proc b <- TF = flip 0.6 (H => b..true ; close b  |  T => b..false ; close b)
\end{verbatim}
To reflect this prob.\ distribution in the type of the channel $b$, we assign
$b : \pichoice{\mb{true}^{0.6} : \one, \mb{false}^{0.4} : \one}$.
In general, we introduce the type formers
\begin{center}
$\pichoice{\ell^{p_\ell} : A_\ell}_{\ell \in L}$  \hspace{3em} and \hspace{3em} $\pechoice{\ell^{p_\ell} : A_\ell}_{\ell \in L}$
\end{center}
for probabilistic internal and external choice. The types are similar
to their  deterministic versions but labels are annotated with
probabilities $p_\ell$.
In a well-formed type, we have $p_\ell \in [0,1]$
and $\sum_{\ell \in L} p_\ell = 1$.
The internal choice $\pichoice{\ell^{p_\ell} : A_\ell}_{\ell \in L}$ requires the
provider to send label $k \in L$ with probability $p_k$.

Receiving on a probabilistic channel can be seen as an external
version of a probabilistic branch. We again first consider a
simplified version of the probabilistic receive that is a special case
of the rule ${\pechoiceop}R$ in Figure~\ref{fig:prob-types}.
\begin{mathpar}
  \infer[\text{simple}{-}{\pechoiceop}R]
  {\D \entailpot{q}
  \epcase{x}{\ell}{Q_\ell}_{\ell \in L} :: (x : \pechoice{\ell^{p_\ell} : A_\ell}_{\ell \in L} )}
  {
    (\forall \ell \in L) \enskip \D \entailpot{q_\ell} Q_\ell :: (x : A_\ell) \qquad
    q = \tsum_{\ell \in L} p_\ell  \cdot q_\ell
  }
\end{mathpar}
Similar to the $\m{flip}$, the rule $\text{simple}{-}{\pechoiceop}R$
is similar to the rule ${\echoiceop}R$ but takes the weighted sum
$q = \tsum_{\ell \in L} p_\ell \cdot q_\ell$
as initial potential instead of the maximum.
The additional premises in the rule ${\echoiceop}R$ are
used to enable probabilistic sending of labels.

\paragraph{\textbf{Sending on a Probabilistic Channel}}

Validating the prob.\ distributions of processes that perform a
probabilistic send is one of the most interesting aspects of the type
system.
It would be possible to combine a probabilistic branching with the
probabilistic sending of the label in one atomic operation.
However, we are presenting a more flexible approach that \emph{decouples} the
sending of labels on probabilistic channels from
probabilistic branching.
The key idea is to alter the probabilities in the session types of the
channels in the context.

It is beneficial to first discuss the type rules ${\pichoiceop}R$
and ${\pechoiceop}L$ for probabilistic send.
The most notable feature of the rules is that we require that the
probability $p_k$ of the label $k$ that is sent along channel $x$ must be
$1$.
In general, we need to apply probabilistic branching to alter the
probabilities on the channel to arrive at such a trivial distribution.
For example, in the type derivation of the process $\m{TF}$,
the channel $b$
has type $\pichoice{\mb{true}^{1} : \one, \mb{false}^{0} : \one}$
in the $\heads$
branch of the flip and type
$\pichoice{\mb{true}^{0} : \one, \mb{false}^{1} : \one}$
in the $\tails$ branch of the flip.

In a probabilistic branching, we are using weighted sums of prob.\
distributions on labels in the similar way as we are using weighted
sums of potential annotations.
In our running example $\m{TF}$, the type derivation is sound because
of the following relation.
\begin{center}
$     \pichoice{\mb{true}^{0.6} : \one, \mb{false}^{0.4} : \one}  =
    0.6 \cdot \pichoice{\mb{true}^{1} : \one, \mb{false}^{0} : \one}
    +^{\m{R}}
    0.4 \cdot \pichoice{\mb{true}^{0} : \one, \mb{false}^{1} : \one}
 $
\end{center}
The two types on the right side of the equation are the types of the
channel $b$
in the branches of the flip. The probabilities $0.6$
and $0.4$ are the probabilities of the respective branches.
The operation $+^{\m{R}}$
combines the label probabilities $p$
point-wise. For example we have $0.6 = 0.6 \cdot 1 + 0.4 \cdot 0$.

Formally, we define two such weighted sum relations $+^{\m{R}}$ and 
$+^{\m{L}}$.
However, their intended effect is identical and they simply reflect
the duality of types in the consumed and offered channels.
They are used in the probabilistic branching rules ${\pichoiceop}L$,
${\pechoiceop}R$,
and $\m{flip}$
to allow different probability annotations in different branches.
Here, we generalize the notion of weighted sums to n-ary sums.
Note that the prob.\ in a probabilistic branching does not
uniquely determine the prob.\ distributions of the channels. For
instance, we have $A = p \cdot A +^{\m{R}} (1-p) \cdot A$ and
$A = p \cdot A +^{\m{L}} (1-p) \cdot A$ for any $p$.

\begin{figure}
  \centering
\small
  \begin{mathpar}
  \infer[\pichoiceop]
  {
    \pichoice{\ell^{p_\ell} : A_\ell}_{\ell \in L} =
    p \cdot \pichoice{\ell^{q_\ell} : A_\ell}_{\ell \in L} +^{\m{R}}
    (1-p) \cdot \pichoice{\ell^{r_\ell} : A_\ell}_{\ell \in L}}
  {
    (\forall \ell \in L) \and
    p_\ell = p \cdot q_\ell + (1-p) \cdot r_\ell
  }
  \and
  \infer[\pechoiceop]
  {
    \pechoice{\ell^{p_\ell} : A_\ell}_{\ell \in L} =
    p \cdot \pechoice{\ell^{q_\ell} : A_\ell}_{\ell \in L} +^{\m{L}}
    (1-p) \cdot \pechoice{\ell^{r_\ell} : A_\ell}_{\ell \in L}}
  {
    (\forall \ell \in L) \and
    p_\ell = p \cdot q_\ell + (1-p) \cdot r_\ell
  }
\and
  \infer[\ichoiceop]
  {
    \ichoice{\ell : A_\ell}_{\ell \in L} =
    p \cdot \ichoice{\ell : A_\ell}_{\ell \in L} +^{*}
    (1-p) \cdot \ichoice{\ell : A_\ell}_{\ell \in L}}
  {
  }
  \and
  \infer[\echoiceop]
  {
    \echoice{\ell : A_\ell}_{\ell \in L} =
    p \cdot \echoice{\ell : A_\ell}_{\ell \in L} +^{*}
    (1-p) \cdot \echoice{\ell : A_\ell}_{\ell \in L}}
  {
  }
  \and
  \infer[\tensor]
  {
    A \tensor B =
    p \cdot A \tensor B +^{*}
    (1-p) \cdot A \tensor B
  }
  {
  }
  \and
  \infer[\lolli]
  {
    A \lolli B =
    p \cdot A \lolli B +^{*}
    (1-p) \cdot A \lolli B
  }
  {
  }
  \and
  \infer[\one]
  {
    \one =
    p \cdot \one +^{*}
    (1-p) \cdot \one
  }
  {}
\\
  \infer[\m{emp}]
  {[] = p \cdot [] +^{\m{L}} (1-p) \cdot []}
  {}
\and
  \infer[\m{chan}]
  {
    \D, (x : A) =
    p \cdot \D_1, (x : B) +^{\m{L}}
    (1-p) \cdot \D_2, (x : C)
  }
  {
    A = p \cdot B +^{\m{L}} (1-p) \cdot C \and
    \D = p \cdot \D_1 +^{\m{L}} (1-p) \cdot \D_2
  }
  \end{mathpar}
  \caption{Inductive Definition of the Weighted Sum Relations}
  \label{fig:prob-split}
\end{figure}

\paragraph{\textbf{Weighted Sums of Session Types}}

The weighted sum relations $A = p \cdot A_1 +^{\m{R}} (1-p) \cdot A_2$
(resp.  $A = p \cdot A_1 +^{\m{L}} (1-p) \cdot A_2$)
for probabilistic session types $A$,
$A_1$,
and $A_2$ are inductively defined in Figure~\ref{fig:prob-split}.
The most interesting cases are the dual rules $\pichoiceop$
and $\pechoiceop$,
which are applied to channels with internal or external choices
that correspond to a probabilistic send.
In these rules we apply the weighted
$p_\ell = p \cdot q_\ell + (1-p) \cdot r_\ell$
to the outermost labels.
\emph{Note that these rules do not recursively apply weighted sums to
  continuation types $A_\ell$.}
Probability distributions at deeper level of the types (in $A_\ell$)
are not altered in different branches.
Similarly, other rules do not alter the distributions on the types.
So in all other cases $A = p \cdot A_1 +^{\m{R}} (1-p) \cdot A_2$
and $A = p \cdot A_1 +^{\m{L}} (1-p) \cdot A_2$ implies $A = A_1 = A_2$.
This is reflected by the other rules in Figure~\ref{fig:prob-split}.
The wildcard $*$ stands for both $\m{L}$ and $\m{R}$.
In the rules $\m{chan}$ and $\m{emp}$,
the relation $+^{\m{L}}$
is extended point-wise to contexts, using the same
notations. 

\paragraph{\textbf{Counterexamples for Deep Weighted Sums}}

At first, the shallow definitions of the weighted sums seem to be
overly restrictive. However, a closer examination shows that using
more general definitions of weighted sum of probabilities are in
general not compatible with the intended semantics of probabilistic
choice types. 
Consider for instance the following type. 
\begin{center}
$ B \triangleq
  \ichoice{\mb{true} : \pichoice{\mb{H}^{0.5} : \one, \mb{T}^{0.5} : \one}
, \mb{false} : \pichoice{\mb{H}^{0.5} : \one, \mb{T}^{0.5} : \one}
}$  
\end{center}
A process that offers on a channel of type $B$ should first send a
boolean and then provide a fair coin flip. However, if we would allow
a nested weighted sum of probabilities then the following undesirable
implementation would type check.
The problem is that a process that uses channel $b$ does not receive
a fair distribution after receiving the boolean label.
\begin{lstlisting}
  decl bad : . |- (b : $B$)
  proc b <- bad = flip 0.5 ( H => b.true ; b..H; close b 
                           | T => b.false ; b..T; close b )
\end{lstlisting}

Another potential use of randomness would be to not only change the
probabilities on internal choices but also on external choices on the
top level. However, this would lead to unsound behavior.
Consider for example the following process $\m{bad'}$. It offers a
channel $x : \pichoice{\mb{H}^{0.5} : \one, \mb{T}^{0.5} : \one}$ with
a fair probabilistic internal choice. However, we could justify the
following unsound typing if we allowed the type
$\pichoice{\mb{H}^{0.5} : \one, \mb{T}^{0.5} : \one}$ of $y$ to be
split as $\pichoice{\mb{H}^{1} : \one, \mb{T}^{0} : \one}$ and
$\pichoice{\mb{H}^{0} : \one, \mb{T}^{1} : \one}$ in the two branches
of the flip.
\begin{lstlisting}
  decl bad' : (y : $\pichoice{\mb{H}^{0.5} : \one, \mb{T}^{0.5} : \one}$) |- (x : $\pichoice{\mb{H}^{1} : \one, \mb{T}^{0} : \one}$)
  proc x <- bad' <- y = flip 0.5 ( H => pcase y ( H => x..H; x <-> y
                                                | T => x..T; x <-> y )
                                 | T => pcase y ( H => x..T; x <-> y
                                                | T => x..H; x <-> y ) )
\end{lstlisting}

%


\section{Meta Theory}
\label{sec:metatheory}

In this section, we formalize the meta-theory of \lang{} and proofs can be found in \cref{appendix:meta-theory}.
We first illustrate the difficulties for standard semantic constructions to integrate internal and external choices with probabilities in a concurrent system (\cref{subsec:first-attempt}).
Then we develop a novel probabilistic \emph{nested-multiverse} semantics that retains local distribution information to deal with the difficulties (\cref{subsec:nested-multiverse}).
Finally, we sketch our proof methodology for the soundness of \lang{} (\cref{subsec:soundness}).

\subsection{Difficulties for Standard Approaches}
\label{subsec:first-attempt}

Operationally, a state of a concurrent system is a \emph{configuration} of running processes in the system.
The probabilistic-flip expressions are the sole source of randomness in \lang{}.
In the literature, there are two canonical approaches to extending a non-probabilistic semantics with probabilities:
\begin{itemize}
  \item \emph{Trace-based}~\cite{ICFP:BLG16,ESOP:CP19}:
  The state-to-state transition relation is interpreted under
  a \emph{fixed} trace of random sources (e.g., a trace of coin flips),
  where the random constructs (e.g., coin flips) are resolved as deterministic readouts from the trace.
  In our setting, however, this approach lacks a mechanism for tracking the \emph{correlations} among traces,
  which have been shown to be important for expected-cost analysis~\cite{ICFP:WKH20}.
  With the presence of probabilistic internal- and external-choice types, it is also unclear how to guarantee type preservation for the transitions guided by a fixed trace.

  \item \emph{Distribution-based}~\cite{ICFP:BLG16,JCSS:Kozen81}:
  The state-to-state transition relation is lifted to a state-to-distribution transition relation, where the support of the distribution consists of successor states.
  In our setting, this approach amounts to updating operational rules to transit from configurations to distributions on configurations.
  However, this approach loses track of \emph{correlations} among configurations derived from different outcomes of probabilistic-flip expressions.
  As a consequence, to guarantee type preservation, one has to reason about the outcomes of a probabilistic-flip expression \emph{separately}, which might not be always possible.
\end{itemize}

Recently, \citet{InversoMP20} proved the soundness of a probabilistic session-type system \emph{without} standard internal and external choices, and channel passing, with respect to a distribution-based semantics.
Their idea is that to prove type preservation for a probabilistic-flip expression, one needs to first construct new type derivations for all possible configurations after the probabilistic flip, and then combine the types via a ``weighted sum'' with respect to the flip probability.
This approach is \emph{global}, in the sense that after a \emph{local} coin flip in a process, one has to re-analyze the \emph{whole} configuration with other unchanged processes.
However, this approach would fail if one attempts to add standard internal and external choices to the system.

\begin{example}[Difficulties for adding standard internal and external choices]
  Consider the program below:
\begin{lstlisting}
  decl P : . |- (y : $\pichoice{\mb{H}^{0.5} : \one, \mb{T}^{0.5} : \one}$)
  proc y <- P = flip 0.5 ( H => y..H; close y | T => y..T; close y )
  decl Q : (y : $\pichoice{\mb{H}^{0.5} : \one, \mb{T}^{0.5} : \one}$) (z : $\pechoice{\mb{H}^{0.5}:\one,\mb{T}^{0.5}:\one}$) |- (x : $\echoice{H:\one,T:\one}$)
  proc x <- Q y z =
    case x ( H => pcase y ( H => flip 0.6 ( H => z..H; ... | T => z..T; ... )
                          | T => flip 0.4 ( H => z..H; ... | T => z..T; ... ) )
           | T => pcase y ( H => flip 0.7 ( H => z..H; ... | T => z..T; ... )
                          | T => flip 0.3 ( H => z..H; ... | T => z..T; ... ) ) )
\end{lstlisting}
  
  Intuitively, there is a probabilistic choice (the case expression on $y$) \emph{inside} a standard choice (the case expression on $x$).
  If in the operational semantics, we make a step on the process $P$ first and the coin shows heads, then the type of the channel $y$ becomes $\pichoice{H^{1} : \one, T^{0} : \one}$.
  As a consequence, the types of the channel $z$ for the two probabilistic case expressions become
  $\pechoice{H^{0.6}: \one, T^{0.4}: \one}$ and $\pechoice{H^{0.7}:\one, T^{0.3}:\one}$, respectively.
  However, we cannot apply the rule (${\echoiceop}R$) to re-derive a type-judgment for this configuration, because the rule requires that both branches of the case expression on $x$ have the same type on the channel $z$.
\end{example}

\subsection{A Nested-Multiverse Semantics}
\label{subsec:nested-multiverse}

The example above suggests that \emph{we should make probabilistic flips as local as they could be}.
In other words, the influence of a coin flip should remain within the flipped process, \emph{until} the process tries to communicate with other processes.
With this intuition, we define the \emph{nested-multiverse} semantic objects to have either form
(i) $\proc{c}{w,P}$ for a process $P$ that is provided along channel $c$ and has performed $w$ units of work, or form
(ii) $\proc{c}{\{\calC_i : p_i\}_{i \in \calI}}$, for a distribution---resulted from flip expressions---of local \emph{configurations} that are provided along channel $c$, where each configuration is a set of semantic objects and $\tsum_{i \in \calI} p_i = 1$.
Intuitively, these semantic objects collect nested information from different ``universes'' (e.g., the results of coin flips) explicitly.
Formally, the syntax of semantic objects is defined as follows, with the understanding that in a configuration $\calO_1 \parallel \cdots \parallel \calO_n$, the sender is always put to the \emph{right} of the receiver:
\begin{align*}
  \calE & \Coloneqq (w,P) \mid \{ \calC_i : p_i \}_{i \in \calI} \\
  \calO & \Coloneqq \proc{c}{\calE} \\
  \calC & \Coloneqq \calO_1 \parallel \cdots \parallel \calO_n
\end{align*}

\cref{fig:pstl-semantics} presents the rules of the novel nested-multiverse semantics of \lang{}.
The rules should be understood as \emph{multiset-rewriting} rules~\cite{Cervesato09SEM}:
every rule mentions only the rewritten parts of a configuration.
We distinguish two kinds of evaluation rules:
\begin{itemize}
  \item \emph{Single-process} rules: The relation $\sstep$ has exactly \emph{one} semantic object on the left-hand-side.
  These rules do \emph{not} involve communication.
  
  \item \emph{Communication} rules: The relation $\cstep{d,\kappa}$ has exactly \emph{two} semantic objects on the left-hand-side.
  The relation is intended to describe a communication carried out on channel $d$ with \emph{sort} $\kappa \in \{ \m{det}, {\pichoiceop}, {\pechoiceop} \}$.
  The sort $\kappa$ is used to categorize the communication on channel $d$: The sorts $\pichoiceop$ and $\pechoiceop$ stand for probabilistic internal and external choices, respectively, while $\m{det}$ represents all other kinds of communication.
  Such information becomes useful when two communicating semantic objects are distribution objects.
  For example, in the rule (\textsc{C:BDist:L}), if the underlying communication is of sort $\pichoiceop$, i.e., the object $\proc{d}{\{ \calC_j' : p_j' \}_{j \in \calJ}}$ sends probabilistically on channel $d$, and $\proc{c}{\{ \calC_i : p_i \}_{i \in \calI}}$ receives on channel $d$.
  Intuitively, $\proc{d}{\{ \calC_j' : p_j' \}_{j \in \calJ}}$ has evaluated one or more probabilistic flips to obtain the local distribution, to achieve the goal of sending labels probabilistically on channel $d$.
  Thus, on the receiving side, these different senders in different ``universes'' should be considered as a whole, in order to justify the probabilities on the $\pichoiceop$-typed channel.
  Therefore, the rule (\textsc{C:BDist:L}) keeps $\proc{d}{\{ \calC_j' : p_j' \}_{j \in \calJ}}$ intact, but decomposes $\proc{c}{\{ \calC_i : p_i \}_{i \in \calI}}$.
\end{itemize}

\begin{figure}
\centering
\begin{footnotesize}
\begin{tabular}{ll}
  (\textsc{E:Def}) & $\proc{c}{w, \ecut{x}{f}{\many{d}}{Q}} \sstep \proc{c}{w,Q[b/x]} \parallel \proc{b}{0,P_f[b/x',\many{d}/\many{y'}]}$ \\
  & \quad for $\many{y' : B} \entailpot{q} f = P_f :: (x' : A) \in \Sg$ and $b$ \emph{fresh} \\
  (\textsc{E:Work}) & $\proc{c}{w,\ework{r} \semi P} \sstep \proc{c}{w+r,P}$ \\
  (\textsc{E:Flip}) & $\proc{c}{w, \eflip{p}{P_H}{P_T}} \sstep \proc{c}{ \{ \proc{c}{w,P_H} : p , \proc{c}{w,P_T} : 1-p \} }$ \\
  (\textsc{E:Dist}) & $\proc{c}{\{ \calC_i : p_i \}_{i \in \calI}} \sstep \proc{c}{\{ \calC_i : p_i \}_{i \in \calI \setminus \{i_0\}} \dplus \{ \calC_{i_0}' : p_{i_0} \}}$ \\
  & \quad for some $i_0 \in \calI$ such that $\calC_{i_0} \sstep \calC'_{i_0}$ \\
  \\
  (\textsc{C:$\ichoiceop$}) & $\proc{c}{w_c, \ecase{d}{\ell}{Q_\ell}_{\ell \in L}} \parallel \proc{d}{w_d, \esendl{d}{k} \semi P} \cstep{d,\m{det}} \proc{c}{w_c, Q_k} \parallel \proc{d}{w_d, P}$ \\
  (\textsc{C:$\echoiceop$}) & $\proc{c}{w_c, \esendl{d}{k} \semi Q } \parallel \proc{d}{w_d, \ecase{d}{\ell}{P_\ell}_{\ell \in L}} \cstep{d,\m{det}} \proc{c}{w_c, Q} \parallel \proc{d}{w_d,P_k}$ \\
  (\textsc{C:$\tensor$}) & $\proc{c}{w_c, \erecvch{d}{y} \semi Q } \parallel \proc{d}{w_d, \esendch{d}{e} \semi P} \cstep{d,\m{det}} \proc{c}{w_c, Q[e/y]} \parallel \proc{d}{w_d, P}$ \\
  (\textsc{C:$\lolli$}) & $\proc{c}{w_c, \esendch{d}{e} \semi Q} \parallel \proc{d}{w_d, \erecvch{d}{y} \semi P} \cstep{d,\m{det}} \proc{c}{w_c, Q} \parallel \proc{d}{w_d, P[e/y]}$ \\
  (\textsc{C:$\one$}) & $\proc{c}{w_c, \ewait{d} \semi Q} \parallel \proc{d}{w_d, \eclose{d}} \cstep{d,\m{det}} \proc{c}{w_c+w_d, Q}$ \\
  (\textsc{C:Id}) & $\proc{c}{w_c, Q\tuple{d}} \parallel \proc{d}{w_d, \fwd{d}{e}} \cstep{d,\kappa} \proc{c}{w_c+w_d, Q\tuple{d}[e/d]}$ \\
  & \quad for $Q\tuple{d} \cblocked{(d,\kappa)}$ \\
  (\textsc{C:$\paypot$}) & $\proc{c}{w_c, \eget{d}{r} \semi Q} \parallel \proc{d}{w_d, \epay{d}{r} \semi P} \cstep{d,\m{det}} \proc{c}{w_c, Q} \parallel \proc{d}{w_d,P}$ \\
  (\textsc{C:$\getpot$}) & $\proc{c}{w_c, \epay{d}{r} \semi Q} \parallel \proc{d}{w_d, \eget{d}{r} \semi P} \cstep{d,\m{det}} \proc{c}{w_c,Q} \parallel \proc{d}{w_d,P}$ \\
  (\textsc{C:$\pichoiceop$}) & $\proc{c}{w_c, \epcase{d}{\ell}{Q_\ell}_{\ell \in L}} \parallel \proc{d}{w_d, \esendlp{d}{k} \semi P} \cstep{d,\pichoiceop} \proc{c}{w_c, Q_k} \parallel \proc{d}{w_d, P}$ \\
  (\textsc{C:$\pechoiceop$}) & $\proc{c}{w_c, \esendlp{d}{k} \semi Q } \parallel \proc{d}{w_d, \epcase{d}{\ell}{P_\ell}_{\ell \in L}} \cstep{d,\pechoiceop} \proc{c}{w_c, Q} \parallel \proc{d}{w_d,P_k}$ \\
  \\
  (\textsc{C:Dist}) & $\proc{c}{\{ \calC_i : p_i \}_{i \in \calI}} \cstep{d,\kappa} \proc{c}{\{ \calC_i : p_i \}_{i \in \calI \setminus \{i_0\}} \dplus \{ \calC_{i_0}' : p_{i_0} \}}$ \\
  & \quad for some $i_0 \in \calI$ such that $\calC_{i_0} \cstep{d,\kappa} \calC_{i_0}'$ \\
  (\textsc{C:SDist:R}) & $\proc{c}{w,Q} \parallel \proc{d}{\{ \calC_i : p_i \}_{i \in \calI}} \cstep{d,\kappa} \calC'$ \\
  & \quad for $\proc{c}{\{ ( \proc{c}{w,Q} \parallel \calC_i ): p_i \}_{i \in \calI}} \cstep{d,\kappa} \calC'$ \\
  (\textsc{C:SDist:L}) & $\proc{c}{\{ \calC_i : p_i \}_{i \in \calI}} \parallel \proc{d}{w,P} \cstep{d,\kappa} \calC'$ \\
  & \quad for $\proc{c}{\{ (\calC_i \parallel \proc{d}{w,P}) : p_i \}_{i \in \calI}} \cstep{d,\kappa} \calC'$ \\
  (\textsc{C:BDist:D}) & $\proc{c}{\{ \calC_i : p_i \}_{i \in \calI}} \parallel \proc{d}{\{ \calC'_j : p'_j \}_{j \in \calJ}} \cstep{d,\m{det}} \calC''$ \\
  & \quad for $\proc{c}{\{ (\calC_i \parallel \calC'_j) : p_i \cdot p'_j \}_{i \in \calI,j \in \calJ}} \cstep{d,\m{det}} \calC''$ \\
  (\textsc{C:BDist:R}) & $\proc{c}{\{ \calC_i : p_i \}_{i \in \calI}} \parallel \proc{d}{\{ \calC'_j : p'_j \}_{j \in \calJ}} \cstep{d,\pechoiceop} \calC''$ \\
  & \quad for $\proc{c}{\{ (  \proc{c}{\{ \calC_i : p_i\}_{i \in \calI}} \parallel \calC_j') : p'_j  \}_{j \in \calJ}} \cstep{d,\pechoiceop} \calC''$ \\
  (\textsc{C:BDist:L}) & $\proc{c}{\{ \calC_i : p_i \}_{i \in \calI}} \parallel \proc{d}{\{ \calC'_j : p'_j \}_{j \in \calJ}} \cstep{d,\pichoiceop} \calC''$ \\
  & \quad for $\proc{c}{\{ ( \calC_i \parallel \proc{d}{\{ \calC'_j : p'_j \}_{j \in \calJ}} ) : p_i \}_{i \in \calI}} \cstep{d,\pichoiceop} \calC''$
\end{tabular}
\end{footnotesize}
\caption{Rules for the multiset-rewriting small-step nested-multiverse semantics of \lang{}.}
\label{fig:pstl-semantics}
\end{figure}

The rule (\textsc{E:Flip}) deals with probabilistic flips.
To evaluate a coin flip $\eflip{p}{P_H}{P_T}$ while keeping the randomness local in the flipped process, this rule creates a local distribution object, whose support contains two process objects: $P_H$ with probability $p$ and $P_T$ with probability $(1-p)$.
Note that this rule does not change the work counter of the flipped process.
If the flipped process ($P_H$ or $P_T$) can further evaluate probabilistic-flip expressions, the configuration will become nested naturally.

There are two non-communication deterministic rules.
In the rule (\textsc{E:Work}) for work tracking, we simply increment the work counter of the process.
The other rule (\textsc{E:Def}) describes process spawning.
Let $f$ be a defined process.
The spawning expression $\ecut{x}{f}{\many{d}}{Q}$ creates a new process providing \emph{fresh} channel $b$, which is constructed from the definition $P_f$ of $f$ with proper renaming.
The work counter of the newly created process is set to zero.

In the semantics, communication is \emph{synchronous}: processes sending a message pause their evaluation until the message is received.
In the rule (\textsc{C:$\ichoiceop$}) for internal choices, the provider $\esendl{d}{k} \semi P$ sends a message $k$ along the $d$ and continues as $P$, while the client receives the message $k$ and selects branch $k$.
The rule (\textsc{C:$\pichoiceop$}) for probabilistic internal choices is almost the same as (\textsc{S:$\ichoiceop$}), except that it marks the communication with $\pichoiceop$.
In the rule (\textsc{C:$\tensor$}) for sending channels, the provider $\esendch{d}{e} \semi P$ sends the channel $e$ and continues as $P$, while the client receives the channel $e$ and substitutes $e$ for $y$ in the continuation $Q$.
In the rule (\textsc{C:$\paypot$}) for paying potential, the provider of $d$ pays $r$ units of potential along $d$ to the client.
The dual rules (\textsc{C:$\echoiceop$}), (\textsc{C:$\pechoiceop$}), (\textsc{C:$\lolli$}), and (\textsc{S:$\getpot$}) simply reverse the role of the provider and the client in (\textsc{C:$\ichoiceop$}), (\textsc{C:$\pichoiceop$}), (\textsc{C:$\tensor$}), and (\textsc{C:$\paypot$}), respectively.
Note that these eight rules do not change the work counters of the processes.

The rule (\textsc{C:$\one$}) accounts for termination, i.e., the client simply waits the provider to terminate.
A similar rule (\textsc{S:Id}) deals with forwarding. Operationally, a process $\fwd{d}{e}$ \emph{forwards} any message that arrives on $e$ to $d$ and vice-versa.
There is an extra side condition $Q\tuple{d} \cblocked{(d,\kappa)}$, which means that the expression $Q\tuple{d}$ is communicating along channel $d$ whose sort is $\kappa$.
Note that we write $Q\tuple{d}$ to indicate that the channel $d$ must occur freely in the process $Q$.
We will formulate the definition of blocked semantic objects later in \cref{subsec:soundness}.
For both termination and forwarding, the work performed by the provider is absorbed by the client.

\subsection{Type Soundness}
\label{subsec:soundness}

In this section, we prove the soundness of \lang{} with respect to the nested-multiverse semantics.
Different from the system presented by~\citet{InversoMP20}, \lang{} is a conservative extension of the resource-aware session types~\cite{Das18RAST}.

\paragraph{Configuration typing}
The type rules for configurations are given below.
The judgment $\Dl \potconf{q} \calC :: \Gm$ means that the configuration $\calC$ uses the channels in the context $\Dl$ and provides the channels in $\Gm$, and the nonnegative number $q$ denotes the \emph{expected} value of the sum of the total potential and work done by the system.
The rule (\textsc{T:Compose}) imposes an order on linear configurations and flattens the linear object tree in a way that for any semantic object the providers of the channels used by the object are to the \emph{right} of the object in the configuration.
\begin{mathpar}\footnotesize
  \Rule{T:Proc}
  { \Dl \entailpot{q} P :: (c : A)
  }
  { \Dl \potconf{q+w} \proc{c}{w,P} :: (c : A)  }
  \and
  \Rule{T:Dist}
  { \Forall{i \in \calI} \Dl_i \potconf{q_i} \calC_i :: (c : A_i) \\\\
    \tsum^{\m{L}}_{i \in \calI} p_i \cdot \Dl_i = \Dl \\
    \tsum^{\m{R}}_{i \in \calI} p_i \cdot A_i = A \\
    \tsum_{i \in \calI} p_i \cdot q_i = q
  }
  { \Dl \potconf{q} \proc{c}{\{ \calC_i : p_i \}_{i \in \calI}} :: (c : A) }
  \and
  \Rule{T:Compose}
  { \Dl_1, \Dl' \potconf{q_1} \calO :: (c : A) \\
    \Dl_2 \potconf{q_2} \calC :: (\Dl, \Dl')
  }
  {  \Dl_1, \Dl_2 \potconf{q_1+q_2} (\calO \parallel \calC) :: (\Dl, (c : A)) }
\end{mathpar}

\paragraph{Type preservation and global progress}
We prove the preservation theorem below by induction on the nested-multiverse semantics.

\begin{theorem}\label{the:preservation}
  Suppose that $\Dl \potconf{q} \calC :: \Gm$.
  If $\calC \sstep \calC'$ or $\calC \cstep{d,\kappa} \calC'$ for some $d,\kappa$, then $\Dl \potconf{q} \calC' :: \Gm$.
\end{theorem}

The global progress turns out to be more complex because we have to discover possible communication in a nested configuration.
To aid the proof of global progress, we define several relations to characterize the \emph{status} of a semantic object:
\begin{itemize}
  \item $\calC \poised$ means that \emph{all} process in $\calC$ are communicating along their providing channels. \cref{fig:poised} lists the rules.
  
  \item $\calC \live$ means that there exists a process in $\calC$ that can make a step \emph{without} communication. \cref{fig:live} lists the rules.
  
  \item $\calC \cpoised{(d,\kappa)}$ means that \emph{some} process in $\calC$ is communicating along its providing channel $d$ of sort $\kappa$, and channel $d$ is external to $\calC$. \cref{fig:cpoised} lists the rules. Note that we introduce a new sort $\top$ as the \emph{super}-sort of $\m{det},{\pichoiceop},{\pechoiceop}$.
  
  \item $\calC \cblocked{(d,\kappa)}$ means that \emph{some} process in $\calC$ is communicating along its consumed channel $d$ of sort $\kappa$, and channel $d$ is external to $\calC$. \cref{fig:cblocked} lists the rules.
  
  \item $\calC \comm{(d,\kappa)}$ means that there exist two processes that are going to communicate in $\calC$ along channel $d$ of sort $\kappa$. \cref{fig:comm} lists the rules. Note that in the rule (\textsc{CM:Compose:C}), we use $\kappa <: \kappa'$ to handle the case where $\kappa' = \top$.
\end{itemize}

\begin{figure}
\begin{mathpar}\footnotesize
  \Rule{P:$\one R$}
  {
  }
  { \proc{c}{w, \eclose{c}} \poised }
  \and
  \Rule{P:${\ichoiceop}R$}
  {
  }
  { \proc{c}{w, \esendl{c}{k} \semi P} \poised }
  \and
  \Rule{P:${\echoiceop}R$}
  {
  }
  { \proc{c}{w, \ecase{c}{\ell}{P_\ell}_{\ell \in L}} \poised }
  \and
  \Rule{P:${\pichoiceop}R$}
  {
  }
  { \proc{c}{w, \esendlp{c}{k} \semi P} \poised }
  \and
  \Rule{P:${\pechoiceop}R$}
  {
  }
  { \proc{c}{w, \epcase{c}{\ell}{P_\ell}_{\ell \in L}} \poised }
  \and
  \Rule{P:${\paypot}R$}
  {
  }
  { \proc{c}{w, \epay{c}{r} \semi P} \poised }
  \and
  \Rule{P:${\getpot}R$}
  {
  }
  { \proc{c}{w, \eget{c}{r} \semi P} \poised }
  \and
  \Rule{P:${\tensor}R$}
  {
  }
  { \proc{c}{w, \esendch{c}{e} \semi P } \poised }
  \and
  \Rule{P:${\lolli}R$}
  {
  }
  { \proc{c}{w, \erecvch{c}{y} \semi P} \poised }
  \and
  \Rule{P:IdR}
  {
  }
  { \proc{c}{w, \fwd{c}{e}} \poised }
  \and
  \Rule{P:Dist}
  { \Forall{i \in \calI} \calC_i \poised
  }
  { \proc{c}{\{ \calC_i : p_i \}_{i \in \calI}} \poised }
  \and
  \Rule{P:Compose}
  { \calO \poised \\
    \calC  \poised
  }
  { (\calO \parallel \calC) \poised }
\end{mathpar}
\caption{Rules for poised semantic objects.}
\label{fig:poised}
\end{figure}

\begin{figure}
\begin{mathpar}\footnotesize
  \Rule{L:Flip}
  {
  }
  { \proc{c}{w, \eflip{p}{P_H}{P_T}} \live }
  \and
  \Rule{L:Work}
  {
  }
  { \proc{c}{w, \ework{r} \semi P} \live }
  \and
  \Rule{L:Def}
  {
  }
  { \proc{c}{w, \ecut{x}{f}{\many{d}}{Q}} \live }
  \and
  \Rule{L:Dist}
  { \calC_{i_0} \live
  }
  { \proc{c}{\{ \calC_i : p_i \}_{i \in \calI}} \live }
  \and
  \Rule{L:Compose:H}
  { \calO\live
  }
  { (\calO \parallel \calC) \live }
  \and
  \Rule{L:Compose:T}
  { \calC\live
  }
  { (\calO \parallel \calC) \live }
\end{mathpar}
\caption{Rules for live semantic objects.}
\label{fig:live}
\end{figure}

\begin{figure}
\begin{mathpar}\footnotesize
  \Rule{PR:${\one}$}
  {
  }
  {  \proc{c}{w,\eclose{c}} \cpoised{(c,\m{det})} }
  \and
  \Rule{PR:${\ichoiceop}$}
  {
  }
  { \proc{c}{w, \esendl{c}{k} \semi P} \cpoised{(c,\m{det})} }
  \and
  \Rule{PR:${\echoiceop}$}
  {
  }
  { \proc{c}{w, \ecase{c}{\ell}{P_\ell}_{\ell \in L} \cpoised{(c,\m{det})}} }
  \and
  \Rule{PR:${\pichoiceop}$}
  {
  }
  { \proc{c}{w, \esendlp{c}{k} \semi P} \cpoised{(c,\pichoiceop)} }
  \and
  \Rule{PR:${\pechoiceop}$}
  {
  }
  { \proc{c}{w, \epcase{c}{\ell}{P_\ell}_{\ell \in L} \cpoised{(c,\pechoiceop)}} }
  \and
  \Rule{PR:${\paypot}$}
  {
  }
  { \proc{c}{w, \epay{c}{r} \semi P} \cpoised{(c,\m{det})} }
  \and
  \Rule{PR:${\getpot}$}
  {
  }
  { \proc{c}{w, \eget{c}{r} \semi P} \cpoised{(c,\m{det})} }
  \and
  \Rule{PR:$\tensor$}
  {
  }
  { \proc{c}{w, \esendch{c}{e} \semi P} \cpoised{(c,\m{det})} }
  \and
  \Rule{PR:$\lolli$}
  {
  }
  { \proc{c}{w, \erecvch{c}{y} \semi P} \cpoised{(c,\m{det})} }
  \and
  \Rule{PR:Id}
  {
  }
  { \proc{c}{w, \fwd{c}{e} } \cpoised{(c, \top)} }
  \and 
  \Rule{PR:Dist}
  { \calC_{i_0} \cpoised{(d,\kappa)}
  }
  { \proc{e}{\{ \calC_i : p_i \}_{i \in \calI} } \cpoised{(d,\kappa)} }
  \and
  \Rule{PR:Compose:H}
  { \calO \cpoised{(d,\kappa)}
  }
  { (\calO \parallel \calC) \cpoised{(d,\kappa)} }
  \and
  \Rule{PR:Compose:T}
  { \calC \cpoised{(d,\kappa)} \\
    d \not\in \mathrm{FV}^{\m{L}}(\calO)
  }
  { (\calO \parallel \calC) \cpoised{(d,\kappa)} }
\end{mathpar}
\caption{Rules for $(d,\kappa)$-poised semantic objects.}
\label{fig:cpoised}
\end{figure}

\begin{figure}
\begin{mathpar}\footnotesize
  \Rule{BL:${\one}$}
  {
  }
  { \proc{d}{\ewait{c} \semi Q} \cblocked{(c,\m{det})} }
  \and
  \Rule{BL:$\ichoiceop$}
  {
  }
  { \proc{d}{\ecase{c}{\ell}{Q_\ell}_{\ell \in L}} \cblocked{(c,\m{det})} }
  \and
  \Rule{BL:$\echoiceop$}
  {
  }
  { \proc{d}{\esendl{c}{k} \semi Q} \cblocked{(c,\m{det})} }
  \and
  \Rule{BL:$\pichoiceop$}
  {
  }
  { \proc{d}{\epcase{c}{\ell}{Q_\ell}_{\ell \in L}} \cblocked{(c,\pichoiceop)} }
  \and
  \Rule{BL:$\pechoiceop$}
  {
  }
  { \proc{d}{\esendlp{c}{k} \semi Q} \cblocked{(c,\pechoiceop)} }
  \and
  \Rule{BL:$\paypot$}
  {
  }
  { \proc{d}{\eget{c}{r} \semi Q} \cblocked{(c,\m{det})} }
  \and
  \Rule{BL:$\getpot$}
  {
  }
  { \proc{d}{\epay{c}{r} \semi Q} \cblocked{(c,\m{det})} }
  \and
  \Rule{BL:$\tensor$}
  {
  }
  { \proc{d}{\erecvch{c}{y} \semi Q} \cblocked{(c,\m{det})} }
  \and
  \Rule{BL:$\lolli$}
  {
  }
  { \proc{d}{\esendch{c}{e} \semi Q} \cblocked{(c,\m{det})} }
  \and
  \Rule{BL:Dist}
  { \calC_{i_0} \cblocked{(d,\kappa)}
  }
  { \proc{e}{\{ \calC_i : p_i \}_{i \in \calI}} \cblocked{(d,\kappa)} }
  \and
  \Rule{BL:Compose:H}
  { \calO \cblocked{(d,\kappa)} \\
    d \not\in \mathrm{FV}^{\m{R}}(\calC)
  }
  { (\calO \parallel \calC) \cblocked{(d,\kappa)} }
  \and
  \Rule{BL:Compose:T}
  { \calC \cblocked{(d,\kappa)}
  }
  { (\calO \parallel \calC) \cblocked{(d,\kappa)} }
\end{mathpar}
\caption{Rules for $(d,\kappa)$-blocked semantic objects.}
\label{fig:cblocked}
\end{figure}

\begin{figure}
\begin{mathpar}\footnotesize
  \Rule{CM:Dist}
  { \calC_{i_0} \comm{(d,\kappa)}
  }
  { \proc{e}{\{ \calC_i : p_i\}_{i \in \calI}} \comm{(d,\kappa)} }
  \and
  \Rule{CM:Compose:H}
  { \calO \comm{(d,\kappa)}
  }
  { (\calO \parallel \calC) \comm{(d,\kappa)} }
  \and
  \Rule{CM:Compose:T}
  { \calC \comm{(d,\kappa)}
  }
  { (\calO \parallel \calC) \comm{(d,\kappa)} }
  \and
  \Rule{CM:Compose:C}
  { \calO \cblocked{(d,\kappa)} \\
    \calC \cpoised{(d,\kappa')} \\
    \kappa <: \kappa'
  }
  { (\calO \parallel \calC) \comm{(d,\kappa)} }
\end{mathpar}
\caption{Rules for semantic objects with possible communication inside.}
\label{fig:comm}
\end{figure}

We first prove that these status characterizations introduced above are sufficient conditions for a configuration to make a step in the nested-multiverse semantics.

\begin{lemma}\label{lem:goodstatus}\
  \begin{itemize}
    \item If $\calC \live$, then there exists $\calC'$ such that $\calC \sstep \calC'$.
    \item If $(\cdot) \potconf{q} \calC :: \Gm$ and $\calC \comm{(d,\kappa)}$,
  then $C \cstep{d,\kappa} \calC'$ for some $\calC'$.
  \end{itemize}
\end{lemma}

Then we prove that for a well-typed configuration, we can always find a suitable status characterization for it.

\begin{lemma}\label{lem:progress}
  If $(\cdot) \potconf{q} \calC :: \Gm$, then at least one of the cases below holds:
  \begin{enumerate}[(i)]
    \item $\calC \live$,
    \item $\calC \comm{(d,\kappa)}$ for some $d,\kappa$, or
    \item $\calC \poised$.
  \end{enumerate}
\end{lemma}

Finally, we can prove global progress of our type system.

\begin{theorem}
  If $ (\cdot) \potconf{q} \calC :: \Gm$, then either
  \begin{enumerate}[(i)]
    \item $\calC \sstep \calC'$ for some $\calC'$, or $\calC \cstep{d,\kappa} \calC'$ for some $\calC',d,\kappa$, or
    \item $\calC \poised$.
  \end{enumerate}
\end{theorem}
\begin{proof}
  Appeal to \cref{lem:goodstatus,lem:progress}.
\end{proof}

\paragraph{Expected work analysis}
To reason about the \emph{expected} amount of work done by a probabilistic system, we harness Markov-chain-based reasoning~\cite{ESOP:KKM16,LICS:OKK16} to construct a stochastic process of system states.
To construct the Markov chain, we have to ``flatten'' a nested-multiverse configuration to a \emph{distribution} on non-nested configurations, i.e., sequences of processes $\calD \Coloneqq \proc{c_1}{w_1,P_1} \parallel \cdots \parallel \proc{c_n}{w_n,P_n}$.
We formalize the ``flattening'' procedure via a \emph{simulation} relation showed below.
\begin{mathpar}\footnotesize
  \Rule{FL:Proc}
  { 
  }
  {  \proc{c}{w,P} \approx \{ \proc{c}{w,P} : 1 \} }
  \and
  \Rule{FL:Dist}
  { \Forall{i \in \calI} \calC_i \approx \mu_i
  }
  { \proc{c}{\{ \calC_i : p_i \}_{i \in \calI}} \approx \tsum_{i \in \calI} p_i \cdot \mu_i  }
  \and
  \Rule{FL:Compose}
  { \calO \approx \mu_1 \\
    \calC \approx \mu_2
  }
  { (\calO \parallel \calC) \approx \{ (\calD_1 \parallel \calD_2 : \mu_1(\calD_1) \cdot \mu_2(\calD_2) \}_{\calD_1 \in \dom{\mu_1}, \calD_2 \in \dom{\mu_2}} }
\end{mathpar}

We then define the expected total work with respect to the Markov-chain semantics and prove that our system derives a sound \emph{upper} bound on the expected work.
We denote the total work done by a non-nested configuration $\calD = \many{\proc{c_i}{w_i,P_i}}$ by $\m{work}(\calD) \defeq \tsum_i w_i$.

\begin{theorem}
  Suppose that $(\cdot) \potconf{q} \calC :: \Gm$.
  Then we can construct a Markov chain $\{\calD_n\}_{n \in \bbN}$ such that
  \begin{enumerate}[(i)]
    \item there exists a sequence $\{\calC_n\}_{n \in \bbN}$ such that $\calC_0 = \calC$, and for each $n \in \bbN$, it holds that $(\cdot) \potconf{q} \calC_n :: \Gm$, $\calC_n \approx \calD_n$, and
    \item for each $n \in \bbN$, it holds that $\bbE[\m{work}(\calD_n)] \le q$, where the expectation is computed with respect to the Markov chain.
  \end{enumerate}
  
  The expected total work for the executions starting from $\calC$ can then be defined as $\m{etw}(\calC) \defeq \lim_{n \to \infty} \bbE[\m{work}(\calD_n)]$.
  As a corollary of the Monotone Convergence Theorem, it holds that $\m{etw}(\calC) \le q$.
\end{theorem}

\section{Applications of Probabilistic Session Types}\label{sec:examples}

Probabilistic session types can be employed for a variety of diverse applications.
In this work, we demonstrate examples from 3 categories: \emph{(i)} implementing
and inferring expected complexity of randomized distributed algorithms,
\emph{(ii)} verifying limiting distributions of Markov chains, and
\emph{(iii)} proving correctness of digital contracts.

\subsection{Randomized Distributed Algorithms}
Since session types support concurrent programming, we can implement and
analyze randomized distributed algorithms in our language. In this paper,
we present the randomized dining philosophers, the randomized synchronous
leader election protocols and the dining cryptographers protocol.

\paragraph{\textbf{Randomized Dining Philosophers}}
Dining philosophers~\cite{Dijkstra71ACTA,Hoare78CACM} is a standard
illustration of synchronization and deadlock problems in concurrent
systems. The problem is formulated as five philosophers seated at
a circular table with food in front of them. Forks are placed between
each pair of adjacent philosophers. Each philosophers alternates between
thinking and eating with the constraint that they need both left and
right forks to eat. The fork is a shared resource used by both adjacent
philosophers, and represented using the shared session type
\begin{mathpar}
  \m{sfork} \triangleq \up \pichoice{\mb{available}^{0.4} : \down \m{sfork},
  \mb{unavailable}^{0.6} : \down \m{sfork}}
\end{mathpar}
The $\up$ type operator defines that $\m{sfork}$ is a shared session
type~\cite{Balzer17ICFP}
that can be acquired by either philosopher. Once acquired, it sends a label
describing its status. If the fork is available, it sends the label $\mb{available}$,
otherwise it sends the label $\mb{unavailable}$. We describe this behavior
with a probabilistic choice $\pichoiceop$, and the actual probabilities of
each label depends on the ratio expected amount of time spent thinking and
eating. Here, we arbitrarily decide that the forks are available 40\% of the
time.
Although we have only conducted a formal soundness theorem of probabilistic
session types in the linear fragment~\cite{Das18RAST}, we believe our
results will extend to the shared fragment. Sharing in session types
is largely orthogonal to probabilistic behavior since the two new type
formers we introduce ($\pichoiceop$ and $\pechoiceop$) only exist in the
linear fragment which has no interaction with the shared fragment.

In 1981, \citet{Lehmann81POPL} proved that there is no fully distributed
and symmetric \emph{deterministic} algorithm for the dining philosophers
problem that is \emph{deadlock-free}. They also proposed a randomized
deadlock-free algorithm for the same. The key idea is that each philosopher
tosses a coin to decide whether they acquire the left or the right fork
first. Since the coin tosses for each philosopher are independent random
events, eventually some philosopher would obtain both forks. We implemented
the randomized algorithm in our language using two processes.

\begin{verbatim}
  decl thinking : (l : sfork) (r : sfork) |{*}- (phil : 1)
  decl eating : (l : lfork) (r : lfork) |{*}- (phil : 1)
\end{verbatim}
The type $\m{lfork} = \pichoice{\mb{available}^{0.4} : \down \m{sfork},
\mb{unavailable}^{0.6} : \down \m{sfork}}$ represents an acquired fork.
The thinking process uses two channels as arguments: the fork on the left
and right and flips a coin. If the coin outputs $\heads$, they first acquire
left fork, then the right fork and only if both are available, they transition
to eating calling the same process. If the coin outputs $\tails$, they follow
the same procedure except for acquiring the right fork first. If either
of the forks is unavailable, they recurse back to thinking again.

Our language automatically infers the expected amount of time it takes
for a philosopher to start eating. The $\{*\}$ on the turnstile for the
$\m{thinking}$ and $\m{eating}$ processes signals the inference engine
to compute the expected potential required to typecheck the process.
In this example, we use the \textbf{flip} cost model that counts the
expected number of flips. For the above probabilities, the inference
engine returns the following potentials.
\begin{verbatim}
  decl thinking : (l : sfork) (r : sfork) |{6.25}- (phil : 1)
  decl eating : (l : lfork) (r : lfork) |{0}- (phil : 1)
\end{verbatim}


Figure~\ref{fig:plots-rnd}(a) describes the expected cost of the
$\m{thinking}$ process plotted against the probability of
availability of forks. Thus, if the types of the shared fork
and thinking process are
\begin{tabbing}
  \quad$\m{sfork} \triangleq \up \pichoice{\mb{available}^{p} : \down \m{sfork},
  \mb{unavailable}^{1-p} : \down \m{sfork}}$
\end{tabbing}
\begin{verbatim}
  decl thinking : (l : sfork) (r : sfork) |{q}- (phil : 1)
\end{verbatim}
then Figure~\ref{fig:plots-rnd}(a) plots $p$ and $q$ on the x and
y axes respectively. It is clear from Figure~\ref{fig:plots-rnd}(a)
that the expected cost is inversely proportional to the availability
of forks, which matches the expected behavior of the algorithm.

\paragraph{\textbf{Synchronous Leader Election}}
A leader election protocol operates on a network of processes that
communicate with each other to designate a \emph{unique} process
among them as the leader. Such protocols have diverse applications
since the leader can organize task distribution, monitor process
co-ordination, and gather and broadcast messages.

In 1990, \citet{Itai90IC} proved that if the processes are indistinguishable,
then there exists no \emph{deterministic} protocol to elect a leader.
They also proposed a randomized protocol for leader election in a
\emph{ring network} which proceeds
in rounds. In each round, each process (independently) chooses a
random number in the range $\{1, \ldots, K\}$ for some parameter $K$
as an \emph{id}. The processes then pass their ids around the ring.
If there is a \emph{unique maximum id}, then that process is elected
as the leader. Otherwise, the processes initiate a new round. This
protocol terminates with probability 1 because eventually the
random numbers chosen by the processes will have a unique maximum.


\begin{figure}
  \centering
  \includegraphics[width=0.9\linewidth]{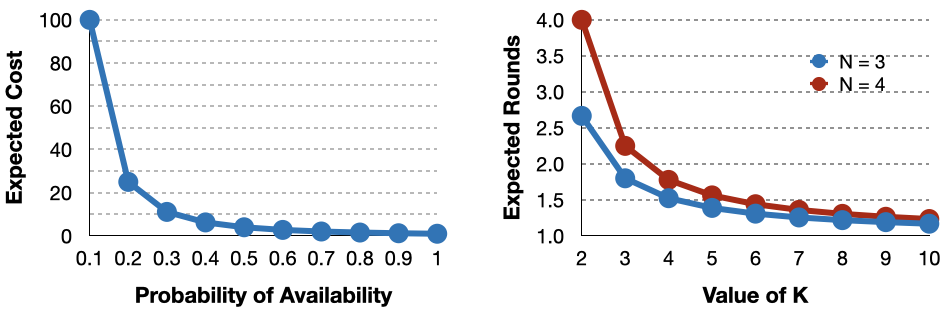}
  \caption{(a) Expected cost of dining philosophers (y-axis) vs probability
  of availability of forks (x-axis), and (b) Expected number of rounds in
  leader election (y-axis) vs K (x-axis)}
  \label{fig:plots-rnd}
  \Description{Plots for Randomized Algorithms}
\end{figure}

We have implemented this protocol in our language which automatically
infers its expected cost. In this protocol, we have used a special
cost model that counts the number of rounds. Since we can only
represent constant probabilities, the protocol has been implemented
with a fixed number of processes and $K$. Figure~\ref{fig:plots-rnd}(b)
plots the expected number of rounds vs $K$. We have implemented this
protocol for 2 ring networks with 3 and 4 processes, respectively
as showed in Figure~\ref{fig:plots-rnd}(b). Note that the probability
of obtaining a unique maximum is directly proportional to the number of
processes and inversely proportional to $K$. Thus, the expected
rounds decreases as $K$ increases for a fixed network. Also, for a
fixed $K$, increasing the number of processes in the ring increases
the expected number of rounds. Both these observations are confirmed
by Figure~\ref{fig:plots-rnd}.

\paragraph{\textbf{Dining Cryptographers}}
This is a standard protocol~\cite{Chaum88JC} that demonstrates transfer of messages
that are unconditionally and cryptographically secure with sender
and recipient untraceability. It is often applied to performing
secure multi-party computation. The problem scenario contains three
cryptographers seated at a circular table for dinner. The waiter
informs them that the dinner is paid for by either one of the
cryptographers or the NSA (National Security Agency). The
cryptographers want to respect their right to make anonymous
payments, but still want to know whether the bill was paid by the
NSA or not.

The protocol operates in two stages. In the first stage, each
cryptographer flips an \emph{unbiased coin} and informs the
cryptographer on the right of the outcome. Each cryptographer
then compares the outcome of their own coin toss with the informed
coin toss. In the second stage, each cryptographer publicly announces a
value which is either \emph{i)} `agree' if the two coin tosses matched,
or \emph{ii)} `disagree' if the two coin tosses do not match. Except
if a cryptographer paid the bill, they publicly announce the complement
value, i.e. `agree' if the coin tosses did not match, and `disagree'
otherwise. Finally, an even number of `agree's indicates that NSA paid
the bill, and an odd number indicates one of the cryptographer paid.
Most importantly, the protocol does not reveal which cryptographer
actually paid the bill (if one of them did).

Crucial to the correctness of the protocol is the use of an \emph{unbiased
coin}. We can ensure this with a probabilistic session type for the
shared coin used in this protocol defined as
\begin{eqnarray*}
  \m{scoin} & \triangleq & \up \pichoice{\heads^{0.5} : \, \down \m{scoin},
  \tails^{0.5} : \, \down \m{scoin}} \\
  \m{outcome} & \triangleq & \pichoice{\mb{agree}^{0.5} : \one,
  \mb{disagree}^{0.5} : \one}
\end{eqnarray*}
The coin is represented using a shared type so that it can be shared by
adjacent cryptographers. The $\up$ denotes a shared type, and the coin
guarantees to return $\heads$ and $\tails$ with equal probability. Once
it sends the toss value, it transitions back to the shared $\m{scoin}$
type using a $\down$ operator. The type $\m{outcome}$ denotes the value
that is publicly announced by the cryptographer, it can be $\mb{agree}$
or $\mb{disagree}$ with equal probability. The cryptographer process
is then declared as
\begin{verbatim}
  decl cryptographer : (left : scoin), (my : scoin) |{*}- (c : outcome)
\end{verbatim}
The type of the $\m{cryptographer}$ process shows that it uses two shared
channels: $\m{left}$ denotes the shared coin of the cryptographer
on the left, while $\m{my}$ denotes the cryptographer's own coin. Finally,
the process offers $c : \m{outcome}$ which means it will send one of the
two outcomes and then terminate.

Our inference engine automatically computes the cost for the $\m{cryptographer}$
process. Since our analysis is parametric in the cost model, the programmer
is free to choose their own. Under the \textbf{flip} cost model, its expected
potential is $1$, while under the \textbf{send} cost model, its expected
potential is $2$. This corroborates our intution, since the process flips
one coin, and sends 2 messages (outcome `agree'/`disagree' and $\m{close}$).

\subsection{Markov Chains}
We have already seen in Section~\ref{sec:overview} that our language
is adept at representing and analyzing Markov chains. One of the most common
problems studied about Markov chains is their \emph{asymptotic behavior}.
This involves studying the fraction of time spent in each state of a Markov
chain. The \emph{limiting distribution} of a Markov chain $M$ is a vector
$\mb{\overline{\pi}}$ where $\pi_i$ denotes the fraction of time spent in
state $i$. A chain is often described using a transition probability matrix
$\mb{P}$ such that $\mb{P}_{ij}$ denotes the probability of transitioning from
state $i$ to state $j$. Then, the limiting distribution satisfies the
following equations
\begin{mathpar}
  \mb{\overline{\pi}} \cdot \mb{P} = \mb{\overline{\pi}} \qquad \text{and}
  \qquad \sum_{i} \pi_i = 1
\end{mathpar}
These equations can be verified by our type system, as we will
demonstrate through several examples. Thus, we can verify whether a certain
distribution $\mb{\overline{\pi}}$ is a limiting distribution for a Markov chain $M$.

\paragraph{\textbf{Google's PageRank algorithm}}
The PageRank algorithm~\cite{Page99TR} was designed at Google to rank web
pages in their search engine results. The algorithm represents each webpage
as a state in a giant Markov chain. And we draw an edge from state $i$ to state
$j$ if there is a link on page $i$ leading to page $j$. Next, we need to assign
probabilities on each edge. For simplicity, if a state has $k$ outgoing links,
each is marked with probability $1/k$. The limiting distribution of such a chain
then denotes the fraction of time a user spends on each page. Hence, the pages
are ranked based on their limiting probabilities in decreasing order.

\begin{figure}
  \centering
  \includegraphics[width=0.5\linewidth]{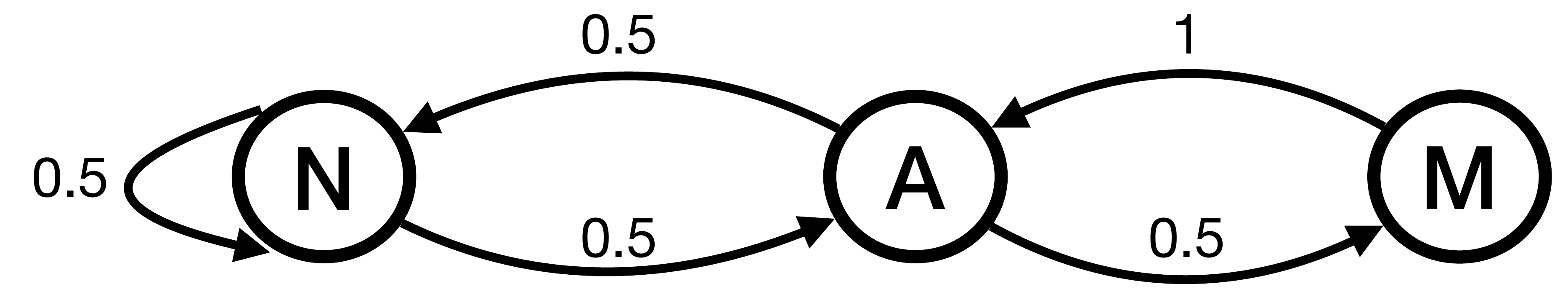}
  \caption{Markov chain with 3 webpages A, M and N}
  \label{fig:pagerank}
  \Description{Chain of PageRank}
\end{figure}

As an illustration, Figure~\ref{fig:pagerank} presents a simple network borrowed
from prior work~\cite{HarcholBalter13Book} containing 3 webpages named A, M and N
with the corresponding edge probabilities. The limiting distribution for this
chain is known to be $\pi_A = 0.4, \pi_M = 0.2, \pi_N = 0.4$.

We describe how to implement and verify this distribution in our language. First,
the limiting distribution is represented with a probabilistic session type named $\m{limit}$.
\begin{mathpar}
  \m{limit} \triangleq \pichoice{A^{0.4} : \one, M^{0.2} : \one, N^{0.4} : \one}
\end{mathpar}
The probability of each label exactly matches the limiting probability of the
corresponding state. The transition probability matrix is represented using a process 
\begin{verbatim}
  decl transition : (in : limit) |- (out : limit)
  proc out <- transition in =
    case in ( A => flip {0.5} ( H => out.N ; wait in ; close out
                              | T => out.M ; wait in ; close out )
            | M => out.A ; wait in ; close out
            | N => flip {0.5} ( H => out.A ; wait in ; close out
                              | T => out.N ; wait in ; close out ) )
\end{verbatim}
This $\m{transition}$ process exactly implements the equation
$\mb{\overline{\pi}} \cdot \mb{P} = \mb{\overline{\pi}}$. We case analyze
on the input channel $\m{in}$, where each branch denotes the state the chain
is currently in. In branch $A$, with prob. $0.5$, we transition to
state $N$ represented by $\esendl{\m{out}}{N}$ and with prob. $0.5$, we transition
to state $M$ represented by $\esendl{\m{out}}{M}$. In branch $M$, since
we always transition to state $A$, we send $\esendl{\m{out}}{A}$. Finally,
state $N$ transitions to $A$ with prob. $0.5$ and back to $N$ with prob. $0.5$,
as represented in the $\heads$ and $\tails$ branches by $\esendl{\m{out}}{A}$
and $\esendl{\m{out}}{N}$ respectively. Finally, we terminate each branch.

The type system automatically computes the probability of outputting each
label by taking the weighted sum of each branch. Because the input and
output types are equal, we conclude that the $\m{limit}$ type
indeed represents the limiting distribution. Furthermore, the validity of
type $\m{limit}$ ensures $\sum_{i} \pi_i = 1$.
We have implemented several standard Markov chains taken
from prior work~\cite{HarcholBalter13Book} such as probabilistic weather
prediction, random walks, natural numbers, machine repair, etc. and
verified their limiting distribution.
Finally, note that our LP solver cannot infer the probabilities for the
type $\m{limit}$. Since $\m{limit}$ is both the input and output type,
the constraints generated by our typing rules are non-linear.
In the future, we plan to use more advanced solvers to solve non-linear
constraints. 

\subsection{Digital Contracts}

Binary session types have also been employed in representing digital
contracts~\cite{Das19Nomos} such as auctions, lotteries, bank accounts, etc.
We can use probabilistic session types to inject probabilistic
behavior into such contracts. We describe here one such example to
model slot machines.

\paragraph{\textbf{Slot Machines}}
The core idea behind slots is that any player can use a \emph{ticket}
to play \emph{once} on a slot. And the player wins with a certain
probability $p$ and loses otherwise. If the player wins, they win
the entire money stored in the machine. Remarkably, we can use
potential to represent money! Thus, a player pays one unit of potential
to play the slots, and all the potential stored inside the slot machine is
paid to the player when they win the slots.
\begin{mathpar}
  \m{slot} \triangleq \up \getpot^1
  \pichoice{\mb{won}^{0.2} : \paypot^* \down \m{slot},
  \mb{lost}^{0.8} : \,\down \m{slot}}
\end{mathpar}

The type $\m{slot}$ represents the interface to a slot machine. To allow
multiple players, the type is shared: the $\up$ represents that the channel
must be acquired to play. Once acquired, the player must deposit the ticket
by paying $1$ unit of potential, as represented with $\getpot^1$. Then,
the type transitions to a probabilistic internal choice, denoting that
the player can win with prob. $0.2$. Thus, the player is \emph{guaranteed
a 20\% chance of winning}. If the player wins, the type sends the $\mb{won}$
label followed by an unknown ($*$) amount of potential. The $*$ indicates that we
would like the type system to infer how much potential the player wins.
On the other hand, if
the player loses, the type only sends the $\mb{lost}$ label. Then, in either
case, the type detaches with the $\down$ operator issuing a release. The
slot machine is represented using the $\m{machine}$ process as
\begin{verbatim}
  decl machine : . |{*}- (sl : slot)
  proc sl <- machine =
    ll <- accept sl ;
    get ll {1} ;
    flip {0.2} ( H => ll.won ;
                      pay ll {*} ;
                      sl <- detach ll ;
                      sl <- machine
               | T => ll.lost ;
                      sl <- detach ll ;
                      sl <- machine )
\end{verbatim}
The process offers channel $sl : \m{slot}$ and does not use any channels.
The process initiates with accepting an acquire request creating a fresh
linear channel $ll$ where messages are exchanged. It then receives 1 unit
of potential and flips with prob. $0.2$. In the $\heads$ branch, it sends
the label $\mb{won}$ followed by sending $*$ units of potential. The process then
detaches from the player consuming the linear channel $ll$ and recovering
the shared channel $sl$ and recurses. In the $\tails$ branch, the process
sends the $\mb{lost}$ label, detaches from the client and recurses.

The inference engine automatically computes the potential to be paid to a
player who wins. It produces the following annotation
\begin{mathpar}
  \m{slot} \triangleq \up \getpot^1
  \pichoice{\mb{won}^{0.2} : \paypot^5 \down \m{slot},
  \mb{lost}^{0.8} : \,\down \m{slot}}
\end{mathpar}
Thus, the player wins $5$ times the price of a ticket. This matches our
expectation as the winning probability is $0.2 = 1/5$. In general, the
expected winnings of a slot machine with winning probability $p$ can
be calculated to be $1/p$. This is confirmed in our language for any
fixed constant $p$. In conclusion, probabilistic session types statically
guarantee the winning probability and also determine the expected winnings
for digital contracts.

\section{Implementation and Evaluation}\label{sec:impl}
We have implemented an open-source prototype for \lang{} in OCaml
(7622 lines of code). The lexer and parser (601 lines of code) for \lang{} are
implemented using Menhir~\cite{Menhir19}, an LR(1) parser generator for
OCaml. A program in \lang{} is a sequence of mutually recursive type
and process definitions.

\paragraph{\textbf{Type Reconstruction}}
We implemented a \emph{bi-directional type checker}~\cite{Pierce00TOPLAS}
for \lang{} (4024 lines of code) specifically focusing on the quality of
error messages. The programmer provides the initial type for each process
in the declaration, and the intermediate types are reconstructed while
type checking the corresponding definition. This aids in localizing the
source of the error as the program location where type reconstruction fails.

An important aspect of the type checking algorithm is type reconstruction
resulting from the \emph{non-determinism}
in the rules for $\m{flip}$ and $\m{pcase}$. To typecheck a probabilistic branch,
we need to \emph{guess} the types of each channel in each branch (see
Figure~\ref{fig:prob-types}). This problem is exacerbated when such branches
are nested. Consider the $\m{unbias}$ process again that involves such a nesting
from Section~\ref{sec:overview}.
\begin{mathpar}
  \footnotesize{
    \infer[]
    {(b : \m{pbool}) \vdash \m{pcase} \; b \; (\m{true} \Rightarrow \eflip{0.5}{\ldots}{\ldots} \mid \m{false} \Rightarrow \eflip{0.5}{\ldots}{\ldots}) :: (c : \m{ubool})}
    {\infer[]
    {
      (b : 1) \vdash \eflip{0.5}{\esendlp{c}{\m{false}} \semi \ldots}{\esendlp{c}{\m{true}} \semi \ldots} :: (c : \mb{A})
    }
    {
      (b : 1) \vdash \esendlp{c}{\m{false}} \semi \ldots :: (c : \mb{A_1}) \and
      (b : 1) \vdash \esendlp{c}{\m{true}} \semi \ldots :: (c : \mb{A_2}) \and
      \mb{A} = 0.5 \cdot \mb{A_1} +^{\m{R}} 0.5 \cdot \mb{A_2}
    }
    }
  }
\end{mathpar}
We show a part of the derivation tree of the $\m{unbias}$ process ($\m{false}$ branch not shown).
Although we know the initial type $c : \m{ubool}$ (conclusion at the bottom),
we need to guess its intermediate type in the $\m{true}$ and $\m{false}$ branches.
Suppose we guess type $\mb{A}$ for $c$ in the $\m{true}$ branch and similarly
type $\mb{B}$ in the $\m{false}$ branch (not shown). We would obtain the constraint
$\m{ubool} = 0.6 \cdot \mb{A} +^{\m{R}} 0.4 \cdot \mb{B}$ taking the prob. of each
branch into account. Then, we again need to typecheck
a $\m{flip}$ expression. So, we again guess types $\mb{A_1}$ and $\mb{A_2}$ in
the $\heads$ and $\tails$ branches (resp. $\mb{B_1}$ and $\mb{B_2}$ for $\mb{B}$) respectively.
Then, we obtain the constraint $\mb{A} = 0.5 \cdot \mb{A_1} +^{\m{R}} 0.5 \cdot \mb{A_2}$
as shown in the derivation. The outcome of the type checking of this process would
depend on the satisfication of these constraints.

Our type checker solves this issue using a 2-phase reconstruction algorithm.
We initialize the program by replacing all $*$ annotations in all the type and process
definitions by variables.
Then, in the first phase of type checking, all prob.\ values are ignored
and the remaining typing constraints are checked.
In this phase, we also use the prob.\ values from the type of each probabilistic
channel to fill in the prob.\ of each branch in a $\m{pcase}$ expression.

In the second phase, we use a \emph{bottom-up} approach to guess the type of each
probabilistic channel in each branch.
For instance, in the derivation above, consider the left premise on the top.
Since we send label $\mb{false}$ on $c$, the type $A_1$ must be
$\pichoice{\mb{true}^0 : \one, \mb{false}^1 : \one}$.
Similarly, from the second premise, we infer that $\mb{A_2}$ must be
$\pichoice{\mb{true}^1 : \one, \mb{false}^0 : \one}$.
And the third premise shows how to compute type $\mb{A}$ from types $\mb{A_1}$
and $\mb{A_2}$.
In general, applying the typing rules generate linear constraints on the
probability variables which are collected by the type checker.
%
These constraints are then shipped to an off-the-shelf LP solver called Coin-Or. The LP solver
either solves these equations, or returns that they are infeasible.
In the former case, we obtain the solution and substitute them back in the
typing derivation. In the latter case, we
report a type checking failure.

\paragraph{\textbf{Potential Inference}}
We have implemented an inference engine to automatically compute the
expected cost of programs in \lang{} (821 lines of code).
Since the exact potential annotations depend on the cost assigned to
each operation and are difficult to predict statically, we found
inference to be extremely useful to make \lang{} practically applicable.

Using ideas from existing techniques for type inference for
AARA~\cite{Jost03,HoffmannW15}, we reduce the reconstruction of 
potential annotations to linear optimization.
To this end, the inference engine again relies on the Coin-Or LP solver. 
The programmer can indicate unknown potential using $*$ in the program code.
Thus, resource-aware session types can be marked with $\paypot^*$ and $\getpot^*$,
and process definitions can be marked with $|\{*\}-$ on the turnstile.

The inference engine first iterates over the program and substitutes
the $*$ annotations with potential variables. Then, the bidirectional typing rules
are applied, approximately checking the program (modulo potential annotations)
while also generating linear constraints for potential (see the rules for
$\m{flip}$, $\pichoiceop L$ and $\pechoiceop R$). Finally, these constraints are
shipped to the LP solver, which minimizes the value of the
potential annotations to achieve tight bounds. The LP solver either returns
that the constraints are infeasible, or returns a satisfying assignment,
which is then substituted into the program. The pretty printer (527 lines of
code) then prints the final program for the programmer to verify the potential
annotations.

\begin{table}[t]
\small
\centering
\begin{tabular}{@{}l r r r r r r r r}
\textbf{Program} & \textbf{LOC} & \textbf{Defs} & \textbf{Procs} & \textbf{T (ms)} &
\textbf{Vars} & \textbf{Cons} &
\textbf{I (ms)} \\
\midrule
3 die & 37 & 3 & 3 & 0.301 & 24 & 72 & 2.953 \\
6 die & 109 & 7 & 7 & 0.339 & 54 & 160 & 4.925 \\
exp. trials & 10 & 1 & 1 & 0.319 & 8 & 24 & 2.830 \\
fair coin & 17 & 1 & 1 & 0.313 & 20 & 54 & 2.964 \\
pagerank & 26 & 1 & 1 & 0.324 & 34 & 105 & 3.054 \\
repair & 22 & 2 & 1 & 0.343 & 39 & 118 & 2.924 \\
rnd walk & 27 & 1 & 2 & 0.043 & 32 & 96 & 3.466 \\
din. phil. & 94 & 1 & 5 & 0.361 & 135 & 412 & 3.612 \\
leader & 161 & 0 & 18 & 0.341 & 126 & 360 & 3.337 \\
din. crypto. & 30 & 1 & 1 & 0.345 & 49 & 151 & 2.871 \\
lossy chan. & 18 & 2 & 1 & 0.329 & 19 & 59 & 3.165 \\
slots & 17 & 1 & 1 & 0.451 & 21 & 58 & 3.351 \\
lottery & 16 & 1 & 1 & 0.342 & 19 & 52 & 2.816 \\
nats & 42 & 6 & 3 & 0.356 & 48 & 137 & 3.134 \\
\midrule
\end{tabular}
\caption{Evaluation of \lang{}.
   LOC = lines of code;
   Defs = \#type definitions;
   Procs = \#process definitions;
   T (ms) = type checking time in ms;
   Vars = \#potential variables generated during type checking; 
   Cons = \#constraints generated during type checking;
   I (ms) = potential inference time in ms.}
\label{tab:eval}
\vspace{-2em}
\end{table}

\paragraph{\textbf{Evaluation}}

All the examples presented so far have been implemented and type checked
in the \lang{} prototype. In addition, we have implemented several
other benchmarks that we briefly describe below.
Table~\ref{tab:eval} contains a compilation of our experiments with
the these programs and the prototype implementation. The experiments
were run on an Intel Core i5 2.7 GHz processor with 16 GB 1867 MHz DDR3
memory. It presents the program name, its lines of code (LOC), the
number of type (Defs) and process definitions (Procs), the type checking time (T (ms)),
number of potential variables introduced (Vars),
number of potential constraints that were generated while
type checking (Cons) and the time the LP solver took to infer their
values (I (ms)).

The program \textbf{3 die} is 3-faced die implementation presented in
Section~\ref{sec:overview} and \textbf{6 die} is the standard 6-faced
die implementation~\cite{Knuth76}.
Similarly, \textbf{fair coin} contains a process that recursively
flips a biased coin to produce a fair coin. The probabilistic
annotations on the process type guarantees that the process uses a
\emph{biased coin} and produces a \emph{fair coin}. The inference
engine computes the expected number of flips needed.
The \textbf{lossy chan.} example implements a channel from distributed
systems that drops messages with a certain probability. We can
automatically infer the expected number of messages a sender needs to
send so that the recipient receives it.
The lottery contract \textbf{lottery} guarantees a certain winning
probability in its session type.  The program \textbf{exp. trials:}
recursively flips a coin until it outputs $\heads$.

The \textbf{repair} program implements a Markov chain representing a
faulty machine.  The limiting distribution verified by \lang{}
estimates the fraction of time the machine spends in repair.
Similarly, \textbf{rnd walk} and \textbf{weather} are Markov chains
describing a random walk along a 2D line and probabilistic weather
patterns, respectively. We can in infer the expected time for the
random walker to reach home and verify the limiting distribution
that estimates the expected fraction of rainy and sunny days.
The program \textbf{nats} implements Probabilistic natural numbers
that send \emph{successor} and \emph{zero} with a fixed
probability. We implemented standard functions like add, double, etc.
to study their expected behavior.
The programs \textbf{din. phil} for randomized dining
  philosophers, \textbf{leader} for synchronous leader election,
  \textbf{din. crypto.} for dining cryptographers, 
  \textbf{slots} for the slot machine contract,
  and \textbf{pagerank} for the PageRank algorithm have been described 
in detail in Section~\ref{sec:examples}.


\section{Related Work}\label{sec:related}
 We classify the related work into five categories.


\paragraph{\textbf{Probabilistic Session Types}}
\citet{Aman19Arxiv} proposed a typing system extending
multiparty session types~\cite{Honda08POPL} with probabilistic internal
choice and non-deterministic external choice. Their session typing
discipline contains probabilistic intervals as opposed to \lang{},
where probabalities are exact.
\citet{InversoMP20} developed a system with probabilistic
binary session types.
Their system does not support non-probabilistic internal/external
choices and channel passing, whereas \lang{} is a conservatory
extension of intuitionistic session types~\cite{Caires10concur}.
A distinguishing feature of \lang{}
from all the aforementioned works is that we automatically infer
expected cost.

\paragraph{\textbf{Probabilistic Process Algebras}}
Probabalities were introduced in process algebras~\cite{Bergstra84IC}
by a probabilistic internal choice operator and extended with parallel
composition~\cite{Andova99ARTPS}. \citet{Herescu00FOSSACS} later
proposed an extension of the asynchronous $\pi$-calculus with a notion
of random choice distinguishing between probabilistic choice internal
to a process and non-deterministic external choice made by an adversarial
scheduler. They further use these techniques to prove probabilistic
correctness of the leader election protocol under any possible scheduler.
Other extensions to model time~\cite{Hansson94Book} and
performance~\cite{Hillston96Book} have also been proposed working on the same
principle of a probabilistic choice operator. In contrast to \lang{}, none
of these works enhance the type system with probabalities.

\paragraph{\textbf{Semantics of Probabilistic Languages}}
Several denotational models combining probabilistic and non-deterministic choices
have been developed~\cite{phd:Jones89,ENTCS:TKP09,LICS:Varacca02,JMSCS:VW06,ENTCS:MOW04,ENTCS:WHR19}, and some of them are focused on probabilistic
concurrency~\cite{CONCUR:Mislove00,phd:Varacca03,Varacca07QAPL}.
In the development of \lang{}, we use an operational semantic model.
Trace-based operational semantics have been used to analyze probabilistic concurrent programs~\cite{Tassarotti19POPL,TOPLAS:HSP83}.
There, the semantics maps a concurrent program to a distribution of execution traces.
Our approach is different from those techniques, in the sense that they investigate each
execution trace separately and connect them in the final phase of the analysis, whereas
our nested-multiverse semantics accounts for correlations among executions in different ``universes.''

\paragraph{\textbf{Reasoning About Probabilistic Programs}}

There exist several works on automatic expected cost analysis of
sequential (imperative) probabilistic
programs~\cite{PLDI:NCH18,PLDI:WFG19,CAV:CFG16,TACAS:KUH19}.  They can
derive symbolic polynomial bounds and can be seen as an automation of
Kozen's weakest pre-expectation
caclulus~\cite{JCSS:Kozen81,ESOP:KKM16}.
An automated type-based variant of this idea has been introduced
recently~\cite{WangKH20} in the context of automatic amortized resource
analysis~\cite{Jost03,HoffmannW15}.
The novelty of this work, is an automatic analysis for a probabilistic
language with concurrency. It builds on previous work on work analysis
for (deterministic) session types~\cite{Das18RAST}.

We are not aware of any other automated rule-based system for
deriving resource bounds for concurrent probabilistic programs.
However, there are multiple systems that can be used for manually
deriving expected cost
bounds\cite{Hansson94FAC,Tassarotti19POPL,Tassarotti018}.
More broadly, model and reasoning about probabilistic programs has been
extensively studied since the 80s~\cite{JCSS:Kozen81,SharirPH84,Kozen85,book:MM05}.
Recent work on modeling and analyzing
probabilistic networks~\cite{smolka2019scalable,ProbNetKAT,bayonet} can be viewed as reasoning systems
for probabilistic message passing systems.
However, the work on networks focuses on finite state systems and global
properties.
In contrast, the contribution of this article is the integration of
(local) probability distributions in session types and the automatic
expected resource analysis for message-passing processes.

\paragraph{\textbf{Probabilistic Model Checking}}
Most closely related to \lang{} are other works on verification
and complexity analysis of randomized algorithms. The probabilistic
model checker PRISM~\cite{Kwiatkowska11CAV} supports analyzing
discrete- and continuous-time Markov chains~\cite{Kwiatkowska07SFM},
Markov decision processes~\cite{Forejt11SFM} and probabilistic
timed automata~\cite{Kwiatkowska07IC}, as well as analyzing randomized
distributed algorithms~\cite{Norman04VOSS}. Instead of using type systems,
it provides a state-based language and a
specification language (that subsumes standard temporal logics)
to specify the model and property to be checked, respectively.
It employs state-of-the-art symbolic data structures such as Binary
Decision Diagrams (BDDs) and Multi-Terminal BDDs, a discrete-even
simulation engine, and analysis techniques such as quantitative
abstraction refinement and symmetry reduction. It has also been applied
for model checking $\pi$-calculus~\cite{Norman07QEST} including
\citet{Chaum88JC}'s dining cryptographers protocol. \citet{Bertrand19CONCUR}
extend threshold automata to model randomized algorithms parameterized
by the number of processes and failures under round-rigid schedules where
no process can initiate round $r+1$ until all processes complete round $r$.
In contrast to these works where the specification needs to be defined
separately, \lang{} unifies implementation and verification using an
enhanced type system. Another distinguishing feature of \lang{} is
that the linear fragment guarantees \emph{deadlock freedom} and a
\emph{confluence} property, i.e., the final configuration would be the
same under any adversarial scheduler.  Finally, types enable a
compositional analysis of different components instead of a whole
program analysis and bounds can be symbolic.


\section{Conclusion}\label{sec:conclusion}
In this article, we presented \lang{}, a probabilistic concurrent language
relying on novel \emph{probabilistic session types}. We also introduced
a novel nested-multiverse semantics to prove session fidelity, correctness
of expected bounds, and probability consistency. We employed them to
implement and verify correctness of Markov chains, infer expected cost of
randomized distributed algorithms and analyze expected behavior of
digital contracts.

One promising future direction is representing symbolic probabilities
on the choice operators. Currently, we only allow constant prob.\
annotations on the labels, but many distributed algorithms rely on choosing a
random number in the interval $\{1,..K\}$ for a variable $K$. Our
type system is also currently limited to producing linear expected
bounds. We plan to extend the type system to handle higher-degree
polynomial bounds. Randomized programs also often have logarithmic
expected bounds, so it would be interesting to infer such bounds as well.


\bibliographystyle{ACM-Reference-Format}
\bibliography{db,refs}

\clearpage
\appendix
\allowdisplaybreaks

\section{Technical Details of the Meta Theory}
\label{appendix:meta-theory}

\subsection{Proof of Preservation}
\label{appendix:preservation}

We start by assuming several standard results for session types (\cref{Prop:Subst,Prop:Perm,Prop:Truncate,Prop:Invariant,Prop:Existence}).

\begin{proposition}\label{Prop:Subst}
  The substitutions below are type-preserving and thus admissible:
  \begin{enumerate}[(i)]
    \item If $\Dl \entailpot{q} P :: (x : A)$, then for any fresh $y$, it holds that $\Dl \entailpot{q} P[y/x] :: (y : A)$.
    \item If $\Dl, (x:A) \entailpot{q} P :: (z : C)$, then for any fresh $y$, it holds that $\Dl, (y:A) \entailpot{q} P[y/x] :: (z : C)$.
  \end{enumerate}
\end{proposition}

\begin{proposition}\label{Prop:Perm}
  If $\Dl \potconf{q} (\calC_1 \parallel \proc{c}{\calE_c\tuple{d}} \parallel \calC_2 \parallel \proc{d}{\calE_d} \parallel \calC_3) :: \Gm$,
  then $\Dl \potconf{q} (\calC_1 \parallel \proc{c}{\calE_c\tuple{d}} \parallel \proc{d}{\calE_d} \parallel \calC_2 \parallel \calC_3) :: \Gm$.
\end{proposition}

\begin{proposition}\label{Prop:Truncate}
  If $\Dl \potconf{q} (\calC_1 \parallel \calC_2 \parallel \calC_3) :: \Gm$,
  then there exist $\Dl',q',\Gm'$ such that $\Dl' \potconf{q'} \calC_2 :: \Gm'$.
\end{proposition}

\begin{proposition}\label{Prop:Invariant}
  If $\Dl \potconf{q} (\calC_1 \parallel \calC_2 \parallel \calC_3) :: \Gm$ and $\Dl' \potconf{q'} \calC_2 :: \Gm'$, then for any $\calC_2'$ such that $\Dl' \potconf{q'} \calC_2' :: \Gm'$, it holds that $\Dl \potconf{q} (\calC_1 \parallel \calC_2' \parallel \calC_3) :: \Gm$.
\end{proposition}

\begin{proposition}\label{Prop:Existence}
  If $\Dl \potconf{q} \calC :: \Gm$, then for all $c \in \dom{\Gm}$, there exists exactly one $\proc{c}{\calE}$ for some $\calE$, in $\calC$.
\end{proposition}

The lemma below extends the rule (\textsc{T:Compose}) to prepend multiple semantic objects to a configuration.

\begin{lemma}\label{Lem:ConfExt}
  If $\Dl_1, \Dl' \potconf{q_1} \calC_1 :: \Gm$ and $\Dl_2 \potconf{q_2} \calC_2 :: (\Dl,\Dl')$,
  then $\Dl_1,\Dl_2 \potconf{q_1+q_2} (\calC_1 \parallel \calC_2) :: (\Dl, \Gm)$.
\end{lemma}
\begin{proof}
  By induction on the derivation of $\Dl_1,\Dl' \potconf{q_1} \calC_1 :: \Gm$.
\end{proof}

Now we are able to prove preservation for single-process operational rules.

\begin{theorem}[Preservation, part I of \cref{the:preservation}]\label{The:Tree:PreservationSingle}
  If $\Dl \potconf{q} \calC :: \Gm$ and $\calC \sstep \calC'$, then $\Dl \potconf{q} \calC' :: \Gm$.
\end{theorem}
\begin{proof}
  By induction on the derivation of $\calC \sstep \calC'$.
  By \cref{Prop:Truncate,Prop:Invariant}, it suffices to consider the atomic rewriting rules.
  Then we proceed by inversion on $\Dl \potconf{q} \calC :: \Gm$.
  
  \begin{description}[labelindent=\parindent]
  \item[Case:]
  \[\small\Rule{E:Flip}
  {
  }
  { \proc{c}{w,\eflip{p}{P_H}{P_T}} \sstep \proc{c}{\{ \proc{c}{w,P_H} : p, \proc{c}{w,P_T} : 1-p \}} }
  \]
  
  By inversion on the typing judgment, we have
  \[\small
  \inferrule
  { \Dl_H \entailpot{q_H} P_H :: (c : A_H) \\
    \Dl_T \entailpot{q_T} P_T :: (c : A_T) \\
    \Dl = p \cdot \Dl_H +^{\m{L}} (1-p) \cdot \Dl_T \\
    A = p \cdot A_H +^{\m{R}} (1-p) \cdot A_T \\
    q-w = p \cdot q_H + (1-p) \cdot q_T
  }
  { \Dl \potconf{q} \proc{c}{w,\eflip{p}{P_H}{P_T}} :: (c : A) }
  \]
    
  Thus, by \textsc{(T:Proc)}, we have $\Dl_H \potconf{q_H+w} \proc{c}{w,P_H} :: (c : A_H)$, $\Dl_T \potconf{q_T+w} \proc{c}{w,P_T} :: (c : A_T)$.
    
  Then we conclude by \textsc{(T:Dist)}, and the fact that $p \cdot (q_H+w) +(1-p) \cdot (q_T+w) = (p \cdot q_H + (1-p) \cdot q_T) + w = (q-w)+w = q$.
  
  \item[Case:]
  \[\small\Rule{E:Work}
  {
  }
  { \proc{c}{w,\ework{r} \semi P} \sstep \proc{c}{w+r,P}  }
  \]
  
  By inversion on the typing judgment, we have
  \[\small
  \inferrule
  { \Dl \entailpot{q-w-r} P :: (c : A)
  }
  { \Dl \potconf{q} \proc{c}{w,\ework{r} \semi P} :: (c : A) }
  \]
    
  Thus, by \textsc{(T:Proc)}, we have $\Dl \potconf{q} \proc{c}{w+r,P} :: (c : A)$.
  
  \item[Case:]
  \[\small
  \Rule{E:Dist}
  { \calC_{i_0} \sstep \calC_{i_0}'
  }
  { \proc{c}{\{ \calC_i : p_i \}_{i \in \calI}} \sstep \proc{c}{\{ \calC_i : p_i \}_{i \in \calI \setminus \{i_0\}} \dplus \{ \calC'_{i_0} : p_{i_0} \}} }
  \]
  
  By inversion on the typing judgment, we have
  \[\small
  \inferrule
  { \Forall{ i \in \calI} \Dl_i \potconf{q_i} \calC_i :: (c : A_i) \\
    \tsum^{\m{L}}_{i \in \calI} p_i \cdot \Dl_i = \Dl \\
    \tsum^{\m{R}}_{i \in \calI} p_i \cdot A_i = A \\
    \tsum_{i \in \calI} p_i \cdot q_i = q
  }
  { \Dl \potconf{q} \proc{c}{\{ \calC_i : p_i \}_{i \in \calI} } :: (c : A) }
  \]
    
  By induction hypothesis, we have $\Dl_{i_0} \potconf{q_{i_0}} \calC'_{i_0} :: (c : A_{i_0})$.
    
  Then we conclude by \textsc{(T:Dist)}.
  
  \item[Case:]
  \[\small\Rule{E:Def}
  { \many{y' : B} \entailpot{r} f = P_f :: (x' : C) \in \Sg \\
    b\; \text{fresh}
  }
  { \proc{c}{w, \ecut{x}{f}{\many{d}}{Q}} \sstep \proc{c}{w,Q[b/x]} \parallel \proc{b}{0,P_f[b/x',\many{d}/\many{y'}]}   }
  \]

  By inversion on the typing judgment, we have
  \[\small
  \inferrule
  { \many{y' : B} \entailpot{r} f = P_f :: (x' : C) \in \Sg \\
    \Dl', (x:C) \entailpot{q-w-r} Q :: (c : A)
  }
  { \Dl',\many{d:B} \potconf{q} \proc{c}{w,\ecut{x}{f}{\many{d}}{Q}} :: (c : A) }
  \]

  Thus, by \cref{Prop:Subst} and \textsc{(T:Proc}), we have $\Dl',(b:C) \potconf{q-r} \proc{c}{w,Q[b/x]} :: (c : A)$, as well as $\many{d:B} \potconf{r} \proc{b}{0,P_f[b/x',\many{d}/\many{y'}]} :: (b : C)$.
  
  Then we conclude by \textsc{(T:Compose)}.
  \end{description}
\end{proof}

\begin{theorem}[Preservation, part II of \cref{the:preservation}]\label{The:Tree:PreservationComm}
  If $\Dl \potconf{q} \calC :: \Gm$ and $\calC \cstep{d,\kappa} \calC'$, then $\Dl \potconf{q} \calC' :: \Gm$.
\end{theorem}
\begin{proof}
  By induction on the derivation of $C \cstep{d,\kappa} \calC'$.
  By \cref{Prop:Perm,Prop:Truncate,Prop:Invariant,Prop:Existence}, it suffices to consider the atomic rewriting rules.
  Then we proceed by inversion of $\Dl \potconf{q} \calC :: \Gm$.
  
  \begin{description}[labelindent=\parindent]    
    \item[Case:]
    \[\small\Rule{C:SDist:R}
    { \proc{c}{\{ ( \proc{c}{w,Q} \parallel \calC_i) : p_i \}_{i \in \calI}} \cstep{d,\kappa} \calC'
    }
    { \proc{c}{w,Q} \parallel \proc{d}{\{ \calC_i : p_i\}_{i \in \calI}} \cstep{d,\kappa} \calC'  }
    \]
    
    By inversion on the typing judgment, we have
    \[\small
    \inferrule
    { \Dl_1,(d:B) \potconf{q_1} \proc{c}{w,Q} :: (c : A)
      \\
      \inferrule
      { \Forall{i \in \calI} \Dl_{2,i} \potconf{q_{2,i}} \calC_i :: (d : B_i) \\\\
        \Dl_2 = \tsum^{\m{L}}_{i \in \calI} p_i \cdot \Dl_{2,i} \\\\
        B = \tsum^{\m{R}}_{i \in \calI} p_i \cdot B_i \\\\
        q_2 = \tsum_{i \in \calI} p_i \cdot q_{2,i}
      }
      { \Dl_2 \potconf{q_2} \proc{d}{\{ \calC_i : p_i \}_{i \in \calI}} :: (d : B) }
    }
    { \Dl_1,\Dl_2 \potconf{q_1+q_2} (\proc{c}{w,Q} \parallel \proc{d}{\{ \calC_i : p_i \}_{i \in \calI}})  :: (c : A) }
    \]
    
    We proceed by inversion on $B = \tsum^{\m{R}}_{i \in \calI} p_i \cdot B_i$.
    
    \begin{itemize}
      \item If $B \neq \pichoice{\cdots}$, then $B_i = B$ and $\Dl_1, \Dl_{2,i} \potconf{q_1+q_{2,i}} (\proc{c}{w,Q} \parallel \calC_i) :: (c : A)$ for all $i \in \calI$.
      
      Then we derive $\Dl_1,\Dl_2 \potconf{q_1+q_2} \proc{c}{\{ (\proc{c}{w,Q} \parallel \calC_i ) : p_i \}_{i \in \calI}} :: (c : A)$ by \textsc{(T:Dist)}.
      
      \item Otherwise, suppose that $B = \pichoice{\ell^{p_\ell'} : B_\ell}_{\ell \in L}$, and $B_i = \pichoice{\ell^{p_{\ell,i}} : B_\ell}_{\ell \in L}$ for $i \in \calI$, such that for each $\ell \in L$, it holds that $p_\ell' = \tsum_{i \in \calI} p_i \cdot p_{\ell,i}$.
      
      In this case, $Q$ must be a probabilistic case expression $\epcase{d}{\ell}{Q_\ell}_{\ell \in L}$.
      
      By inversion on the typing judgment, we have
      \[\small
      \inferrule
      { \Forall{\ell \in L} \Dl_{1,\ell}, (d:B_\ell) \entailpot{q_{1,\ell}} Q_\ell :: (c : A_\ell) \\
        \tsum^{\m{L}}_{\ell \in L} p_\ell' \cdot \Dl_{1,\ell} = \Dl_1 \\
        \tsum^{\m{R}}_{\ell \in L} p_\ell' \cdot A_\ell = A \\
        \tsum_{\ell \in L} p_\ell' \cdot q_{1,\ell} = q_1-w
      }
      { \Dl_1, (d:\pichoice{\ell^{p_\ell'} : B_\ell}_{\ell \in L}) \potconf{q_1} \proc{c}{w, \epcase{d}{\ell}{Q_\ell}_{\ell \in L}} :: (c : A) }
      \]
      
      Thus, for all $i \in \calI$, define $\Dl'_{1,i} \defeq \tsum^{\m{L}}_{\ell \in L} p_{\ell,i} \cdot \Dl_{1,\ell}$, $A'_i \defeq \tsum^{\m{R}}_{\ell \in L} p_{\ell,i} \cdot A_\ell$, $q_{1,i}' \defeq \tsum_{\ell \in L} p_{\ell,i} \cdot q_{1,\ell} + w$.
      
      By \textsc{(${\pichoiceop}L$)} and \textsc{(T:Proc)}, we have $\Dl'_{1,i}, (d : B_i) \potconf{q'_{1,i}} \proc{c}{w,Q} :: (c : A'_i)$.
      
      Then by \textsc{(T:Compose)}, we have $\Dl'_{1,i},\Dl_{2,i} \potconf{q'_{1,i}+q_{2,i}} ( \proc{c}{w,Q} \parallel \calC_i) :: (c : A'_i)$.
            
      Then we derive $\Dl_1,\Dl_2 \potconf{q_1+q_2} \proc{c}{\{ ( \proc{c}{w,Q} \parallel \calC_i  ):p_i \}_{i \in \calI}} :: (c : A)$ by \textsc{(T:Dist)}, and the fact that
      \begin{small}
      \begin{align*}
        \tsum^{\m{L}}_{i \in \calI} p_i \cdot (\Dl'_{1,i},\Dl_{2,i}) & = \tsum^{\m{L}}_{i \in \calI} p_i \cdot \tsum^{\m{L}}_{\ell \in L} p_{\ell,i} \cdot \Dl_{1,\ell}, \Dl_2 \\
        & = \tsum^{\m{L}}_{\ell \in L} (\tsum_{i \in \calI} p_i \cdot p_{\ell,i}) \cdot \Dl_{1,\ell} , \Dl_2 \\
        & = \tsum^{\m{L}}_{\ell \in L} p'_{\ell} \cdot \Dl_{1,\ell}, \Dl_2 \\
        & = \Dl_1,\Dl_2, \\
        \tsum^{\m{R}}_{i \in \calI} p_i \cdot A'_i & = \tsum^{\m{R}}_{i \in \calI} p_i \cdot \tsum^{\m{R}}_{\ell \in L} p_{\ell,i} \cdot A_\ell \\
        & = \tsum^{\m{R}}_{\ell \in L} (\tsum_{i \in \calI} p_i \cdot p_{\ell,i}) \cdot A_\ell \\
        & = \tsum^{\m{R}}_{\ell \in L} p'_{\ell} \cdot A_\ell \\
        & = A, \\
        \tsum_{i \in \calI} p_i \cdot (q'_{1,i} + q_{2,i}) & = \tsum_{i \in \calI} p_i \cdot (\tsum_{\ell \in L} p_{\ell,i} \cdot q_{1,\ell} + w) + q_2 \\
        & = \tsum_{\ell \in L} (\tsum_{i \in \calI} p_i \cdot p_{\ell,i}) \cdot q_{1,\ell} + w + q_2 \\
        & = \tsum_{\ell \in L} p_\ell' \cdot q_{1,\ell} + w + q_2 \\
        & = q_1-w+w+q_2 \\
        & = q_1 + q_2.
      \end{align*}
      \end{small}
    \end{itemize}
        
    Then, by induction hypothesis, we conclude that $\Dl_1,\Dl_2 \potconf{q_1+q_2} \calC' :: (c : A)$.
  
    \item[Case:]
    \[\small
    \Rule{C:SDist:L}
    { \proc{c}{\{ (\calC_i \parallel \proc{d}{w,P}) : p_i \}_{i \in \calI}} \cstep{d,\kappa} \calC'
    }
    { \proc{c}{\{ \calC_i : p_i \}_{i \in \calI}} \parallel \proc{d}{w,P} \cstep{d,\kappa} \calC' }
    \]
    
    By inversion on the typing judgment, we have
    \[\small
    \inferrule
    { \inferrule
      { \Forall{i \in \calI} \Dl_{1,i}, (d:B_i) \potconf{q_{1,i}} \calC_i :: (c : A_i) \\\\
        \Dl_{1} = \tsum^{\m{L}}_{i \in \calI} p_i \cdot \Dl_{1,i} \\
        B = \tsum^{\m{L}}_{i \in \calI} p_i \cdot B_i \\\\
        A = \tsum^{\m{R}}_{i \in \calI} p_i \cdot A_i \\\\
        q_1 = \tsum_{i \in \calI} p_i \cdot q_{1,i}
      }
      { \Dl_1,(d:B) \potconf{q_1} \proc{c}{\{ \calC_i : p_i \}_{i \in \calI}} :: (c : A) }
      \\
      \Dl_2 \potconf{q_2} \proc{d}{w,P} :: (d : B)
    }
    { \Dl_1,\Dl_2 \potconf{q_1+q_2} \proc{c}{\{\calC_i:p_i\}_{i \in \calI}} \parallel \proc{d}{w,P} :: (c : A) }
    \]
    
    We proceed by inversion on $B = \tsum^{\m{L}}_{i \in \calI} p_i \cdot B_i$.
    
    \begin{itemize}
      \item If $B \neq \pechoice{\cdots}$, then $B_i = B$ and $\Dl_{1,i},\Dl_{2} \potconf{q_{1,i}+q_{2}} (\calC_i \parallel \proc{d}{w,P})  :: (c : A_i)$ for all $i \in \calI$.
      
      Then we derive $\Dl_1,\Dl_2 \potconf{q_1+q_2} \proc{c}{\{ (\calC_i \parallel \proc{d}{w,P}) : p_i \}_{i \in \calI}} :: (c : A)$ by \textsc{(T:Dist)}.
      
      \item Otherwise, suppose that $B = \pechoice{\ell^{p_\ell'} : B_\ell}_{\ell \in L}$, and $B_i = \pechoice{\ell^{p_{\ell,i}} : B_\ell}_{\ell \in L}$ for $i \in \calI$, such that for each $\ell \in L$, it holds that $p_\ell' = \tsum_{i \in \calI} p_i \cdot p_{\ell,i}$.
      
      In this case, $P$ must be a probabilistic case expression $\epcase{d}{\ell}{P_\ell}_{\ell \in L}$, or a forwarding expression $\fwd{d}{e}$.
      
      \begin{itemize}
      \item
      By inversion on the typing judgment, we have
      \[\small
      \inferrule
      { \Forall{\ell \in L} \Dl_{2,\ell} \entailpot{q_{2,\ell}} P_\ell :: (c : B_\ell) \\
        \tsum^{\m{L}}_{\ell \in L} p_\ell' \cdot \Dl_{2,\ell} = \Dl_2 \\
        \tsum_{\ell \in L} p_\ell' \cdot q_{2,\ell} = q_2-w
      }
      { \Dl_2 \potconf{q_2} \proc{d}{w,\epcase{d}{\ell}{P_\ell}_{\ell \in L}} :: (d : \pechoice{\ell^{p'_\ell} : B_\ell}_{\ell \in L} ) }
      \]
      
      Thus, for all $i \in \calI$, define $\Dl'_{2,i} \defeq \tsum^{\m{L}}_{\ell \in L} p_{\ell,i} \cdot \Dl_{2,\ell}$, $q_{2,i}' \defeq \tsum_{\ell \in L} p_{\ell,i} \cdot q_{2,\ell} + w$.
      
      By \textsc{(${\pechoiceop}R$)} and \textsc{(T:Proc)}, we have $\Dl'_{2,i} \potconf{q'_{2,i}} \proc{d}{w,P} :: (d : B_i)$.
      
      Then by \cref{Lem:ConfExt}, we have $\Dl_{1,i},\Dl'_{2,i} \potconf{q_{1,i}+q'_{2,i}}  (\calC_i \parallel \proc{d}{w,P} ) :: (c : A_i)$.
      
      Then we derive $\Dl_1,\Dl_2 \potconf{q_1+q_2} \proc{c}{\{(\calC_i \parallel \proc{d}{w,P}) : p_i \}_{i \in \calI}} :: (c : A)$ by \textsc{(T:Dist)}, and the fact that
      \begin{small}
      \begin{align*}
        \tsum^{\m{L}}_{i \in \calI} p_i \cdot (\Dl_{1,i},\Dl'_{2,i}) & = \Dl_1, \tsum^{\m{L}}_{i \in \calI} p_i \cdot \tsum^{\m{L}}_{\ell \in L} p_{\ell,i} \cdot \Dl_{2,\ell} \\
        & = \Dl_1, \tsum^{\m{L}}_{\ell \in L} (\tsum_{i \in \calI} p_i \cdot p_{\ell,i}) \cdot \Dl_{2,\ell} \\
        & = \Dl_1, \tsum^{\m{L}}_{\ell \in L} p'_{\ell} \cdot \Dl_{2,\ell} \\
        & = \Dl_1,\Dl_2, \\
        \tsum_{i \in \calI} p_i \cdot (q_{1,i} + q'_{2,i}) & = q_1 + \tsum_{i \in \calI} p_i \cdot (\tsum_{\ell \in L} p_{\ell,i} \cdot q_{2,\ell} + w) \\
        & = q_1 + \tsum_{\ell \in L} (\tsum_{i \in \calI} p_i \cdot p_{\ell,i}) \cdot q_{2,\ell} + w \\
        & = q_1 + \tsum_{\ell \in L} p_\ell' \cdot q_{2,\ell} + w \\
        & = q_1 + q_2-w+w \\
        & = q_1 + q_2.
      \end{align*}
      \end{small}
      
      \item
      By inversion on the typing judgment, we have
      \[\small
      \inferrule
      { (e : \pechoice{\ell^{p_\ell'} : B_\ell}_{\ell \in L} ) \entailpot{0} \fwd{d}{e} :: (d : \pechoice{\ell^{p_\ell'} : B_\ell}_{\ell \in L})
      }
      { (e : \pechoice{\ell^{p_\ell'} : B_\ell}_{\ell \in L} ) \potconf{w} \proc{d}{w, \fwd{d}{e}} :: (d : \pechoice{\ell^{p_\ell'} : B_\ell}_{\ell \in L}) }
      \]
      
      Thus, for all $i \in \calI$, define $\Dl_{2,i}' \defeq (e : B_i)$, $q'_{2,i} \defeq w$.
      
      By \textsc{($\m{id}$)} and \textsc{(T:Proc)}, we have $\Dl'_{2,i} \potconf{q'_{2,i}} \proc{d}{w,P} :: (d : B_i)$.
      
      Then by \cref{Lem:ConfExt}, we have $\Dl_{1,i},\Dl'_{2,i} \potconf{q_{1,i}+q'_{2,i}} (\calC_i \parallel \proc{d}{w,P}) :: (c : A_i)$.
      
      Then we derive $\Dl_1,(e:B) \potconf{q_1+w} \proc{c}{\{(\calC_i \parallel \proc{d}{w,P}) : p_i \}_{i \in \calI}}  :: (c : A)$ by \textsc{(T:Dist)}, and the fact that
      \begin{small}
      \begin{align*}
        \tsum^{\m{L}}_{i \in \calI} p_i \cdot (\Dl_{1,i}, \Dl'_{2,i}) & = \Dl_1, (e : \tsum^{\m{L}}_{i \in \calI} p_i \cdot B_i) \\
        & = \Dl_1, (e : B), \\
        \tsum_{i \in \calI} p_i \cdot (q_{1,i} + q_{2,i}') & = q_1 + \tsum_{i \in \calI} p_i \cdot w \\
        & = q_1 + w.
      \end{align*}
      \end{small}
      \end{itemize}
    \end{itemize}
    
    Then, by induction hypothesis, we conclude that $\Dl_1,\Dl_2 \potconf{q_1+q_2} \calC' :: (c : A)$.
    
    \item[Case:]
    \[\small\Rule{C:BDist:D}
    { \proc{c}{\{ (\calC_i \parallel \calC_j') : p_i \cdot p_j' \}_{i \in \calI,j\in\calJ}} \cstep{d,\m{det}} \calC''
    }
    { \proc{c}{\{ \calC_i : p_i\}_{i \in \calI}} \parallel \proc{d}{\{ \calC_j' : p_j'\}_{j \in \calJ}} \cstep{d,\m{det}} \calC''  }
    \]
    
    By inversion on the typing judgement, we have
    \[\small
      \inferrule
      { \inferrule
        { \Forall{i \in \calI} \Dl_{1,i},(d:B_{1,i}) \potconf{q_{1,i}} \calC_i :: (c : A_{i}) \\\\
          \Dl_1 = \tsum^{\m{L}}_{i \in \calI} p_i \cdot \Dl_{1,i} \\
          B = \tsum^{\m{L}}_{i \in \calI} p_i \cdot B_{1,i} \\\\
          A = \tsum^{\m{R}}_{i \in \calI} p_i  \cdot A_{i} \\\\
          q_1 = \tsum_{i \in \calI} p_i \cdot q_{1,i}
        }
        { \Dl_1,(d:B) \potconf{q_1} \proc{c}{\{\calC_i : p_i \}_{i \in \calI}} :: (c : A) }
        \\
        \inferrule
        { \Forall{j \in \calJ} \Dl_{2,j} \potconf{q_{2,j}} \calC_j' :: (d : B_{2,j}) \\\\
          \Dl_2 = \tsum^{\m{L}}_{j \in \calJ} p_j' \cdot \Dl_{2,j} \\\\
          B = \tsum^{\m{R}}_{j \in \calJ} p_j' \cdot B_{2,j} \\\\
          q_2 = \tsum_{j \in \calI} p_j' \cdot q_{2,j}
        }
        { \Dl_2 \potconf{q_2} \proc{d}{\{ \calC_j' : p_j' \}_{j \in \calJ}} :: (d : B) }
      }
      { \Dl_1, \Dl_2 \potconf{q_1+q_2} \proc{c}{ \{\calC_i:p_i\}_{i \in \calI} } \parallel \proc{d}{\{\calC'_j:p_j'\}_{j \in \calJ}} :: (c : A) }
    \]
    
    Because $B \neq \pichoice{\cdots}$ and $B \neq \pechoice{\cdots}$, we have $B_{1,i} = B$ for all $i \in \calI$, $B_{2,j} = B$ for all $j \in \calJ$, by shallow weighted sums.
      
    Thus for all $i \in \calI$, $j \in \calJ$, we have $\Dl_{1,i}, \Dl_{2,j} \potconf{q_{1,i}+q_{2,j}} (\calC_i \parallel \calC_j') :: (c : A_i)$ by \cref{Lem:ConfExt}.
            
    Then we derive $\Dl_1,\Dl_2 \potconf{q_1+q_2} \proc{c}{\{ (\calC_i \parallel \calC'_j) : p_i \cdot p_j' \}_{i \in \calI, j \in \calJ}} :: (c : A)$ by \textsc{(T:Dist)}, and the fact that
    \begin{small}
    \begin{align*}
      \tsum^{\m{L}}_{i \in \calI, j \in \calJ} (p_i \cdot p_j') \cdot (\Dl_{1,i},\Dl_{2,j}) & = \tsum^{\m{L}}_{i \in \calI,j \in \calJ} (p_i \cdot p_j') \cdot \Dl_{1,i},  \tsum^{\m{L}}_{i \in \calI,j \in \calJ} (p_i \cdot p_j') \cdot \Dl_{2,j} \\
      & = \tsum^{\m{L}}_{j \in \calJ} p_j' \cdot \tsum^{\m{L}}_{i \in \calI} p_i \cdot \Dl_{1,i}, 
          \tsum^{\m{L}}_{i \in \calI} p_i \cdot \tsum^{\m{L}}_{j \in \calJ} p_j' \cdot \Dl_{2,j} \\
      & = \tsum^{\m{L}}_{j \in \calJ} p_j' \cdot \Dl_1, \tsum^{\m{L}}_{i \in \calI} p_i \cdot \Dl_2 \\
      & = \Dl_1,\Dl_2, \\
      \tsum^{\m{R}}_{i \in \calI, j \in \calJ} (p_i \cdot p_j') \cdot A_i & = \tsum^{\m{R}}_{j \in \calJ} p_j' \cdot \tsum^{\m{R}}_{i \in \calI} p_i \cdot A_i \\
      & = \tsum^{\m{R}}_{j \in \calJ} p_j' \cdot A \\
      & = A, \\
      \tsum_{i \in \calI, j \in \calJ} (p_i \cdot p_j') \cdot (q_{1,i}+q_{2,j}) & = \tsum_{i \in \calI, j \in \calJ} (p_i \cdot p_j') \cdot q_{1,i } + \tsum_{i \in \calI, j \in \calJ} (p_i \cdot p_j') \cdot q_{2,j} \\
      & = \tsum_{j \in \calJ} p_j' \cdot \tsum_{i \in \calI} p_i \cdot q_{1,i} + \tsum_{i \in \calI} p_i \cdot \tsum_{j \in \calJ} p_j' \cdot q_{2,j} \\
      & = \tsum_{j \in \calJ} p_j' \cdot q_1 + \tsum_{i \in \calI} p_i \cdot q_2 \\
      & = q_1 + q_2.
    \end{align*}
    \end{small}
    
    Then, by induction hypothesis, we conclude that $\Dl_1,\Dl_2 \potconf{q_1+q_2} \calC'' :: (c : A)$.
     
    \item[Case:]
    \[\small
    \Rule{C:BDist:R}
    {  \proc{c}{\{ (\proc{c}{\{ \calC_i : p_i\}_{i \in \calI}} \parallel \calC_j') : p'_j  \}_{j \in \calJ}} \cstep{d,\pechoiceop} \calC''
    }
    { \proc{c}{\{ \calC_i : p_i \}_{i \in \calI}} \parallel \proc{d}{\{ \calC'_j : p'_j \}_{j \in \calJ}} \cstep{d,\pechoiceop} \calC'' }
    \]
    
    By inversion on the typing judgement, we have
    \[\small
      \inferrule
      { { \Dl_1,(d:B) \potconf{q_1} \proc{c}{\{ \calC_i : p_i \}_{i \in \calI}} :: (c : A) }
        \\
        \inferrule
        { \Forall{j \in \calJ} \Dl_{2,j} \potconf{q_{2,j}} \calC_j' :: (d : B_{2,j}) \\\\
          \Dl_2 = \tsum^{\m{L}}_{j \in \calJ} p_j' \cdot \Dl_{2,j} \\\\
          B = \tsum^{\m{R}}_{j \in \calJ} p_j'  \cdot B_{2,j} \\\\
          q_2 = \tsum_{j \in \calJ} p_j' \cdot q_{2,j}
        }
        { \Dl_2 \potconf{q_2} \proc{d}{\{\calC_j' : p_j' \}_{j \in \calJ}} :: (d : B) }
      }
      { \Dl_1, \Dl_2 \potconf{q_1+q_2} \proc{C}{\{\calC_i:p_i\}_{i \in \calI}} \parallel \proc{d}{\{\calC'_j:p_j'\}_{j \in \calJ}} :: (c : A) }
    \]
    
    Because $B = \pechoice{\ell^{p_\ell'}: B_\ell}_{\ell \in L}$,
    we know that $B_{2,j} = B$ for all $j \in \calJ$.
    
    Thus, for each $j \in \calJ$, by \textsc{(T:Compose)}, we have $\Dl_1,\Dl_{2,j} \potconf{q_{1}+q_{2,j}} (\proc{c}{\{ \calC_i : p_i \}_{i \in \calI}} \parallel \calC_j' ) :: (c : A)$.
    
    Then we derive $\Dl_1,\Dl_2 \potconf{q_1+q_2} \proc{c}{\{ (\proc{c}{\{\calC_i : p_i \}_{i \in \calI}} \parallel \calC_j' ) : p_j' \}_{j \in \calJ}} :: (c : A) $ by \textsc{(T:Dist)}.
    
    Thus, by induction hypothesis, we conclude that $\Dl_1,\Dl_2 \potconf{q_1+q_2} \calC'' :: (c : A)$.
    
    \item[Case:]
    \[\small
    \Rule{C:BDist:L}
    { \proc{c}{\{ ( \calC_i \parallel \proc{d}{\{ \calC'_j : p'_j \}_{j \in \calJ}} ) : p_i \}_{i \in \calI}} \cstep{d,\pichoiceop} \calC''
    }
    { \proc{c}{\{ \calC_i : p_i \}_{i \in \calI}} \parallel \proc{d}{\{ \calC'_j : p'_j \}_{j \in \calJ}} \cstep{d,\pichoiceop} \calC''  }
    \]
    
    By inversion on the typing judgement, we have
    \[\small
      \inferrule
      { 
        \inferrule
        { \Forall{i \in \calI} \Dl_{1,i},(d:B_{1,i}) \potconf{q_{1,i}} \calC_i :: (c : A_i) \\\\
          \Dl_1 = \tsum^{\m{L}}_{i \in \calI} p_i \cdot \Dl_{1,i} \\
          B = \tsum^{\m{L}}_{i \in \calI} p_i \cdot B_{1,i} \\\\
          A = \tsum^{\m{R}}_{i \in \calI} p_i \cdot A_i \\\\
          q_1 = \tsum_{i \in \calI} p_i \cdot q_{1,i}
        }
        { \Dl_1,(d:B) \potconf{q_1} \proc{c}{\{ \calC_i : p_i \}_{i \in \calI}} :: (c : A) }
        \\
        { \Dl_2 \potconf{q_2} \proc{d}{\{\calC_j' : p_j' \}_{j \in \calJ}} :: (d : B) }
      }
      { \Dl_1, \Dl_2 \potconf{q_1+q_2} \proc{c}{\{\calC_i:p_i\}_{i \in \calI} } \parallel \proc{d}{\{\calC'_j:p_j'\}_{j \in \calJ}} :: (c : A) }
    \]
    
    Because $B = \pichoice{\ell^{p_\ell'}: B_\ell}_{\ell \in L}$,
    we know that $B_{1,i} = B$ for all $i \in \calI$.
    
    Thus, for each $i \in \calI$, by \cref{Lem:ConfExt}, we have $\Dl_{1,i},\Dl_2 \potconf{q_{1,i}+q_2} (\calC_i \parallel \proc{d}{\{ \calC_j' : p_j' \}_{j \in \calJ}}  ) :: (c : A_i)$.
    
    Then we derive $\Dl_1,\Dl_2 \potconf{q_1+q_2} \proc{c}{\{ ( \calC_i \parallel \proc{d}{\{ \calC_j' : p_j' \}_{j \in \calJ}}  ) : p_i \}_{i \in \calI}} :: (c : A)$ by \textsc{(T:Dist)}.
    
    Thus, by induction hypothesis, we conclude that $\Dl_1,\Dl_2 \potconf{q_1+q_2} \calC'' :: (c : A)$.
    
    \item[Case:]
    \[\small\Rule{C:Dist}
    { \calC_{i_0} \cstep{d,\kappa} \calC_{i_0}'
    }
    { \proc{c}{\{ \calC_i : p_i \}_{i \in \calI}} \cstep{d,\kappa} \proc{c}{\{ \calC_i : p_i \}_{i \in \calI \setminus \{i_0\}} \dplus \{ \calC_{i_0}' : p_{i_0} \}} }
    \]
    
    Appeal to induction hypothesis and \textsc{(T:Dist)}.
    
    \item[Case:]
    \[\small\Rule{C:$\pichoiceop$}
    {
    }
    { \proc{c}{w_c, \epcase{d}{\ell}{Q_\ell}_{\ell \in L} } \parallel \proc{d}{w_d, \esendlp{d}{k} \semi P} \cstep{d,\pichoiceop} \proc{c}{w_c,Q_k} \parallel \proc{d}{w_d,P} }
    \]
    
    By inversion on the typing judgment, we have
    \[\small
    \inferrule
    { \inferrule
      { \Dl_1, (d : B_k) \entailpot{q_1-w_c} Q_k :: (c : A)
      }
      { \Dl_1, (d : \pichoice{\ell^{p_\ell} : B_\ell}_{\ell \in L}) \potconf{q_1} \proc{c}{w_c, \epcase{d}{\ell}{Q_\ell}_{\ell \in L}} :: (c : A) }
      \\
      \inferrule
      { p_k = 1 \\
        p_j = 0\; (j \neq k) \\
        \Dl_2 \entailpot{q_2-w_d} P :: (d : B_k)
      }
      { \Dl_2 \potconf{q_2} \proc{d}{w_d, \esendlp{d}{k} \semi P} :: (d : \pichoice{\ell^{p_\ell} : B_\ell}_{\ell \in L}) }
    }
    { \Dl_1,\Dl_2 \potconf{q_1+q_2} (\proc{c}{w_c, \epcase{d}{\ell}{Q_\ell}_{\ell \in L} } \parallel \proc{d}{w_d, \esendlp{d}{k} \semi P}) :: (c : A) }
    \]
    
    By \textsc{(T:Proc)}, we have $\Dl_1, (d:B_k) \potconf{q_1} \proc{c}{w_c,Q_k} :: (c : A)$ and $\Dl_2 \potconf{q_2} \proc{d}{w_d,P} :: (d : B_k)$.
    
    Then we conclude by \textsc{(T:Compose)}.
    
    \item[Case:]
    \[\small\Rule{C:$\pechoiceop$}
    {
    }
    { \proc{c}{w_c, \esendlp{d}{k} \semi Q} \parallel \proc{d}{w_d, \epcase{d}{\ell}{P_\ell}_{\ell \in L} } \cstep{d,\pechoiceop} \proc{c}{w_c,Q} \parallel \proc{d}{w_d,P_k} }
    \]
    
    By inversion on the typing judgment, we have
    \[\small
    \inferrule
    { \inferrule
      { p_k = 1 \\
        p_j = 0\; (j \neq k) \\
        \Dl_1, (c : B_k) \entailpot{q_1-w_c} Q :: (d : A)
      }
      { \Dl_1, (d : \pechoice{\ell^{p_\ell} : B_\ell}_{\ell \in L}) \potconf{q_1} \proc{c}{w_c, \esendlp{d}{k} \semi Q } :: (c : A) }
      \\
      \inferrule
      { \Dl_2 \entailpot{q_2-w_d} P_k :: (d : B_k)
      }
      { \Dl_2 \potconf{q_2} \proc{d}{w_d, \epcase{d}{\ell}{P_\ell}_{\ell \in L}} :: (d : \pechoice{\ell^{p_\ell} : B_\ell}_{\ell \in L}) }
    }
    { \Dl_1,\Dl_2 \potconf{q_1+q_2} (\proc{c}{w_c, \esendlp{d}{k} \semi Q  } \parallel \proc{d}{w_d, \epcase{d}{\ell}{P_\ell}_{\ell \in L} } ):: (c : A) }
    \]
    
    By \textsc{(T:Proc)}, we have $\Dl_1, (d:B_k) \potconf{q_1} \proc{c}{w_c,Q} :: (c : A)$ and $\Dl_2 \potconf{q_2} \proc{d}{w_d,P_k} :: (d : B_k)$.
    
    Then we conclude by \textsc{(T:Compose)}.
    
    \item[Case:]
    \[\small\Rule{C:$\ichoiceop$}
    {
    }
    { \proc{c}{w_c, \ecase{d}{\ell}{Q_\ell}_{\ell \in L} } \parallel \proc{d}{w_d, \esendl{d}{k} \semi P} \cstep{d,\m{det}} \proc{c}{w_c,Q_k} \parallel \proc{d}{w_d,P} }
    \]
    
    By inversion on the typing judgment, we have
    \[\small
    \inferrule
    { \inferrule
      { \Dl_1, (d : B_k) \entailpot{q_1-w_c} Q_k :: (c : A)
      }
      { \Dl_1, (d : \ichoice{\ell : B_\ell}_{\ell \in L}) \potconf{q_1} \proc{c}{w_c, \ecase{d}{\ell}{Q_\ell}_{\ell \in L}} :: (c : A) }
      \\
      \inferrule
      { \Dl_2 \entailpot{q_2-w_d} P :: (d : B_k)
      }
      { \Dl_2 \potconf{q_2} \proc{d}{w_d, \esendl{d}{k} \semi P} :: (d : \ichoice{\ell : B_\ell}_{\ell \in L}) }
    }
    { \Dl_1,\Dl_2 \potconf{q_1+q_2} (\proc{c}{w_c, \ecase{d}{\ell}{Q_\ell}_{\ell \in L} } \parallel \proc{d}{w_d, \esendl{d}{k} \semi P}) :: (c : A) }
    \]
    
    By \textsc{(T:Proc)}, we have $\Dl_1, (d:B_k) \potconf{q_1} \proc{c}{w_c,Q_k} :: (c : A)$ and $\Dl_2 \potconf{q_2} \proc{d}{w_d,P} :: (d : B_k)$.
    
    Then we conclude by \textsc{(T:Compose)}.
    
    \item[Case:]
    \[\small\Rule{C:$\echoiceop$}
    {
    }
    { \proc{c}{w_c, \esendl{d}{k} \semi Q} \parallel \proc{d}{w_d, \ecase{d}{\ell}{P_\ell}_{\ell \in L} } \cstep{d,\m{det}} \proc{c}{w_c,Q} \parallel \proc{d}{w_d,P_k} }
    \]
    
    By inversion on the typing judgment, we have
    \[\small
    \inferrule
    { \inferrule
      { \Dl_1, (c : B_k) \entailpot{q_1-w_c} Q :: (d : A)
      }
      { \Dl_1, (d : \echoice{\ell : B_\ell}_{\ell \in L}) \potconf{q_1} \proc{c}{w_c, \esendl{d}{k} \semi Q } :: (c : A) }
      \\
      \inferrule
      { \Dl_2 \entailpot{q_2-w_d} P_k :: (d : B_k)
      }
      { \Dl_2 \potconf{q_2} \proc{d}{w_d, \ecase{d}{\ell}{P_\ell}_{\ell \in L}} :: (d : \echoice{\ell : B_\ell}_{\ell \in L}) }
    }
    { \Dl_1,\Dl_2 \potconf{q_1+q_2} (\proc{c}{w_c, \esendl{d}{k} \semi Q  } \parallel \proc{d}{w_d, \ecase{d}{\ell}{P_\ell}_{\ell \in L} } ):: (c : A) }
    \]
    
    By \textsc{(T:Proc)}, we have $\Dl_1, (d:B_k) \potconf{q_1} \proc{c}{w_c,Q} :: (c : A)$ and $\Dl_2 \potconf{q_2} \proc{d}{w_d,P_k} :: (d : B_k)$.
    
    Then we conclude by \textsc{(T:Compose)}.
    
    \item[Case:]
    \[\small
    \Rule{C:$\one$}
    {
    }
    { \proc{c}{w_c, \ewait{d} \semi Q} \parallel \proc{d}{w_d, \eclose{d}} \cstep{d,\m{det}} \proc{c}{w_c+w_d, Q} }
    \]
    
    By inversion on the typing judgement, we have
    \[\small
    \inferrule
    { \inferrule
      { \Dl \entailpot{q_1-w_c} Q :: (c : A)
      }
      { \Dl, (d:\one) \potconf{q_1} \proc{c}{w_c,\ewait{d} \semi Q} :: (c : A) }
      \\
      \inferrule
      { (\cdot) \entailpot{0} \eclose{d} :: (d : \one)
      }
      { (\cdot) \potconf{w_d} \proc{d}{w_d,\eclose{d}} :: (d : \one) }
    }
    { \Dl \potconf{q_1+w_d} (\proc{c}{w_c,\ewait{d} \semi Q} \parallel \proc{d}{w_d,\eclose{d}}) :: (c : A) }
    \]
    
    By \textsc{(T:Proc)}, we have $\Dl \potconf{q_1+w_d} \proc{c,w_c+w_d}{Q} :: (c : A)$.
    
    \item[Case:]
    \[\small
    \Rule{C:Id}
    { \proc{c}{w_c, Q\tuple{d}} \cblocked{(d,\kappa)}
    }
    { \proc{c}{w_c, Q\tuple{d}} \parallel \proc{d}{w_d, \fwd{d}{e}} \cstep{d,\kappa} \proc{c}{w_c+w_d, Q\tuple{d}[e/d]} }
    \]
    
    By inversion on the typing judgment, we have
    \[\small
    \inferrule
    { \inferrule
      { \Dl_1, (d:B) \entailpot{q_1-w_c} Q\tuple{d} :: (c : A)
      }
      { \Dl_1, (d : B) \potconf{q_1} \proc{c}{w_c, Q\tuple{d}} :: (c : A) }
      \\
      \inferrule
      { (e:B) \entailpot{0} \fwd{d}{e} :: (d : B)
      }
      { (e : B) \potconf{w_d} \proc{d}{w_d, \fwd{d}{e}} :: (d : B) }
    }
    { \Dl_1, (e : B) \potconf{q_1+w_d} (\proc{c}{w_c,Q\tuple{d})} \parallel \proc{d}{w_d, \fwd{d}{e}}) :: (c : A) }
    \]
    
    By \cref{Prop:Subst} and \textsc{(T:Proc)}, we have $\Dl_1,(e:B) \potconf{q_1+w_d} \proc{c}{w_c+w_d, Q\tuple{d}[e/d]} :: (c : A)$.
    
    \item[Case:]
    \[\small
    \Rule{C:$\paypot$}
    {
    }
    { \proc{c}{w_c, \eget{d}{r} \semi Q} \parallel \proc{d}{w_d, \epay{d}{r} \semi P} \cstep{d,\m{det}} \proc{c}{w_c, Q} \parallel \proc{d}{w_d,P} }
    \]
    
    By inversion on the typing judgment, we have
    \[\small
    \inferrule
    { \inferrule
      { \Dl_1, (d : B) \entailpot{q_1-w_c+r} Q :: (c : A)
      }
      { \Dl_1, (d : \tpaypot{B}{r}) \potconf{q_1} \proc{c}{w_c, \eget{d}{r} \semi Q} :: (c : A) }
      \\
      \inferrule
      { \Dl_2 \entailpot{q_2-w_d-r} P :: (d : B)
      }
      { \Dl_2 \potconf{q_2} \proc{d}{w_d, \epay{d}{r} \semi P} :: (d : \tpaypot{B}{r}) }
    }
    { \Dl_1,\Dl_2 \potconf{q_1+q_2} (\proc{c}{w_c, \eget{d}{r} \semi Q} \parallel \proc{d}{w_d, \epay{d}{r} \semi P}) :: (c : A)  }
    \]
    
    By \textsc{(T:Proc)}, we have $\Dl_1, (d:B) \potconf{q_1+r} \proc{c}{w_c,Q} :: (c : A)$ and $\Dl_2 \potconf{q_2-r} \proc{d}{w_d,P} :: (d : B)$.
    
    Then we conclude by \textsc{(T:Compose)}.
    
    \item[Case:]
    \[\small
    \Rule{C:$\getpot$}
    {
    }
    { \proc{c}{w_c, \epay{d}{r} \semi Q} \parallel \proc{d}{w_d, \eget{d}{r} \semi P} \cstep{d,\m{det}} \proc{c}{w_c,Q} \parallel \proc{d}{w_d,P} }
    \]
    
    By inversion on the typing judgment, we have
    \[\small
    \inferrule
    { \inferrule
      { \Dl_1, (d : B) \entailpot{q_1-w_c-r} Q :: (c : A)
      }
      { \Dl_1, (d : \tgetpot{B}{r}) \potconf{q_1} \proc{c}{w_c, \epay{d}{r} \semi Q} :: (c : A) }
      \\
      \inferrule
      { \Dl_2 \entailpot{q_2-w_d+r} P :: (d : B)
      }
      { \Dl_2 \potconf{q_2} \proc{d}{w_d, \eget{d}{r} \semi P} :: (d : \tgetpot{B}{r}) }
    }
    { \Dl_1,\Dl_2 \potconf{q_1+q_2} (\proc{c}{w_c, \epay{d}{r} \semi Q} \parallel \proc{d}{w_d, \eget{d}{r} \semi P}) :: (c : A)  }
    \]
    
    By \textsc{(T:Proc)}, we have $\Dl_1, (d : B) \potconf{q_1-r} \proc{c}{w_c,Q} :: (c : A)$ and $\Dl_2 \potconf{q_2+r} \proc{d}{w_d,P} :: (d : B)$.
    
    Then we conclude by \textsc{(T:Compose)}.
  
    \item[Case:]
    \[\small
    \Rule{C:$\tensor$}
    {
    }
    { \proc{c}{w_c, \erecvch{d}{y} \semi Q } \parallel \proc{d}{w_d, \esendch{d}{e} \semi P} \cstep{d,\m{det}} \proc{c}{w_c, Q[e/y]} \parallel \proc{d}{w_d, P} }
    \]
    
    By inversion on the typing judgmenet, we have
    \[\small
    \inferrule
    { \inferrule
      { \Dl_1, (y : B), (d : C) \entailpot{q_1-w_c} Q :: (c : A)
      }
      { \Dl_1, (d : B \tensor C) \potconf{q_1} \proc{c}{w_c, \erecvch{d}{y} \semi Q} :: (c : A) }
      \\
      \inferrule
      { \Dl_2 \entailpot{q_2-w_d} P :: (d : C)
      }
      { \Dl_2, (e : B) \potconf{q_2} \proc{d}{w_d, \esendch{d}{e} \semi P} :: (d : B \tensor C) }
    }
    { \Dl_1, \Dl_2, (e:B) \potconf{q_1+q_2} (\proc{c}{w_c, \erecvch{d}{y} \semi Q } \parallel \proc{d}{w_d, \esendch{d}{e} \semi P}) :: (c : A) }
    \]
    
    By \textsc{(T:Proc)} and \cref{Prop:Subst}, we have $\Dl_1, (e : B), (d : C) \potconf{q_1} \proc{c}{w_c, Q[e/y]} :: (c : A)$ and $\Dl_2 \potconf{q_2} \proc{d}{w_d,P} :: (d : C)$.
    
    Then we conclude by \textsc{(T:Compose)}.
    
    \item[Case:]
    \[\small
    \Rule{C:$\lolli$}
    {
    }
    { \proc{c}{w_c, \esendch{d}{e} \semi Q} \parallel \proc{d}{w_d, \erecvch{d}{y} \semi P} \cstep{d,\m{det}} \proc{c}{w_c, Q} \parallel \proc{d}{w_d, P[e/y]} }
    \]
    
    By inversion on the typing judgmenet, we have
    \[\small
    \inferrule
    { \inferrule
      { \Dl_1, (d : C) \entailpot{q_1-w_c} Q :: (c : A)
      }
      { \Dl_1, (d : B \lolli C), (e : B) \potconf{q_1} \proc{c}{w_c, \esendch{d}{e} \semi Q} :: (c : A) }
      \\
      \inferrule
      { \Dl_2, (y : B) \entailpot{q_2-w_d} P :: (d : C)
      }
      { \Dl_2 \potconf{q_2} \proc{d}{w_d, \erecvch{d}{y} \semi P} :: (d : B \lolli C) }
    }
    { \Dl_1, \Dl_2, (e:B) \potconf{q_1+q_2} (\proc{c}{w_c, \esendch{d}{e} \semi Q } \parallel \proc{d}{w_d, \erecvch{d}{y} \semi P}) :: (c : A) }
    \]
    
    By \textsc{(T:Proc)} and \cref{Prop:Subst}, we have $\Dl_1, (d:C) \potconf{q_1} \proc{c}{w_c,Q} :: (c : A)$ and $\Dl_2,(e:B) \potconf{q_2} \proc{d}{w_d, P[e/y]} :: (d : C)$.
    
    Then we conclude by \textsc{(T:Compose)}.
  \end{description}
\end{proof}

\subsection{Proof of Global Progress}
\label{appendix:global-progress}

The $\mathrm{FV}(\cdot)$ function collects \emph{free} session variables and it is defined as follows:
\begin{small}
\begin{align*}
  \mathrm{FV}^{\m{R}}(\proc{c}{w,P}) & \defeq \{ c \} \\
  \mathrm{FV}^{\m{R}}(\proc{c}{\{ \calC_i : p_i \}_{i \in \calI}}) & \defeq \bigcup_{i \in \calI} \mathrm{FV}^{\m{R}}(\calC_i) \\
  \mathrm{FV}^{\m{R}}(\calO \parallel \calC) & \defeq (\mathrm{FV}^{\m{R}}(\calO) \cup \mathrm{FV}^{\m{R}}(\calC)) \setminus \mathrm{FV}^{\m{L}}(\calO) \\
  \mathrm{FV}^{\m{L}}(\proc{c}{w,P}) & \defeq \mathrm{FV}(P) \setminus \{c \} \\
  \mathrm{FV}^{\m{L}}(\proc{c}{\{ \calC_i : p_i \}_{i \in \calI}}) & \defeq \bigcup_{i \in \calI} \mathrm{FV}^{\m{L}}(\calC_i) \\
  \mathrm{FV}^{\m{L}}(\calO \parallel \calC) & \defeq (\mathrm{FV}^{\m{L}}(\calO) \cup \mathrm{FV}^{\m{L}}(\calC)) \setminus \mathrm{FV}^{\m{R}}(\calC)
\end{align*}
\end{small}

First we prove a basic property of the $\mathrm{FV}(\cdot)$ function.

\begin{proposition}\label{Prop:FreeVariables}
  If $\Dl \potconf{q} \calC :: \Gm$, then $\mathrm{FV}^{\m{L}}(\calC) = \dom{\Dl}$ and $\mathrm{FV}^{\m{R}}(\calC) = \dom{\Gm}$.
\end{proposition}
\begin{proof}
  By induction on the derivation of $\Dl \potconf{q} \calC :: \Gm$.
\end{proof}

We can now prove that a live configuration can make a single-process execution step.

\begin{lemma}\label{Lem:GoodLive}
  If $\calC \live$, then there exists $\calC'$ such that $\calC \sstep \calC'$.
\end{lemma}
\begin{proof}
  By induction on the derivation of $\calC \live$.
  The analysis for \textsc{(L:Flip)}, \textsc{(L:Work)}, \textsc{(L:Def)} is straightforward.
  
  \begin{description}[labelindent=\parindent]
    \item[Case:]
    \[\small
    \Rule{L:Dist}
    { \calC_{i_0} \live
    }
    { \proc{c}{\{ \calC_i : p_i \}_{i \in \calI}} \live }
    \]
    
    By induction hypothesis, we know that $\calC_{i_0} \sstep \calC_{i_0}'$ for some $\calC_{i_0}'$.
    
    Then we conclude by \textsc{(E:Dist)} that $\calC \sstep \proc{c}{\{ \calC_i : p_i \}_{i \in \calI \setminus \{i_0\}} \dplus \{ \calC_{i_0}' : p_{i_0} \}}$.
    
    \item[Case:]
    \[\small
    \Rule{L:Compose:H}
    { \calO \live
    }
    { (\calO \parallel \calC') \live } 
    \]
    
    By induction hypothesis, we know that $\calO \sstep \calD$ for some $\calD$.
    
    Then $\calC \sstep \calD \parallel \calC'$ by multiset rewriting.
    
    \item[Case:]
    \[\small
    \Rule{L:Compose:T}
    { \calC' \live
    }
    { (\calO \parallel \calC') \live }
    \]
    
    By induction hypothesis, we know that $\calC' \sstep \calC''$ for some $\calC''$.
    
    Then $\calC \sstep \calO \parallel \calC''$ by multiset rewriting.
  \end{description}
\end{proof}

We prove two propositions \cref{Prop:PoisedType,Prop:BlockedType} to construct a communication execution step between a $(d,\kappa)$-poised configuration and a $(d,\kappa)$-blocked configuration (\cref{Lem:GoodCommHelper}).

\begin{proposition}\label{Prop:PoisedType}
  If $\calC \cpoised{(d,\kappa)}$ and $\Dl \potconf{q} \calC :: \Gm$,
  then $d \in \dom{\Gm}$ and $\kappa :> |\Gm(d)|$.
\end{proposition}
\begin{proof}
  By induction on the derivation of $\calC \cpoised{(d,\kappa)}$, followed by inversion on $\Dl \potconf{q} \calC :: \Gm$.
  
  \begin{description}[labelindent=\parindent]
    \item[Case:]
    \begin{mathpar}\small
      \Rule{PR:Dist}
      { \calC_{i_0} \cpoised{(d,\kappa)}
      }
      { \proc{e}{\{ \calC_i : p_i \}_{i \in \calI}} \cpoised{(d,\kappa)} }  
      \and
      \Rule{T:Dist}
      { \Forall{i \in \calI} \Dl_i \potconf{q_i} \calC_i :: (e : C_i) \\\\
        \tsum^{\m{L}}_{i \in \calI} p_i \cdot \Dl_i = \Dl \\
        \tsum^{\m{R}}_{i \in \calI} p_i \cdot C_i = C \\
        \tsum_{i \in \calI} p_i \cdot q_i = q
      }
      { \Dl \potconf{q} \proc{e}{\{ \calC_i : p_i \}_{i \in \calI}} :: (e : C) }
    \end{mathpar}
    
    In this case, $\calC = \proc{e}{\{ \calC_i : p_i \}_{i \in \calI}}$, $\Gm = (e : C)$.
    
    By induction hypothesis with $\calC_{i_0}$, we know that $d = e$ and $\kappa = |C_{i_0}| = |C|$.
    
    \item[Case:]
    \begin{mathpar}\small  
      \Rule{PR:Compose:H}
      { \calO \cpoised{(d,\kappa)}
      }
      { (\calO \parallel \calC_2) \cpoised{(d,\kappa)} }
      \and
      \Rule{T:Compose}
      { \Dl_1, \Omega \potconf{q_1} \calO :: (c : A) \\
        \Dl_2 \potconf{q_2} \calC_2 :: (\Gm',\Omega)
      }
      { \Dl_1,\Dl_2 \potconf{q_1+q_2} (\calO \parallel \calC_2) :: (\Gm',(c:A)) }
    \end{mathpar}
    
    In this case, $\calC = (\calO \parallel \calC_2)$, $\Dl = (\Dl_1,\Dl_2)$, $\Gm = (\Gm',(c:A))$, $q=q_1+q_2$.
    
    By induction hypothesis with $\calO$, we know that $d=c$ and $\kappa = |A|$.
    
    \item[Case:]
    \begin{mathpar}\small
      \Rule{PR:Compose:T}
      { \calC_2 \cpoised{(d,\kappa)} \\
        d \not\in \mathrm{FV}^{\m{L}}(\calO)
      }
      { (\calO \parallel \calC_2) \cpoised{(d,\kappa)} }  
      \and
      \Rule{T:Compose}
      { \Dl_1, \Omega \potconf{q_1} \calO :: (c : A) \\
        \Dl_2 \potconf{q_2} \calC_2 :: (\Gm',\Omega)
      }
      { \Dl_1,\Dl_2 \potconf{q_1+q_2} (\calO \parallel \calC_2) :: (\Gm',(c:A)) }
    \end{mathpar}
    
    In this case, $\calC = (\calO \parallel \calC_2)$, $\Dl = (\Dl_1,\Dl_2)$, $\Gm = (\Gm',(c:A))$, $q=q_1+q_2$.
    
    By induction hypothesis with $\calC_2$, we know that $d \in \dom{\Gm',\Omega}$ and $\kappa = |(\Gm',\Omega)(d)|$.
    
    Because $d \not\in \mathrm{FV}^{\m{L}}(\calO)$, we have $d \not\in \dom{\Omega}$ by \cref{Prop:FreeVariables}. Thus $d \in \dom{\Gm'}$.
  \end{description}
\end{proof}

\begin{proposition}\label{Prop:BlockedType}
  If $\calC \cblocked{(d,\kappa)}$ and $\Dl \potconf{q} \calC :: \Gm$, then $d \in \dom{\Dl}$ and $\kappa = |\Dl(d)|$.
\end{proposition}
\begin{proof}
  By induction on the derivation of $\calC \cblocked{(d,\kappa)}$, followed by inversion on $\Dl \potconf{q} \calC :: \Gm$.
\end{proof}

\begin{lemma}\label{Lem:GoodCommHelper}
  If $\calC = (\calC_1 \parallel \calC_2)$ such that $\calC_1 \cblocked{(d,\kappa)}$, $\calC_2 \cpoised{(d,\kappa)}$ (or $\calC_2 \cpoised{(d,\top)}$), $\Dl_1,\Dl',(d:B) \potconf{q_1} \calC_1 :: \Gm_1$, $\Dl_2 \potconf{q_2} \calC_2 :: (\Gm_2,\Dl',(d:B))$, and $\kappa = |B|$,
  then $\calC \cstep{d,\kappa} \calC'$ for some $\calC'$.
\end{lemma}
\begin{proof}
  By nested induction on $\calC_1 \cblocked{(d,\kappa)}$ and $\calC_2 \cpoised{(d,\kappa)}$, followed by inversion on the typing judgments.
  
  \begin{description}[labelindent=\parindent]
    \item[Case:]
    \begin{mathpar}\small
    \Rule{BL:Dist}
    { \calC_{1,i_0} \cblocked{(d,\kappa)}
    }
    { \proc{c}{ \{ \calC_{1, i} : p_i \}_{i \in \calI}} \cblocked{(d,\kappa)} }
    \and
    \Rule{T:Dist}
    { \Forall{i \in \calI} \Dl_{1,i}, (d : B_i) \potconf{q_{1,i}} \calC_{1,i} :: (c : A_i) \\\\
      \tsum^{\m{L}}_{i \in \calI} p_i \cdot \Dl_{1,i} = \Dl_1 \\
      \tsum^{\m{L}}_{i \in \calI} p_i \cdot B_i = B \\\\
      \tsum^{\m{R}}_{i \in \calI} p_i \cdot A_i = A \\
      \tsum_{i \in \calI} p_i \cdot q_{1,i} = q_1
    }
    { \Dl_1, (d:B) \potconf{q_1} \proc{c}{\{ \calC_{1,i} : p_i \}_{i \in \calI}} :: (c : A) }
    \\
    \Rule{PR:Dist}
    { \calC_{2,j_0}' \cpoised{(d,\kappa)}
    }
    { \proc{d}{ \{  \calC_{2,j}' : p_j' \}_{j \in \calJ}} \cpoised{(d,\kappa)} }    
    \and
    \Rule{T:Dist}
    { \Forall{j \in \calJ} \Dl_{2,j} \potconf{q_{2,j}} \calC'_{2,j} :: (d : B_j') \\\\
      \tsum^{\m{L}}_{j \in \calJ} p_j' \cdot \Dl_{2,j} = \Dl_2 \\\\
      \tsum^{\m{R}}_{j \in \calJ} p_j' \cdot B_j' = B \\
      \tsum_{j \in \calJ} p_j' \cdot q_{2,j} = q_2
    }
    { \Dl_2 \potconf{q_2} \proc{d}{\{ \calC'_{2,j} : p'_j \}_{j \in \calJ}} :: (d : B) }
    \end{mathpar}
    
    In this case, $\calC_1 = \proc{c}{\{\calC_{1,i}\}_{i \in \calI}}$, $\calC_2 = \proc{d}{\{ \calC'_{2,j} : p'_j \}_{j \in \calJ}}$, $\Dl' = (\cdot)$, $\Gm_1 = (c : A)$, $\Gm_2 = (\cdot)$.
    
    By case analysis on $\kappa$.
    
    \begin{itemize}
      \item If $\kappa = \m{det}$, then we know that $B \neq \pichoice{\cdots}$ and $B \neq \pechoice{\cdots}$.
      
      Thus, $B_i = B$ for $i \in \calI$ and $B_j' = B$ for $j \in \calJ$.
      
      Then by induction hypothesis with $\calC_{1,i_0},\calC_{2,j_0}'$, we know that $(\calC_{1,i_0} \parallel \calC'_{2,j_0}) \cstep{d,\m{det}} \calD$ for some $\calD$.
      
      By \textsc{(C:Dist)} and \textsc{(C:BDist:D)}, we conclude that $\calC \cstep{d,\m{det}} \proc{c}{ \{ ( \calC_{1,i} \parallel \calC'_{2,j} ) : p_i \cdot p_j' \}_{i \in \calI, j \in \calJ : i \neq i_0 \vee j \neq j_0} \dplus \{ \calD : p_{i_0} \cdot p'_{j_0} \}}$.
    
      \item If $\kappa = \pichoiceop$, then we know that $B = \pichoice{\cdots}$.
      
      Thus, $B_i = B$ for all $i \in \calI$.
      
      Then by induction hypothesis with $\calC_{1,i_0},\calC_2'$, we know that $(\calC_{1,i_0} \parallel \calC_2) \cstep{d,\pichoiceop} \calD$ for some $\calD$.
      
      By \textsc{(C:Dist)} and \textsc{(C:BDist:L)}, we conclude that $\calC \cstep{d,\pichoiceop} \proc{c}{\{ (\calC_{1,i} \parallel \calC_2) : p_i \}_{i \in \calI \setminus \{i_0\}} \dplus \{ \calD : p_{i_0} \}}$.
            
      \item If $\kappa = \pechoiceop$, then we know that $B = \pechoice{\cdots}$.
      
      Thus, $B_j' = B$ for all $j \in \calJ$.
      
      Then by induction hypothesis with $\calC_1,\calC'_{2,j_0}$, we know that $(\calC_1 \parallel \calC_{2,j_0}') \cstep{d,\pechoiceop} \calD$ for some $\calD$.
      
      By \textsc{(C:Dist)} and \textsc{(C:BDist:R)}, we conclude that $\calC \cstep{d,\pechoiceop} \proc{c}{\{ ( \calC_1 \parallel \calC'_{2,j} )  :p'_j\}_{j \in \calJ \setminus \{j_0\}} \dplus \{ \calD : p'_{j_0} \}}$.
    \end{itemize}
    
    \item[Case:]
    \begin{mathpar}\small
      \Rule{BL:Compose:H}
      { \calO_1 \cblocked{(d,\kappa)} \\
        d \not\in \mathrm{FV}^{\m{R}}(\calC_1')
      }
      { (\calO_1 \parallel \calC_1') \cblocked{(d,\kappa)} }
      \and
      \Rule{T:Compose}
      { \Dl_{11}, \Dl_1', \Omega, (d:B) \potconf{q_{11}} \calO_1 :: (c : A) \\
        \Dl_{12}, \Dl_2' \potconf{q_{12}} \calC_1' :: (\Gm'_1,\Omega)
      }
      { \Dl_{11},\Dl_{12},\Dl'_1,\Dl'_2, (d:B) \potconf{q_{11}+q_{12}} (\calO_1 \parallel \calC_1') :: (\Gm'_1,(c:A)) }
    \end{mathpar}
    
    In this case, $\calC_1 = (\calO_1 \parallel \calC_1')$, $\Dl_1 = (\Dl_{11},\Dl_{12})$, $\Dl' = (\Dl'_1,\Dl'_2)$, $\Gm_1 = (\Gm_1', (c:A))$, $q_1 = q_{11}+q_{12}$.
    
    By \cref{Prop:BlockedType} with $\calO_1$, \cref{Prop:FreeVariables}, and $d \not\in \mathrm{FV}^{\m{R}}(\calC_1')$, $d$ can only be placed in the way showed above.
    
    By induction hypothesis with $\calO_1,\calC_2$, we know that $(\calO_1 \parallel \calC_2) \cstep{d,\kappa} \calD$ for some $\calD$.
    
    Then $\calC \cstep{d,\kappa} \calC_1' \parallel \calD$ by multiset rewriting.
    
    \item[Case:]
    \begin{mathpar}\small
      \Rule{BL:Compose:T}
      { \calC_1' \cblocked{(d,\kappa)}
      }
      { (\calO_1 \parallel \calC_1') \cblocked{(d,\kappa)} }
      \and
      \Rule{T:Compose}
      { \Dl_{11},\Dl_1', \Omega \potconf{q_{11}} \calO_1 :: (c : A) \\
        \Dl_{12}, \Dl_2', (d:B) \potconf{q_{12}} \calC_1' :: (\Gm_1',\Omega)
      }
      { \Dl_{11},\Dl_{12},\Dl_1',\Dl_2',(d:B) \potconf{q_{11}+q_{12}} (\calO_1 \parallel \calC_1') :: (\Gm_1',(c:A)) }
    \end{mathpar}
    
    In this case, $\calC_1 = (\calO_1 \parallel \calC_1')$, $\Dl_1 = (\Dl_{11},\Dl_{12})$, $\Dl' = (\Dl_1',\Dl_2')$, $\Gm_1 = (\Gm_1',(c:A))$, $q_1 = q_{11}+q_{12}$.
    
    By \cref{Prop:BlockedType} with $\calC_1'$, $d$ can only be placed in the way showed above.
    
    By induction hypothesis with $\calC_1',\calC_2$, we know that $(\calC_1' \parallel \calC_2) \cstep{d,\kappa} \calD$ for some $\calD$.
    
    Then $\calC \cstep{d,\kappa} \calO_1 \parallel \calD$ by multiset rewriting.
    
    \item[Case:]
    \begin{mathpar}\small
      \Rule{BL:$\cdots$}
      { \cdots
      }
      { \proc{c}{w,Q\tuple{d}} \cblocked{(d,\kappa)} }
      \and
      \Rule{T:$\cdots$}
      { \cdots
      }
      { \Dl_1, \Dl',(d:B) \potconf{q_1} \proc{c}{w,Q\tuple{d}} :: (c : A) }
      \\
      \Rule{PR:Compose:H}
      { \calO_2 \cpoised{(d,\kappa)}
      }
      { (\calO_2 \parallel \calC_2') \cpoised{(d,\kappa)} }
      \and
      \Rule{T:Compose}
      { \Dl_{21}, \Omega \potconf{q_{21}} \calO_2 :: (d : B) \\
        \Dl_{22} \potconf{q_{22}} \calC_2' :: (\Gm_2,\Dl',\Omega)
      }
      { \Dl_{21},\Dl_{22} \potconf{q_{21}+q_{22}} (\calO_2 \parallel \calC_2') :: (\Gm_2, \Dl', (d : B)) }
    \end{mathpar}
    
    In this case, $\calC_1 = \proc{c}{w,Q\tuple{d}}$, $\calC_2 = (\calO_2 \parallel \calC_2')$, $\Gm_1 = (c : A)$, $\Dl_2 = (\Dl_{21},\Dl_{22})$, $q_2 = q_{21}+q_{22}$.
    
    By \cref{Prop:PoisedType} with $\calO_2$, $d$ can only be placed in the way showed above.
    
    By induction hypothesis with $\calC_1, \calO_2$, we know that $(\calC_1 \parallel \calO_2) \cstep{d,\kappa} \calD$ for some $\calD$.
    
    Then $\calC \cstep{d,\kappa} \calD \parallel \calC_2'$ by multiset rewriting.
    
    \item[Case:]
    
    We consider the case where $\calC_1$ is a leaf and $\calC_2$ is a sequence of objects.
    
    \begin{itemize}
      \item
      \begin{mathpar}\small
      \Rule{BL:$\cdots$}
      { \cdots
      }
      { \proc{c}{w,Q\tuple{d}} \cblocked{(d,\kappa)} }
      \and
      \Rule{T:$\cdots$}
      { \cdots
      }
      { \Dl'', (d : B), (e : C) \potconf{q_1} \proc{c}{w,Q\tuple{d}} :: (c : A) }
      \\
      \Rule{PR:Compose:T}
      { \calC_2' \cpoised{(d,\kappa)} \\
        d \not\in \mathrm{FV}^{\m{L}}(\calO_2)
      }
      { (\calO_2 \parallel \calC_2') \cpoised{(d,\kappa)} }
      \and
      \Rule{T:Compose}
      { \Dl_{21}, \Omega \potconf{q_{21}} \calO_2 :: (e : C) \\
        \Dl_{22} \potconf{q_{22}} \calC_2' :: (\Gm_2,\Dl'',(d:B),\Omega)
      }
      { \Dl_{21},\Dl_{22} \potconf{q_{21}+q_{22}} (\calO_2 \parallel \calC_2') :: (\Gm_2, \Dl'', (d : B), (e:C)) }
      \end{mathpar}
    
      In this case, $\calC_1 = \proc{c}{w,Q\tuple{d}}$, $\calC_2 = (\calO_2 \parallel \calC_2')$, $\Dl' = (\Dl'', (e:C))$, $\Gm_1 = (c:A)$, $\Dl_2 = (\Dl_{21},\Dl_{22})$, $q_2 = q_{21}+q_{22}$.
      
      By \cref{Prop:PoisedType} with $\calC_2'$, \cref{Prop:FreeVariables}, and $d \not\in \mathrm{FV}^{\m{L}}(\calO_2)$, $d$ can only be placed in the way showed above.
      
      By induction hypothesis with $\calC_1,\calC_2'$, we know that $(\calC_1 \parallel \calC_2') \cstep{d,\kappa} \calD$ for some $\calD$.
      
      Then $\calC \cstep{d,\kappa} \calD \parallel \calO_2$ by multiset rewriting.
      
      \item
      \begin{mathpar}\small
      \Rule{BL:$\cdots$}
      { \cdots
      }
      { \proc{c}{w,Q\tuple{d}} \cblocked{(d,\kappa)} }
      \and
      \Rule{T:$\cdots$}
      { \cdots
      }
      { \Dl', (d : B) \potconf{q_1} \proc{c}{w,Q\tuple{d}} :: (c : A) }
      \\
      \Rule{PR:Compose:T}
      { \calC_2' \cpoised{(d,\kappa)} \\
        d \not\in \mathrm{FV}^{\m{L}}(\calO_2)
      }
      { (\calO_2 \parallel \calC_2') \cpoised{(d,\kappa)} }
      \and
      \Rule{T:Compose}
      { \Dl_{21}, \Omega \potconf{q_{21}} \calO_2 :: (e : C) \\
        \Dl_{22} \potconf{q_{22}} \calC_2' :: (\Gm_2',\Dl'',(d:B),\Omega)
      }
      { \Dl_{21},\Dl_{22} \potconf{q_{21}+q_{22}} (\calO_2 \parallel \calC_2') :: (\Gm_2', \Dl', (d : B), (e:C)) }
      \end{mathpar}
    
      In this case, $\calC_1 = \proc{c}{w,Q\tuple{d}}$, $\calC_2 = (\calO_2 \parallel \calC_2')$, $\Gm_1 = (c:A)$, $\Dl_2 = (\Dl_{21},\Dl_{22})$, $\Gm_2 = (\Gm_2',(e:C))$, $q_2 = q_{21}+q_{22}$.
      
      By \cref{Prop:PoisedType} with $\calC_2'$, \cref{Prop:FreeVariables}, and $d \not\in \mathrm{FV}^{\m{L}}(\calO_2)$, $d$ can only be placed in the way showed above.
      
      By induction hypothesis with $\calC_1, \calC_2'$, we know that $(\calC_1 \parallel \calC_2') \cstep{d,\kappa} \calD$ for some $\calD$.
      
      Then $\calC \cstep{d,\kappa} \calD \parallel \calO_2$ by multiset rewriting.
    
    \end{itemize}

  \end{description}
\end{proof}

As a direct corollary of \cref{Lem:GoodCommHelper}, the lemma bellow states the result with the $(d,\kappa)$-comm relation.

\begin{lemma}\label{Lem:GoodComm}
  If $(\cdot) \potconf{q} \calC :: \Gm$ and $\calC \comm{(d,\kappa)}$,
  then $C \cstep{d,\kappa} \calC'$ for some $\calC'$.
\end{lemma}
\begin{proof}
  By induction on the derivation of $\calC \comm{(d,\kappa)}$, followed by inversion on $(\cdot) \potconf{q} \calC :: \Gm$.
  
  \begin{description}[labelindent=\parindent]
    \item[Case:]
    \begin{mathpar}\small
    \Rule{CM:Dist}
    { \calC_{i_0} \comm{(d,\kappa)}
    }
    { \proc{e}{\{ \calC_i : p_i \}_{i \in \calI}} \comm{(d,\kappa)} }
    \and
    \Rule{T:Dist}
    { \Forall{i \in \calI} (\cdot) \potconf{q_i} \calC_i :: (e : A_i) \\
      \tsum^{\m{R}}_{i \in \calI} p_i \cdot A_i = A \\
      \tsum_{i \in \calI} p_i \cdot q_i = q
    }
    { (\cdot) \potconf{q} \proc{e}{\{ \calC_i : p_i \}_{i \in \calI}} :: (e : A) }
    \end{mathpar}
    
    By induction hypothesis, we know that $\calC_{i_0} \cstep{d,\kappa} \calC'_{i_0}$ for some $\calC'_{i_0}$.
    
    Then we conclude by \textsc{(C:Dist)} that $\calC \cstep{d,\kappa} \proc{e}{\{ \calC_i : p_i \}_{i \in \calI \setminus \{i_0\}} \dplus \{ \calC_{i_0}' : p_{i_0} \} }$.
        
    \item[Case:]
    \begin{mathpar}\small
    \Rule{CM:Compose:C}
    { \calO \cblocked{(d,\kappa)} \\
      \calC_2 \cpoised{(d,\kappa)}
    } 
    { (\calO \parallel \calC_2) \comm{(d,\kappa)} }
    \and
    \Rule{T:Compose}
    { \Dl' \potconf{q_1} \calO :: (c : A) \\
      (\cdot) \potconf{q_2} \calC_2 :: (\Gm, \Dl')
    }
    { (\cdot) \potconf{q_1+q_2} (\calO \parallel \calC_2) :: (\Gm,(c: A)) }
    \end{mathpar}
  
    By \cref{Prop:BlockedType} with $\calO$, we know that $d \in \dom{\Dl'}$ and $\kappa = |\Dl'(d)|$.
    Then appeal to \cref{Lem:GoodCommHelper}.
  \end{description}
\end{proof}

Now, we prove \cref{Prop:PoisedExistence,Lem:ProgressHelper} to bridge the gap between well-typed configurations and the status-characterizing relations.

\begin{proposition}\label{Prop:PoisedExistence}
  If $\Dl \potconf{q} \calC :: \Gm$ and $\calC \poised$, then for all $d$ in $\dom{\Gm}$, it holds that $\calC \cpoised{(d,|\Gm(d)|)}$ (or $\calC \cpoised{(d,\top)}$).
\end{proposition}
\begin{proof}
  By induction on the derivation of $\Dl \potconf{q} \calC :: \Gm$, followed by inversion on $\calC \poised$.
\end{proof}

\begin{lemma}\label{Lem:ProgressHelper}
  If $\Dl \potconf{q} \calC :: \Gm$, then at least one of the cases below holds:
  \begin{enumerate}[(i)]
    \item $\calC \live$,
    \item $\calC \comm{(d,\kappa)}$ for some $d,\kappa$, or
    \item $\calC \cblocked{(d,|\Dl(d)|)}$ for some $d \in \dom{\Dl}$, or
    \item $\calC \poised$.
  \end{enumerate}  
\end{lemma}
\begin{proof}
  By induction on the derivation of $\Dl \potconf{q} \calC :: \Gm$.
  
  \begin{description}[labelindent=\parindent]
    \item[Case:]
    \[\small
    \Rule{T:Proc}
    { \Dl \entailpot{q-w} P :: (c : A)
    }
    { \Dl \potconf{q} \proc{c}{w,P} :: (c : A) }
    \]
    
    By a case analysis on $P$, we can conclude that either $\proc{c}{w,P} \live$, $\proc{c}{w,P} \cblocked{(d,|\Dl(d)|)}$ for some $d \in \dom{\Dl}$, or $\proc{c}{w,P} \poised$.
    
    \item[Case:]
    \[\small
    \Rule{T:Dist}
    { \Forall{i \in \calI} \Dl_i \potconf{q_i} \calC_i :: (c : A_i) \\
      \tsum^{\m{L}}_{i \in \calI} p_i \cdot \Dl_i = \Dl \\
      \tsum^{\m{R}}_{i \in \calI} p_i \cdot A_i = A \\
      \tsum_{i \in \calI} p_i \cdot q_i = q
    }
    { \Dl \potconf{q} \proc{c}{\{ \calC_i : p_i \}_{i \in \calI}} :: (c : A) }
    \]
    
    We can apply induction hypothesis on all $\calC_i$'s.
    If all $\calC_i$'s are poised, then $\calC$ itself is poised by \textsc{(P:Dist)}.
    
    Otherwise, there exists $i_0$ such that $\calC_{i_0}$ is not poised.
    \begin{itemize}
      \item If $\calC_{i_0}\live$, then $\calC \live$ by \textsc{(L:Dist)}.
      \item If $\calC_{i_0} \comm{(d,\kappa)}$ for some $d,\kappa$, then $\calC \comm{(d,\kappa)}$ by \textsc{(CM:Dist)}.
      \item If $\calC_{i_0} \cblocked{(d,|\Dl_{i_0}(d)|)}$ for some $d \in \dom{\Dl_{i_0}}$, then $\calC \cblocked{(d, |\Dl(d)|)}$ by \textsc{(BL:Dist)}.
    \end{itemize}
    
    \item[Case:]
    \[\small
    \Rule{T:Compose}
    { \Dl_1, \Dl' \potconf{q_1} \calO :: (c : A)
      \\
      \Dl_2 \potconf{q_2} \calC_2 :: (\Gm,\Dl')
    }
    { \Dl_1,\Dl_2 \potconf{q_1+q_2} (\calO \parallel \calC_2) :: (\Gm,(c:A)) }
    \]
    
    We can apply induction hypothesis on $\calO$ and $\calC_2$.
    
    \begin{itemize}
      \item If $\calO \live$ or $\calC_2 \live$, then $\calC \live$ by \textsc{(L:Compose:H)} or \textsc{(L:Compose:T)}.
      \item If $\calO \comm{(d,\kappa)}$ or $\calC_2 \comm{(d,\kappa)}$ for some $d,\kappa$, then $\calC \comm{(d,\kappa)}$ by \textsc{(CM:Compose:H)} or \textsc{(CM:Compose:T)}.
      \item If $\calC_2 \cblocked{(d,|\Dl_2(d)|)}$ for some $d \in \dom{\Dl_2}$, then $\calC \cblocked{(d,|\Dl_2(d)|)}$ by \textsc{(BL:Compose:T)}.
      \item If $\calO \cblocked{(d,|\Dl_1(d)|)}$ for some $d \in \dom{\Dl_1}$, then $\calC \cblocked{(d,|\Dl_1(d)|)}$ by \textsc{(BL:Compose:H)} and $d \not\in \mathrm{FV}^{\m{R}}(\calC_2)$ (because $d \not\in \dom{\Gm,\Dl'}$ and \cref{Prop:FreeVariables}).
      \item If $\calO \cblocked{(d,|\Dl'(d)|)}$ for some $d \in \dom{\Dl'}$ and $\calC_2 \poised$, then by \cref{Prop:PoisedExistence}, we have $\calC_2 \cpoised{(d,|\Dl'(d)|)}$, thus we conclude by \textsc{(CM:Compose:C)} that $\calC \comm{(d,|\Dl'(d)|)}$.
      \item If both $\calO$ and $\calC_2$ are poised, then $\calC$ itself is poised by \textsc{(P:Compose)}.
    \end{itemize}    
  \end{description}
\end{proof}

Finally, we can formulate and prove \emph{global progress} of this type system.

\begin{lemma}\label{Lem:Progress}
  If $(\cdot) \potconf{q} \calC :: \Gm$, then at least one of the cases below holds:
  \begin{enumerate}[(i)]
    \item $\calC \live$,
    \item $\calC \comm{(d,\kappa)}$ for some $d,\kappa$, or
    \item $\calC \poised$.
  \end{enumerate}
\end{lemma}
\begin{proof}
  Appeal to \cref{Lem:ProgressHelper}.
\end{proof}

\begin{theorem}[Global progress]\label{The:Progress}
  If $ (\cdot) \potconf{q} \calC :: \Gm$, then either
  \begin{enumerate}[(i)]
    \item $\calC \sstep \calC'$ for some $\calC'$, or $\calC \cstep{d,\kappa} \calC'$ for some $\calC',d,\kappa$, or
    \item $\calC \poised$.
  \end{enumerate}
\end{theorem}
\begin{proof}
  Appeal to \cref{Lem:Progress,Lem:GoodLive,Lem:GoodComm}.
\end{proof}

\subsection{Expected Work Analysis}
\label{appendix:expected-work}

Non-nested configurations are defined by
\[
\calD \Coloneqq \proc{c_1}{w_1,P_1} \parallel \cdots \parallel \proc{c_n}{w_n,P_n}
\]

We develop a \emph{distribution}-based~\cite{ICFP:BLG16,JCSS:Kozen81} small-step operational semantics on non-nested configurations, based on a \emph{synchronous} semantics for resource-aware session types~\cite{Das18RAST,Balzer17ICFP}.
The semantics should also be seen as a collection of \emph{multiset-rewriting rules}~\cite{Cervesato09SEM}.
The table below lists the rules for this semantics.

\begin{center}
\begin{footnotesize}
  \begin{tabular}{ll}
  (\textsc{S:$\pichoiceop$}) & $\proc{c}{w_c, \esendlp{c}{k} \semi P}, \proc{d}{w_d,\epcase{c}{\ell}{Q_\ell}_{\ell \in L}} \ostep \proc{c}{w_c, P}, \proc{d}{w_d,Q_k}$ \\
  (\textsc{S:$\pechoiceop$}) & $\proc{c}{w_c,\epcase{c}{\ell}{P_\ell}_{\ell \in L}}, \proc{d}{w_d,\esendlp{c}{k} \semi Q} \ostep \proc{c}{w_c,P_k}, \proc{d}{w_d,Q}$ \\
  (\textsc{S:$\one$}) & $\proc{c}{w_c,\eclose{c}}, \proc{d}{w_d,\ewait{c} \semi Q} \ostep \proc{d}{w_c+w_d,Q}$ \\
  (\textsc{S:Id}) & $\proc{c}{w_c, \fwd{c}{e}}, \proc{d}{w_d,Q\tuple{c}} \ostep \proc{d}{w_c+w_d, Q\tuple{c}[e/c]}$ \\
  (\textsc{S:Def}) & $\proc{c}{w, \ecut{x}{f}{\many{d}}{Q}} \ostep \proc{b}{0, P_f[b/x',\many{d}/\many{y'}]}, \proc{c}{w,Q[b/x]}$ \\
  & \quad for $\many{y' : B} \entailpot{q} f = P_f :: (x' : A) \in \Sg$ and $b$ \emph{fresh} \\
  (\textsc{S:$\tensor$}) & $\proc{c}{w_c, \esendch{c}{e} \semi P}, \proc{d}{w_d, \erecvch{c}{y} \semi Q} \step \proc{c}{w_c,P}, \proc{d}{w_d, Q[e/y]}$ \\
  (\textsc{S:$\lolli$}) & $\proc{c}{w_c, \erecvch{c}{y} \semi P}, \proc{d}{w_d,\esendch{c}{e} \semi Q} \step \proc{c}{w_c,P[e/y]}, \proc{d}{w_d,Q}$ \\
  (\textsc{S:$\ichoiceop$}) & $\proc{c}{w_c,\esendl{c}{k} \semi P}, \proc{d}{w_d, \ecase{c}{\ell}{Q_\ell}_{\ell \in L}} \step \proc{c}{w_c,P}, \proc{d}{w_d,Q_k}$ \\
  (\textsc{S:$\echoiceop$}) & $\proc{c}{w_c, \ecase{c}{\ell}{P_\ell}_{\ell \in L}},\proc{d}{w_d, \esendl{c}{k} \semi Q} \step \proc{c}{w_c,P_k}, \proc{d}{w_d,Q} $ \\
  (\textsc{S:$\paypot$}) & $\proc{c}{w_c,\epay{c}{r} \semi P}, \proc{d}{w_d,\eget{c}{r} \semi Q} \ostep \proc{c}{w_c,P}, \proc{d}{w_d,Q}$ \\
  (\textsc{S:$\getpot$}) & $\proc{c}{w_c, \eget{c}{r} \semi P}, \proc{d}{w_d,\epay{c}{r} \semi Q} \ostep \proc{c}{w_c,P},\proc{d}{w_d,Q}$ \\
  (\textsc{S:Work}) & $\proc{c}{w,\ework{r} \semi P} \ostep \proc{c}{w+r,P}$ \\
  \\
  (\textsc{SP:Det}) & $\calD \pstep \delta(\calD')$ for $\calD \ostep \calD'$ and $\delta(\cdot)$ is the \emph{Dirac} distribution \\
  (\textsc{SP:Flip}) & $\proc{c}{w_c,\eflip{p}{Q_H}{Q_T}} \pstep p \cdot \delta(\proc{c}{w_c,Q_H})+(1-p)\cdot \delta(\proc{c}{w_c,Q_T})$ \\
  \end{tabular}
\end{footnotesize}
\end{center}

We now lift the configuration-to-distribution relation $\pstep$ to a distribution-to-distribution relation $\dstep$ as follows:
\begin{mathpar}\small
  \inferrule
  { \mu = \{ \calD_i : p_i \}_{i \in \calI} \\
    \Exists{i_0 \in \calI} \calD_{i_0} \pstep \mu_{i_0}'
  }
  { \mu \dstep \{ \calD_i : p_i \}_{i \in \calI \setminus \{i_0\}} + p_{i_0} \cdot \mu_{i_0}' }
  \and
  \inferrule
  { \mu = \{ \calD_i : p_i \}_{i \in \calI} \\
    \Forall{i \in \calI} \calD_i \poised
  }
  { \mu \dstep \mu }  
\end{mathpar}

We first prove the potential is still an upper bound on the expected work after we ``flatten'' a multiverse semantic object.

\begin{lemma}\label{Lem:SimulationInvariant}
  If $\calC \approx \mu$, $\Dl \potconf{q} \calC :: \Gm$, then $\bbE_{\calD \sim \mu}[\m{work}(\calD)] \le q$, where $\m{work}(\many{\proc{c_i}{w_i,P_i}}) \defeq \tsum_i w_i$.
\end{lemma}
\begin{proof}
  By induction on the derivation of $C \approx \mu$, followed by inversion on $\Dl \potconf{q} \calC :: \Gm$.
  \begin{description}[labelindent=\parindent]
    \item[Case:]
    \begin{mathpar}\small
    \Rule{FL:Proc}
    {
    }
    { \proc{c}{w,P} \approx \delta(\proc{c}{w,P}) }
    \and
    \Rule{T:Proc}
    { \Dl \entailpot{q-w} P :: (c : A)
    }
    { \Dl \potconf{q} \proc{c}{w,P} :: (c : A) }
    \end{mathpar}
    
    By $q - w \ge 0$, we conclude that $q \ge w = \bbE_{\calD \sim \delta(\proc{c}{w,P})}[\m{work}(\calD)]$.
    
    \item[Case:]
    \begin{mathpar}\small
    \Rule{FL:Dist}
    { \Forall{i \in \calI} \calC_i \approx \mu_i
    }
    { \proc{c}{\{\calC_i : p_i\}_{ i \in \calI}} \approx \tsum_{i \in \calI} p_i \cdot \mu_i }
    \and
    \Rule{T:Dist}
    { \Forall{i \in \calI} \Dl_i \potconf{q_i} \calC_i :: (c : A_i) \\
      \tsum^{\m{L}}_{i \in \calI} p_i \cdot \Dl_i = \Dl \\\\
      \tsum^{\m{R}}_{i \in \calI} p_i \cdot A_i = A \\
      \tsum_{i \in \calI} p_i \cdot q_i = q
    }
    { \Dl \potconf{q} \proc{c}{\{\calC_i : p_i\}_{i \in \calI}} :: (c : A) }
    \end{mathpar}
    
    By induction hypothesis, for each $i \in \calI$, we have $\bbE_{\calD \sim \mu_i}[\m{work}(\calD)] \le q_i$.
    
    Thus
    \begin{small}
    \begin{align*}
      \bbE_{\calD \sim \tsum_{i \in \calI} p_i \cdot \mu_i}[\m{work}(\calD)] & = \tsum_{i \in \calI} p_i \cdot \bbE_{\calD \sim \mu_i}[\m{work}(\calD)] \\
      & \le \tsum_{i \in \calI} p_i \cdot q_i \\
      & = q.
    \end{align*} 
    \end{small}
    
    \item[Case:]
    \begin{mathpar}\small
    \Rule{FL:Compose}
    { \calO \approx \mu_1 \\
      \calC \approx \mu_2
    }
    { (\calO \parallel \calC) \approx (\mu_1 \otimes \mu_2) \bind \lambda(\calD_1,\calD_2). \delta(\calD_1 \parallel \calD_2) }  
    \and
    \Rule{T:Compose}
    { \Dl_1,\Dl' \potconf{q_1} \calO :: (c : A) \\
      \Dl_2 \potconf{q_2} \calC :: (\Gm, \Dl')
    }
    { \Dl_1,\Dl_2 \potconf{q_1+q_2} (\calO \parallel \calC) :: (\Gm,(c:A)) }
    \end{mathpar}
    
    By induction hypothesis, we have $\bbE_{\calD \sim \mu_1}[\m{work}(\calD)] \le q_1$ and $\bbE_{\calD \sim \mu_2}[\m{work}(\calD)] \le q_2$.

    Let $\mu \defeq (\mu_1 \otimes \mu_2) \bind \lambda(\calD_1,\calD_2). \delta(\calD_1 \parallel \calD_2)$. Thus
    \begin{small}
    \begin{align*}
      \bbE_{\calD \sim \mu}[\m{work}(\calD)] & = \bbE_{\calD_1 \sim \mu_1, \calD_2 \sim \mu_2}[\m{work}(\calD_1) + \m{work}(\calD_2)] \\
      & = \bbE_{\calD_1 \sim \mu_1}[\m{work}(\calD_1)] + \bbE_{\calD_2 \sim \mu_2}[\m{work}(\calD_2)] \\
      & \le q_1 + q_2 \\
      & = q.
    \end{align*}
    \end{small}
  \end{description}
\end{proof}

Then we prove that the simulation relation preserves the evaluation relation.

\begin{lemma}[Simulation]\label{Lem:SimulationStep}
  If $\calC \approx \mu$, $\calC' \approx \mu'$, and $\calC \sstep \calC'$ or $\calC \cstep{d,\kappa} \calC'$,
  then $\mu \dstep \mu'$.
\end{lemma}
\begin{proof}
  By induction on the derivation of $\calC \sstep \calC'$ or $\calC \cstep{d,\kappa} \calC'$.
  \begin{description}[labelindent=\parindent]
    \item[Case:]
    \[\small\Rule{E:Flip}
    {
    }
    { \proc{c}{w,\eflip{p}{P_H}{P_T}} \sstep \proc{c}{\{ \proc{c}{w,P_H} : p, \proc{c}{w,P_T} : 1-p \}} }
    \]
    
    In this case, $\mu = \delta(\proc{c}{w,\eflip{p}{P_H}{P_T}})$ and $\mu' = p \cdot \delta(\proc{c}{w,P_H}) + (1-p) \cdot \delta(\proc{c}{w,P_T})$.
    
    Then we conclude by (\textsc{SP:Flip}).
    
    \item[Case:]
    \[\small
    \Rule{E:Dist}
    { \calC_{i_0} \sstep \calC_{i_0}'
    }
    { \proc{c}{\{ \calC_i : p_i \}_{i \in \calI}} \sstep \proc{c}{\{ \calC_i : p_i \}_{i \in \calI \setminus \{i_0\}} \dplus \{ \calC'_{i_0} : p_{i_0} \}} }
    \]
    
    By inversion on the simulation relations, we know that $\mu = \tsum_{i \in \calI} p_i \cdot \mu_i$ where $\\calC_i \approx \mu_i$ for each $i \in \calI$, and $\mu' = \tsum_{i \in \calI \setminus \{i_0\}} p_i \cdot \mu_i + p_{i_0} \cdot \mu_{i_0}'$ where $\calC_{i_0}' \approx \mu_{i_0}'$.
    
    By induction hypothesis, we know that $\mu_{i_0} \dstep \mu_{i_0}'$, i.e., there exists some $\calD$ in the domain of $\mu_{i_0}$ that can make a step under the $\pstep$ relation.
    
    Thus, we can add to both sides of $\mu_{i_0} \dstep \mu_{i_0}'$ and conclude that $\tsum_{i \in \calI \setminus \{ i_0\}} p_i \cdot \mu_i + p_{i_0} \cdot \mu_{i_0} \dstep \tsum_{i \in \calI \setminus \{ i_0\}} p_i \cdot \mu_i + p_{i_0} \cdot \mu_{i_0}'$.
    
    \item[Case:]
    \[\small\Rule{C:BDist:D}
    { \proc{c}{\{ (\calC_i \parallel \calC_j') : p_i \cdot p_j' \}_{i \in \calI,j\in\calJ}} \cstep{d,\m{det}} \calC''
    }
    { \proc{c}{\{ \calC_i : p_i\}_{i \in \calI}} \parallel \proc{d}{\{ \calC_j' : p_j'\}_{j \in \calJ}} \cstep{d,\m{det}} \calC''  }
    \]
    
    In this case, let $\calC_i \approx \mu_i$ for each $i \in \calI$ and $\calC_j' \approx \mu_j'$ for each $j \in \calJ$,
    then
    \begin{small}
    \begin{align*}
      \proc{c}{\{ \calC_i : p_i\}_{i \in \calI}} \parallel \proc{d}{\{ \calC_j' : p_j'\}_{j \in \calJ}} & \approx (\tsum_{i \in \calI} p_i \cdot \mu_i) \otimes (\tsum_{j \in \calJ} p_j' \cdot \mu_j') \bind \lambda(\calD_1,\calD_2). \delta(\calD_1 \parallel \calD_2) \\
      & = \tsum_{i \in \calI,j \in \calJ} (p_i \cdot p_j') \cdot (\mu_i \otimes \mu_j' \bind \lambda(\calD_1,\calD_2).\delta(\calD_1 \parallel \calD_2)), \\
      \proc{c}{\{ (\calC_i \parallel \calC_j') : p_i \cdot p_j' \}_{i \in \calI,j\in\calJ}} & \approx \tsum_{i \in \calI, j \in \calJ} (p_i \cdot p_j') \cdot (\mu_i \otimes \mu_j' \bind \lambda(\calD_1,\calD_2).\delta(\calD_1 \parallel \calD_2)),
    \end{align*}
    \end{small}
    thus, by induction hypothesis, we conclude that $\mu \dstep \mu'$.
    
    \item[Case:]
    \[\small
    \Rule{C:BDist:R}
    {  \proc{c}{\{ (\proc{c}{\{ \calC_i : p_i\}_{i \in \calI}} \parallel \calC_j') : p'_j  \}_{j \in \calJ}} \cstep{d,\pechoiceop} \calC''
    }
    { \proc{c}{\{ \calC_i : p_i \}_{i \in \calI}} \parallel \proc{d}{\{ \calC'_j : p'_j \}_{j \in \calJ}} \cstep{d,\pechoiceop} \calC'' }
    \]
    
     In this case, let $\calC_i \approx \mu_i$ for each $i \in \calI$ and $\calC_j' \approx \mu_j'$ for each $j \in \calJ$,
    then
    \begin{small}
    \begin{align*}
      \proc{c}{\{ \calC_i : p_i\}_{i \in \calI}} \parallel \proc{d}{\{ \calC_j' : p_j'\}_{j \in \calJ}} & \approx (\tsum_{i \in \calI} p_i \cdot \mu_i) \otimes (\tsum_{j \in \calJ} p_j' \cdot \mu_j') \bind \lambda(\calD_1,\calD_2). \delta(\calD_1 \parallel \calD_2) \\
      & = \tsum_{i \in \calI,j \in \calJ} (p_i \cdot p_j') \cdot (\mu_i \otimes \mu_j' \bind \lambda(\calD_1,\calD_2).\delta(\calD_1 \parallel \calD_2)), \\
      \proc{c}{\{ (\proc{c}{\{ \calC_i : p_i\}_{i \in \calI}} \parallel \calC_j') : p'_j  \}_{j \in \calJ}} & \approx \tsum_{j \in \calJ} p_j' \cdot (  (\tsum_{i \in \calI} p_i \cdot \mu_i) \otimes \mu_j' \bind \lambda(\calD_1,\calD_2). \delta(\calD_1 \parallel \calD_2) ) \\
      & = \tsum_{j \in \calJ'} p_j' \cdot (\tsum_{i \in \calI} p_i \cdot (\mu_i \otimes \mu_j' \bind \lambda(\calD_1,\calD_2). \delta(\calD_1 \parallel \calD_2) ) ) \\
      & = \tsum_{i \in \calI,j \in \calJ} (p_i \cdot p_j') \cdot (\mu_i \otimes \mu_j' \bind \lambda(\calD_1,\calD_2). \delta(\calD_1 \parallel \calD_2)),
    \end{align*}
    \end{small}
    thus, by induction hypothesis, we conclude that $\mu \dstep \mu'$.
  \end{description}
\end{proof}

Let us fix some initial configuration $\calC$ such that $(\cdot) \potconf{q} \calC :: \Gm$.
By fixing a scheduler to resolve nondeterminism and thinking the configuration-to-distribution relation $\pstep$ as a kernel, we can construct a Markov chain $\{\calD_n\}_{n \in \bbN}$ on configurations.
Note that by \cref{Lem:SimulationStep}, there exists a sequence of \emph{hidden} states $\{\calC_n\}_{n \in \bbN}$ such that $\calC_n \approx \mu_n$ where $\mu_n$ is the distribution of configurations after $n$ steps in the Markov chain.
Also, by \cref{The:Tree:PreservationSingle,The:Tree:PreservationComm,The:Progress}, we know that $(\cdot) \potconf{q} \calC_n :: \Gm$ for all $n \in \bbN$.
Therefore, by \cref{Lem:SimulationInvariant}, we know that $\bbE_{\calD \sim \mu_n}[\m{work}(\calD)] \le q$ for all $n \in \bbN$.
Define $Q_n \defeq \m{work}(\calD_n)$ for each $n  \in \bbN$.
Because work is nonnegative,
$\{Q_n\}_{n \in \bbN}$ forms a nonnegative, monotone, integer-valued stochastic process such that $\bbE[Q_n] \le q$ for all $n \in \bbN$, with respect to the Markov chain $\{\calD_n\}_{n \in \bbN}$.

Let $T$ be the \emph{termination time}, i.e., $T(\omega) \defeq \inf \{ \omega_n \poised \mid n \in \bbN \}$ a random variable on the Markov chain $\{\calD_n\}_{n \in \bbN}$.
Then the random variable $Q_T \defeq \lambda\omega. Q_{T(\omega)}$ represents the total work conductions by the execution trace $\omega$.
Thus, the \emph{expected total work} can be defined by $\bbE[Q_T]$.
Define $Q^{T}_n \defeq \lambda \omega. Q_{\min(T(\omega),n)}$ for $n \in \bbN$.
Then we have $Q_T = \lim_{n \to \infty} Q^{T}_n$ and $Q^T_n \le Q_n$ by the monotonicity of $\{Q_n\}_{n \in \bbN}$.
Therefore, by \emph{Monotone Convergence Theorem}~\cite{book:Williams91}, we conclude that $\bbE[Q_T] = \lim_{n \to \bbN} \bbE[Q^T_n] \le q$, i.e., the expected total work is upper-bounded by $q$.

\subsection{A Partially Successful Attempt}
\label{appendix:failed-attempt}

In this section, we present a \emph{partially successful} approach to developing the meta-theory of \lang{}.
This approach is able to handle probabilistic internal and external choices, termination and forwarding, spawning, and potential passing.
However, this approach would \emph{fail}, if we add either standard internal and external choices, or channel passing, to the feature set.
Note that this approach is \emph{not} a conservative extension of resource-aware session types~\cite{Das18RAST}.

The weighted sums are inductively (and differently from the version in the nested-multiverse development) defined on the structure of the types:
\begin{small}
\begin{align*}
  \tsum^{\m{R}}_{i \in \calI} p_i \cdot \pichoice{\ell^{p_\ell'}: A_{\ell,i}}_{\ell \in L} & \defeq \pichoice{\ell^{p_\ell'} : \tsum^{\m{R}}_{i \in \calI} p_i \cdot A_{\ell,i}}_{\ell \in L} \\
  \tsum^{\m{R}}_{i \in \calI} p_i \cdot \pichoice{\ell^{p_{\ell,i}} : A_\ell}_{\ell \in L} & \defeq \pichoice{\ell^{\tsum_{i \in \calI} p_i \cdot p_{\ell,i}} : A_\ell}_{\ell \in L} \\
  \tsum^{\m{R}}_{i \in \calI} p_i \cdot \pechoice{\ell^{p_\ell'} : A_{\ell,i}}_{\ell \in L} & \defeq \pechoice{\ell^{p_\ell'} : \tsum^{\m{R}}_{i \in \calI} p_i \cdot A_{\ell,i}}_{\ell \in L} \\
  \tsum^{\m{R}}_{i \in \calI} p_i \cdot \one & \defeq \one \\
  \tsum^{\m{R}}_{i \in \calI} p_i \cdot (\tpaypot{A_i}{r}) & \defeq \tpaypot{(\tsum^{\m{R}}_{i \in \calI} p_i \cdot A_i)}{r} \\
  \tsum^{\m{R}}_{i \in \calI} p_i \cdot (\tgetpot{A_i}{r}) & \defeq \tgetpot{(\tsum^{\m{R}}_{i \in \calI} p_i \cdot A_i)}{r} \\
  \\
  \tsum^{\m{L}}_{i \in \calI} p_i \cdot \pichoice{\ell^{p_\ell'} : A_{\ell,i}}_{\ell \in L} & \defeq \pichoice{\ell^{p_\ell'} : \tsum^{\m{L}}_{i \in \calI} p_i \cdot A_{\ell,i}}_{\ell \in L} \\
  \tsum^{\m{L}}_{i \in \calI} p_i \cdot \pechoice{\ell^{p_\ell'} : A_{\ell,i}}_{\ell \in L} & \defeq \pechoice{\ell^{p_\ell'} : \tsum^{\m{L}}_{i \in \calI} p_i \cdot A_{\ell,i}}_{\ell \in L} \\
  \tsum^{\m{L}}_{i \in \calI} p_i \cdot \pechoice{\ell^{p_{\ell,i}} : A_\ell}_{\ell \in L} & \defeq \pechoice{\ell^{\tsum_{i \in \calI} p_i \cdot p_{\ell,i}} : A_\ell}_{\ell \in L} \\
  \tsum^{\m{L}}_{i \in \calI} p_i \cdot \one & \defeq \one \\
  \tsum^{\m{L}}_{i \in \calI} p_i \cdot (\tpaypot{A_i}{r}) & \defeq \tpaypot{(\tsum^{\m{L}}_{i \in \calI} p_i \cdot A_i)}{r} \\
  \tsum^{\m{L}}_{i \in \calI} p_i \cdot (\tgetpot{A_i}{r}) & \defeq \tgetpot{(\tsum^{\m{L}}_{i \in \calI} p_i \cdot A_i)}{r}
\end{align*}
\end{small}

We first prove a key property for weighted sums.

\begin{proposition}\label{Prop:LRRearrange}
  If $A = \tsum^{\m{R}}_{i \in \calI} p_i \cdot A_i$ and $A = \tsum^{\m{L}}_{\ell \in L} p_\ell' \cdot A_\ell'$, then there exists $\{A_{\ell,i}\}_{\ell \in L, i \in \calI}$ such that $A_\ell' = \tsum^{\m{R}}_{i \in \calI} p_i \cdot A_{\ell,i}$ for $\ell \in L$ and $A_i = \tsum^{\m{L}}_{\ell \in L} p_\ell' \cdot A_{\ell,i}$ for $i \in \calI$.
\end{proposition}
\begin{proof}  
  By induction on the structure of $A$. We show the proof for two representative cases.
  
  \begin{itemize}
    \item $A = \pichoice{j^{q_j} : B_j}_{j \in \calJ}$:
    
    By $A = \tsum^{\m{L}}_{\ell \in L} p_\ell' \cdot A_{\ell}'$, we know that $A_\ell' = \pichoice{j^{q_j} : B_{j,\ell}}_{j \in \calJ}$ for each $\ell \in L$ such that $B_j = \tsum^{\m{L}}_{\ell \in L} p_\ell' \cdot B_{j,\ell}$ for each $j \in \calJ$.
    
    We proceed by case analysis on $A = \tsum^{\m{R}}_{i \in \calI} p_i \cdot A_i$.
    \begin{itemize}
      \item If $A_i = \pichoice{j^{q_{j,i}} : B_j}_{j \in \calJ}$ for each $i \in \calI$ such that $q_j = \tsum_{i \in \calI} p_i \cdot q_{j,i}$ for each $j \in \calJ$:
      
      Define $A_{\ell,i} \defeq \pichoice{j^{q_{j,i}} : B_{j,\ell}}_{j \in \calJ}$ for each $\ell \in L$ and $i \in \calI$.
      
      Then for each $\ell \in L$,
      \[
      \tsum^{\m{R}}_{i \in \calI} p_i \cdot A_{\ell,i} = \pichoice{j^{\tsum_{i \in \calI} p_i \cdot q_{j,i}}: B_{j,\ell}}_{j \in \calJ} = \pichoice{j^{q_j} : B_{j,\ell}}_{j \in \calJ} = A_\ell',
      \]
      and for each $i \in \calI$,
      \[
      \tsum^{\m{L}}_{\ell \in L} p_\ell' \cdot A_{\ell,i} = \pichoice{\ell^{q_{j,i}} : \tsum^{\m{L}}_{\ell \in L} p_\ell' \cdot B_{j,\ell}}_{j \in \calJ} = \pichoice{\ell^{q_{j,i}}: B_j}_{j \in \calJ} = A_i.
      \]
      
      \item If $A_i = \pichoice{j^{q_j} : B_{j,i}}_{j \in \calJ}$ for each $i \in \calI$ such that $B_j = \tsum^{\m{R}}_{i \in \calI} p_i \cdot B_{j,i}$ for each $j \in \calJ$:
      
      For each $j \in \calJ$, by induction hypothesis, there exists $\{B_{j,\ell,i}\}_{\ell \in L, i\in \calI}$ such that $B_{j,\ell} = \tsum^{\m{R}}_{i \in \calI} p_i \cdot B_{j,\ell,i}$ for $\ell \in L$ and $B_{j,i}  = \tsum^{\m{L}}_{\ell \in L} p_\ell' \cdot B_{j,\ell,i}$ for $i \in \calI$.
      
      Define $A_{\ell,i} \defeq \pichoice{j^{q_j} : B_{j,\ell,i}}_{j \in \calJ}$ for each $\ell \in L$ and $i \in \calI$.
      
      Then for each $\ell \in L$,
      \[
      \tsum^{\m{R}}_{i \in \calI} p_i \cdot A_{\ell,i} = \pichoice{j^{q_j}: \tsum^{\m{R}}_{i \in \calI} p_i \cdot B_{j,\ell,i}}_{j \in \calJ} = \pichoice{j^{q_j} : B_{j,\ell}}_{j \in \calJ}  = A_\ell',
      \]
      and for each $i \in \calI$,
      \[
      \tsum^{\m{L}}_{\ell \in L} p_\ell' \cdot A_{\ell,i} = \pichoice{j^{q_j}: \tsum^{\m{L}}_{\ell \in L} p_\ell' \cdot B_{j,\ell,i}}_{j \in \calJ} = \pichoice{j^{q_j} : B_{j,i}}_{j \in \calJ} = A_i.
      \]
    \end{itemize}
    
    \item $A = \tpaypot{B}{r}$:
    
    By $A = \tsum^{\m{L}}_{\ell \in L} p_\ell' \cdot A_\ell'$, we know that $A_\ell' = \tpaypot{B_\ell'}{r}$ for each $\ell \in L$ such that $B = \tsum_{\ell \in L}^{\m{L}} p_\ell' \cdot B_\ell'$.
    
    By $A = \tsum^{\m{R}}_{\ell \in L} p_i \cdot A_i$, we know that $A_i = \tpaypot{B_i}{r}$ for each $i \in \calI$ such that $B = \tsum_{i \in \calI}^{\m{R}} p_i \cdot B_i$.
    
    By induction hypothesis, there exists $\{B_{\ell,i}\}_{\ell \in L,i \in \calI}$ such that $B_\ell' = \tsum^{\m{R}}_{i \in \calI} p_i \cdot B_{\ell,i}$ for $\ell \in L$ and $B_i = \tsum^{\m{L}}_{\ell \in L} p_\ell' \cdot B_{\ell,i}$ for $i \in \calI$.
    
    We conclude by defining $A_{\ell,i} \defeq \tpaypot{B_{\ell,i}}{r}$ for each $\ell \in L$ and $i \in \calI$.
  \end{itemize}
\end{proof}

Then, we show that in the type derivation, if one type in the context is a weighted sum of a type distribution, then we can re-derive a type judgment for each type in the support of the distribution.

\begin{lemma}\label{Lem:ProcPropagateLeft}
  If $\Dl,(x:A) \entailpot{q} P :: (z : C)$ and $A = \tsum_{i \in \calI}^{\m{R}} p_i \cdot A_i$,
  then there exist $\{\Dl_i\}_{i \in \calI}$, $\{C_i\}_{i \in \calI}$, and $\{q_i\}_{i \in \calI}$ such that
  \begin{itemize}
    \item for all $i \in \calI$, $\Dl_i,(x:A_i) \entailpot{q_i} P :: (z : C_i)$, and
    \item $\tsum_{i \in \calI}^{\m{L}} p_i \cdot \Dl_i = \Dl$, $\tsum_{i \in \calI}^{\m{R}} p_i \cdot C_i = C$, $\tsum_{i \in \calI} p_i \cdot q_i = q$.
  \end{itemize}
\end{lemma}
\begin{proof}
  By induction on the derivation of $\Dl,(x:A) \entailpot{q} P :: (z : C)$.
  
  \begin{itemize}
    \item
    $\small
    \Rule{$\m{flip}$}
    { p'_H = p' \\
      p'_T = 1 - p' \\
      L = \{H,T\} \\
      \Dl = \tsum_{\ell \in L}^{\m{L}} p'_\ell \cdot \Dl_\ell \\
      A = \tsum_{\ell \in L}^{\m{L}} p'_\ell \cdot A_\ell' \\
      C = \tsum_{\ell \in L}^{\m{R}} p'_\ell \cdot C_\ell \\
      q = \tsum_{\ell \in L} p'_\ell \cdot q_\ell \\
      \Forall{ \ell \in L } \Dl_\ell, (x:A_\ell') \entailpot{q_\ell} Q_\ell :: (z : C_\ell)
    }
    { \Dl, (x:A) \entailpot{q} \eflip{p'}{Q_H}{Q_T} :: (z : C) }
    $
    
    Because $A = \tsum^{\m{R}}_{i \in \calI} p_i \cdot A_i$ and $A = \tsum^{\m{L}}_{\ell \in L} p_\ell' \cdot A_\ell'$,
    by \cref{Prop:LRRearrange},
    there exists $\{A_{\ell,i}\}_{\ell\in L,i \in \calI}$ such that $A_\ell' = \tsum^{\m{R}}_{i \in \calI} p_i \cdot A_{\ell,i}$ for $\ell \in L$ and $A_i = \tsum^{\m{L}}_{\ell \in L} p_\ell' \cdot A_{\ell,i}$ for $i \in \calI$.
    
    By induction hypothesis, for each $\ell \in L$, there exist $\{\Dl_{\ell,i}\}_{ i \in \calI}$, $\{C_{\ell,i}\}_{i \in \calI}$, and $\{q_{\ell,i}\}_{i\in \calI}$ such that
    \begin{itemize}
      \item for all $i \in \calI$, $\Dl_{\ell,i}, (x:A_{\ell,i}) \entailpot{q_{\ell,i}} Q_\ell :: (z : C_{\ell,i})$, and
      \item $\tsum^{\m{L}}_{i \in \calI} p_i \cdot \Dl_{\ell,i} = \Dl_\ell$,
      $\tsum^{\m{R}}_{i \in \calI} p_i \cdot C_{\ell,i} = C_\ell$,
      $\tsum_{i \in \calI} p_i \cdot q_{\ell,i} = q_\ell$.
    \end{itemize}
    
    For each $i \in \calI$, define $\Dl_i \defeq \tsum_{\ell \in L}^{\m{L}} p_\ell' \cdot \Dl_{\ell,i}$,
    $C_i \defeq \tsum_{\ell \in L}^{\m{R}} p_\ell' \cdot C_{\ell,i}$,
    $q_i \defeq \tsum_{\ell \in L} p_\ell' \cdot q_{\ell,i}$.
    
    Then $\Dl_i,(x : A_i) \entailpot{q_i} \eflip{p'}{Q_H}{Q_T} :: (z : C_i)$.
    
    We conclude by the following:
    \begin{small}
    \begin{align*}
      \tsum_{i \in \calI}^{\m{L}} p_i \cdot \Dl_i & = \tsum_{i \in \calI}^{\m{L}} p_i \cdot \tsum_{\ell \in L}^{\m{L}} p'_\ell \cdot \Dl_{\ell,i} \\
      & = \tsum_{\ell \in L}^{\m{L}} p_\ell' \cdot \tsum_{i \in \calI}^{\m{L}} p_i \cdot \Dl_{\ell,i} \\
      & = \tsum_{\ell \in L}^{\m{L}} p_\ell' \cdot \Dl_\ell \\
      & = \Dl, \\
      \tsum_{i \in \calI}^{\m{R}} p_i \cdot C_i & = \tsum_{i \in \calI}^{\m{R}} p_i \cdot \tsum_{\ell \in L}^{\m{R}} p_\ell' \cdot C_{\ell,i} \\
      & = \tsum_{\ell \in L}^{\m{R}} p_\ell' \cdot \tsum_{i \in \calI}^{\m{R}} p_i \cdot C_{\ell,i} \\
      & = \tsum_{\ell \in L}^{\m{R}} p_\ell' \cdot C_\ell \\
      & = C, \\
      \tsum_{i \in \calI} p_i \cdot q_i & = \tsum_{i \in \calI} p_i \cdot \tsum_{\ell \in L} p_\ell' \cdot q_{\ell,i} \\
      & = \tsum_{\ell \in L} p_\ell' \cdot \tsum_{i \in \calI} p_i \cdot q_{\ell,i} \\
      & = \tsum_{\ell \in L} p_\ell' \cdot q_\ell \\
      & = q.
    \end{align*}
    \end{small}

    \item
    $\small
    \Rule{${\pichoiceop}L{=}$}
    { \Dl = \tsum^{\m{L}}_{\ell \in L} p'_\ell \cdot \Dl_\ell \\
      C = \tsum^{\m{R}}_{\ell \in L} p'_\ell \cdot C_\ell \\
      q = \tsum_{\ell \in L} p'_\ell \cdot q_\ell \\
      \Forall{ \ell \in L } \Dl_\ell, (x : A_\ell') \entailpot{q_\ell} Q_\ell :: (z : C_\ell)
    }
    { \Dl,(x: \pichoice{\ell^{p'_\ell} : A_\ell'}_{\ell \in L}) \entailpot{q} \epcase{x}{\ell}{Q_\ell}_{\ell \in L} :: (z : C) }
    $
    
    By the definition of sharing, we know that
    \begin{itemize}
    \item If $A_i = \pichoice{ \ell^{p'_{\ell,i}} : A_\ell' }_{\ell \in L}$ for each $i \in \calI$ such that
    \begin{small}
    \begin{align*}
      & \pichoice{\ell^{p_\ell'} : A_\ell'}_{\ell \in L} = \tsum_{i \in \calI}^{\m{R}} p_i \cdot (\pichoice{\ell^{p_{\ell,i}'} : A_{\ell}'}_{\ell \in L}) = \pichoice{\ell^{ \tsum_{i \in \calI}p_i \cdot p'_{\ell,i}} :  A_{\ell}' }_{\ell \in L} \\
      \implies{} & \forall \ell \in L: p'_\ell = \tsum_{i \in \calI} p_i \cdot p'_{\ell,i}.
    \end{align*}
    \end{small}
    
    For each $i \in \calI$, define $\Dl_i \defeq \tsum^{\m{L}}_{\ell \in L} p'_{\ell,i} \cdot \Dl_\ell$,
    $C_i \defeq \tsum^{\m{R}}_{\ell \in L} p'_{\ell,i} \cdot C_\ell$,
    $q_i \defeq \tsum_{\ell \in L} p'_{\ell,i} \cdot q_\ell$.
    
    Then $\Dl_i,(x:A_i) \entailpot{q_i} \epcase{x}{\ell}{Q_\ell}_{\ell \in L} :: (z : C_i)$.
    
    We conclude by the following:
    \begin{small}
    \begin{align*}
      \tsum_{i \in \calI}^{\m{L}} p_i \cdot \Dl_i & = \tsum_{i \in \calI}^{\m{L}} p_i \cdot \tsum_{\ell \in L}^{\m{L}} p'_{\ell,i} \cdot \Dl_\ell \\
      & = \tsum^{\m{L}}_{\ell \in L} (\tsum_{i \in \calI} p_i \cdot p'_{\ell,i}) \cdot \Dl_\ell \\
      & = \tsum^{\m{L}}_{\ell \in L} p_\ell' \cdot \Dl_\ell \\
      & = \Dl, \\
      \tsum_{i \in \calI}^{\m{R}} p_i \cdot C_i & = \tsum_{i \in \calI}^{\m{R}} p_i \cdot \tsum_{\ell \in L}^{\m{R}} p_{\ell,i}' \cdot C_{\ell} \\
      & = \tsum_{\ell \in L}^{\m{R}} (\tsum_{i \in \calI} p_i \cdot p'_{\ell,i}) \cdot C_{\ell} \\
      & = \tsum_{\ell \in L}^{\m{R}} p_\ell' \cdot C_\ell \\
      & = C, \\
      \tsum_{i \in \calI} p_i \cdot q_i & = \tsum_{i \in \calI} p_i \cdot \tsum_{\ell \in L} p_{\ell,i}' \cdot q_{\ell} \\
      & = \tsum_{\ell \in L} (\tsum_{i \in \calI} p_i \cdot p_{\ell,i}') \cdot q_{\ell} \\
      & = \tsum_{\ell \in L} p_\ell' \cdot q_\ell \\
      & = q.
    \end{align*}
    \end{small}
    
    \item If $A_i = \pichoice{\ell^{p_{\ell}'}: A_{\ell,i}}_{\ell \in L}$ for each $i \in \calI$ such that
    \begin{small}
    \begin{align*}
      & \pichoice{\ell^{p_\ell'} : A_\ell'}_{\ell \in L} = \tsum_{i \in \calI}^{\m{R}} p_i \cdot (\pichoice{\ell^{p_{\ell}'} : A_{\ell,i}}_{\ell \in L}) = \pichoice{\ell^{ p_\ell'} :  \tsum^{\m{R}}_{i \in \calI} p_i \cdot A_{\ell,i} }_{\ell \in L} \\
      \implies{} & \forall \ell \in L: A'_\ell = \tsum_{i \in \calI}^{\m{R}} p_i \cdot A_{\ell,i}.
    \end{align*}
    \end{small}
    
    By induction hypothesis, for each $\ell \in L$, there exist $\{\Dl_{\ell,i}\}_{i \in \calI}$, $\{C_{\ell,i}\}_{i \in \calI}$, and $\{q_{\ell,i}\}_{i \in \calI}$ such that
    \begin{itemize}
      \item for all $i \in \calI$, $\Dl_{\ell,i},(x:A_{\ell,i}) \entailpot{q_{\ell,i}} Q_\ell :: (z : C_{\ell,i})$, and
      \item $\tsum^{\m{L}}_{i \in \calI} p_i \cdot \Dl_{\ell,i} = \Dl_\ell$,
      $\tsum^{\m{R}}_{i \in \calI} p_i \cdot C_{\ell,i} = C_\ell$,
      $\tsum_{i \in \calI} p_i \cdot q_{\ell,i} = q_\ell$.
    \end{itemize}
    
    For each $i \in \calI$, define $\Dl_i \defeq \tsum_{\ell \in L}^{\m{L}} p_\ell' \cdot \Dl_{\ell,i}$,
    $C_i \defeq \tsum_{\ell \in L}^{\m{R}} p_\ell' \cdot C_{\ell,i}$,
    $q_i \defeq \tsum_{\ell \in L} p_\ell' \cdot q_{\ell,i}$.\
    
    Then $\Dl_i,(x:A_i) \entailpot{q_i} \epcase{x}{\ell}{Q_\ell}_{\ell \in L} :: (z : C_i)$.
    
    We conclude by the following:
    \begin{small}
    \begin{align*}
      \tsum_{i \in \calI}^{\m{L}} p_i \cdot \Dl_i & = \tsum_{i \in \calI}^{\m{L}} p_i \cdot \tsum_{\ell \in L}^{\m{L}} p'_\ell \cdot \Dl_{\ell,i} \\
      & = \tsum_{\ell \in L}^{\m{L}} p_\ell' \cdot \tsum_{i \in \calI}^{\m{L}} p_i \cdot \Dl_{\ell,i} \\
      & = \tsum_{\ell \in L}^{\m{L}} p_\ell' \cdot \Dl_\ell \\
      & = \Dl, \\
      \tsum_{i \in \calI}^{\m{R}} p_i \cdot C_i & = \tsum_{i \in \calI}^{\m{R}} p_i \cdot \tsum_{\ell \in L}^{\m{R}} p_\ell' \cdot C_{\ell,i} \\
      & = \tsum_{\ell \in L}^{\m{R}} p_\ell' \cdot \tsum_{i \in \calI}^{\m{R}} p_i \cdot C_{\ell,i} \\
      & = \tsum_{\ell \in L}^{\m{R}} p_\ell' \cdot C_\ell \\
      & = C, \\
      \tsum_{i \in \calI} p_i \cdot q_i & = \tsum_{i \in \calI} p_i \cdot \tsum_{\ell \in L} p_\ell' \cdot q_{\ell,i} \\
      & = \tsum_{\ell \in L} p_\ell' \cdot \tsum_{i \in \calI} p_i \cdot q_{\ell,i} \\
      & = \tsum_{\ell \in L} p_\ell' \cdot q_\ell \\
      & = q.
    \end{align*}
    \end{small}
    
    \end{itemize}
        
    \item
    $\small
    \Rule{${\pichoiceop}L{\neq}$}
    { \Dl = \tsum^{\m{L}}_{\ell \in L} p'_\ell \cdot \Dl_\ell \\
      A = \tsum^{\m{L}}_{\ell \in L} p'_\ell \cdot A_\ell' \\
      C = \tsum^{\m{R}}_{\ell \in L} p'_\ell \cdot C_\ell \\
      q = \tsum_{\ell \in L} p'_\ell \cdot q_\ell \\
      \Forall{ \ell \in L } \Dl_\ell, (x : A_\ell'), (y:B_\ell) \entailpot{q_\ell} Q_\ell :: (z : C_\ell)
    }
    { \Dl,(x:A),(y: \pichoice{\ell^{p'_\ell} : B_\ell}_{\ell \in L}) \entailpot{q} \epcase{y}{\ell}{Q_\ell}_{\ell \in L} :: (z : C) }
    $
  
    Because $A = \tsum^{\m{R}}_{i \in \calI} p_i \cdot A_i$ and $A = \tsum^{\m{L}}_{\ell \in L} p_\ell' \cdot A_\ell'$,
    by \cref{Prop:LRRearrange},
    there exists $\{A_{\ell,i}\}_{\ell\in L,i \in \calI}$ such that $A_\ell' = \tsum^{\m{R}}_{i \in \calI} p_i \cdot A_{\ell,i}$ for $\ell \in L$ and $A_i = \tsum^{\m{L}}_{\ell \in L} p_\ell' \cdot A_{\ell,i}$ for $i \in \calI$.
    
    By induction hypothesis, for each $\ell \in L$, there exist $\{\Dl_{\ell,i}\}_{i \in \calI}$, $\{B_{\ell,i}\}_{i \in \calI}$, $\{C_{\ell,i}\}_{i \in \calI}$, and $\{q_{\ell,i}\}_{i \in \calI}$ such that
    \begin{itemize}
      \item for all $i \in \calI$, $\Dl_{\ell,i},(x:A_{\ell,i}),(y:B_{\ell,i}) \entailpot{q_{\ell,i}} Q_{\ell} :: (z : C_{\ell,i})$, and
      \item $\tsum^{\m{L}}_{i \in \calI} p_i \cdot \Dl_{\ell,i} = \Dl_\ell$,
      $\tsum^{\m{L}}_{i \in \calI} p_i \cdot B_{\ell,i} = B_\ell$,
      $\tsum^{\m{R}}_{i \in \calI} p_i \cdot C_{\ell,i} = C_\ell$,
      $\tsum_{i \in \calI} p_i \cdot q_{\ell,i} = q_\ell$.
    \end{itemize}
    
    For each $i \in \calI$, define $\Dl_i \defeq \tsum^{\m{L}}_{\ell \in L} p_\ell' \cdot \Dl_{\ell,i}$, $C_i \defeq \tsum^{\m{R}}_{\ell \in L} p_\ell' \cdot C_{\ell,i}$, $q_i \defeq \tsum_{\ell \in L} p_\ell' \cdot q_{\ell,i}$.
    
    Then $\Dl_i,(x:A_i),(y:\pichoice{\ell^{p_\ell'} : B_{\ell,i}}_{\ell \in L}) \entailpot{q_i} \epcase{y}{\ell}{Q_\ell}_{\ell \in L} :: (z : C_i)$.
    
    We conclude by the following:
    \begin{small}
    \begin{align*}
      \tsum_{i \in \calI}^{\m{L}} p_i \cdot \Dl_i & = \tsum_{i \in \calI}^{\m{L}} p_i \cdot \tsum_{\ell \in L}^{\m{L}} p'_\ell \cdot \Dl_{\ell,i} \\
      & = \tsum_{\ell \in L}^{\m{L}} p_\ell' \cdot \tsum_{i \in \calI}^{\m{L}} p_i \cdot \Dl_{\ell,i} \\
      & = \tsum_{\ell \in L}^{\m{L}} p_\ell' \cdot \Dl_\ell \\
      & = \Dl, \\
      \tsum_{i \in \calI}^{\m{L}} p_i \cdot \pichoice{\ell^{p_\ell'} : B_{\ell,i}}_{\ell \in L} & = \pichoice{\ell^{p_\ell'} : \tsum^{\m{L}}_{i \in \calI} p_i \cdot B_{\ell,i}}_{\ell \in L} \\
      & = \pichoice{\ell^{p_\ell'} : B_\ell}_{\ell \in L}, \\
      \tsum_{i \in \calI}^{\m{R}} p_i \cdot C_i & = \tsum_{i \in \calI}^{\m{R}} p_i \cdot \tsum_{\ell \in L}^{\m{R}} p_\ell' \cdot C_{\ell,i} \\
      & = \tsum_{\ell \in L}^{\m{R}} p_\ell' \cdot \tsum_{i \in \calI}^{\m{R}} p_i \cdot C_{\ell,i} \\
      & = \tsum_{\ell \in L}^{\m{R}} p_\ell' \cdot C_\ell \\
      & = C, \\
      \tsum_{i \in \calI} p_i \cdot q_i & = \tsum_{i \in \calI} p_i \cdot \tsum_{\ell \in L} p_\ell' \cdot q_{\ell,i} \\
      & = \tsum_{\ell \in L} p_\ell' \cdot \tsum_{i \in \calI} p_i \cdot q_{\ell,i} \\
      & = \tsum_{\ell \in L} p_\ell' \cdot q_\ell \\
      & = q.
    \end{align*}
    \end{small}
    
    \item
    $\small
    \Rule{${\pichoiceop}R$}
    { p_k' = 1 \\
      p_j' = 0\; (j \neq k) \\
      \Dl,(x:A) \entailpot{q} P :: (z : C_k)
    }
    { \Dl,(x:A) \entailpot{q} (\esendlp{z}{k} \semi P) :: (z : \pichoice{\ell^{p_\ell'} : C_\ell}_{\ell \in L}) }
    $
    
    By induction hypothesis, there exist $\{\Dl_i\}_{i \in \calI}$, $\{C_{k,i}\}_{i \in \calI}$, and $\{q_i\}_{i \in \calI}$ such that
    \begin{itemize}
      \item for all $i \in \calI$, $\Dl_i,(x:A_i) \entailpot{q_i} P :: (z : C_{k,i})$, and
      \item $\tsum^{\m{L}}_{i \in \calI} p_i \cdot \Dl_i = \Dl$, $\tsum^{\m{R}}_{i \in \calI} p_i \cdot C_{k,i} = C_k$, $\tsum_{i \in \calI} p_i \cdot q_i = q$.
    \end{itemize}
    
    Then for each $i \in \calI$, $\Dl_i, (x:A_i) \entailpot{q_i} (\esendlp{z}{k} \semi P) :: (z : \pichoice{\ell^{p_\ell'} : C_{\ell,i}}_{\ell \in L})$, where $C_{j,i} \defeq C_j$ for $j \neq k$.
    
    We conclude by the following:
    \begin{small}
    \begin{align*}
      \tsum_{i \in \calI}^{\m{R}} p_i \cdot \pichoice{\ell^{p_\ell'} : C_{\ell,i}}_{\ell \in L} & = \pichoice{\ell^{p_\ell'} : \tsum_{i \in \calI}^{\m{R}} p_i \cdot C_{\ell,i}}_{\ell \in L} \\
      & = \pichoice{\ell^{p_\ell'} : C_\ell}_{\ell \in L}.
    \end{align*}
    \end{small}
    
    \item
    $\small
    \Rule{${\pechoiceop}R$}
    { \Dl = \tsum^{\m{L}}_{\ell \in L} p_\ell' \cdot \Dl_\ell \\
      A = \tsum^{\m{L}}_{\ell \in L} p_\ell' \cdot A_\ell' \\
      q = \tsum_{\ell \in L} p_\ell' \cdot q_\ell \\
      \Forall{\ell \in L} \Dl_\ell,(x:A_\ell) \entailpot{q_\ell} P_\ell :: (z : C_\ell)
    }
    { \Dl,(x:A) \entailpot{q} \epcase{z}{\ell}{P_\ell}_{\ell \in L} :: (z : \pechoice{\ell^{p_\ell'} : C_\ell}_{\ell \in L}) }
    $
    
    Because $A = \tsum^{\m{R}}_{i \in \calI} p_i \cdot A_i$ and $A = \tsum^{\m{L}}_{\ell \in L} p_\ell' \cdot A_\ell'$,
    by \cref{Prop:LRRearrange},
    there exists $\{A_{\ell,i}\}_{\ell\in L,i \in \calI}$ such that $A_\ell' = \tsum^{\m{R}}_{i \in \calI} p_i \cdot A_{\ell,i}$ for $\ell \in L$ and $A_i = \tsum^{\m{L}}_{\ell \in L} p_\ell' \cdot A_{\ell,i}$ for $i \in \calI$.
    
    By induction hypothesis, for each $\ell \in L$, there exist $\{\Dl_{\ell,i}\}_{i \in \calI}$, $\{C_{\ell,i}\}_{i \in \calI}$, and $\{q_{\ell,i}\}_{i \in \calI}$ such that
    \begin{itemize}
      \item for all $i \in \calI$, $\Dl_i,(x:A_{\ell,i}) \entailpot{q_{\ell,i}} P_{\ell} :: (z : C_{\ell,i})$, and
      \item $\tsum^{\m{L}}_{i \in \calI} p_i \cdot \Dl_{\ell,i} = \Dl_\ell$,
      $\tsum^{\m{R}}_{i \in \calI} p_i \cdot C_{\ell,i} = C_\ell$,
      $\tsum_{i \in \calI} p_i \cdot q_{\ell,i} = q_\ell$.
    \end{itemize}
    
    For each $i \in \calI$, define $\Dl_i \defeq \tsum^{\m{L}}_{\ell \in L} p_\ell' \cdot \Dl_{\ell,i}$, $q_i \defeq \tsum_{\ell \in L} p_\ell' \cdot q_{\ell,i}$.
    
    Then $\Dl_i,(x:A_i) \entailpot{q_i} \epcase{z}{\ell}{P_\ell}_{\ell \in L} :: (z : \pechoice{\ell^{p_\ell'} : C_{\ell,i}}_{\ell \in L}  )$.
    
    We conclude by the following:
    \begin{small}
    \begin{align*}
      \tsum_{i \in \calI}^{\m{L}} p_i \cdot \Dl_i & = \tsum_{i \in \calI}^{\m{L}} p_i \cdot \tsum_{\ell \in L}^{\m{L}} p'_\ell \cdot \Dl_{\ell,i} \\
      & = \tsum_{\ell \in L}^{\m{L}} p_\ell' \cdot \tsum_{i \in \calI}^{\m{L}} p_i \cdot \Dl_{\ell,i} \\
      & = \tsum_{\ell \in L}^{\m{L}} p_\ell' \cdot \Dl_\ell \\
      & = \Dl, \\
      \tsum_{i \in \calI}^{\m{R}} p_i \cdot \pechoice{\ell^{p_\ell'} : C_{\ell,i}}_{\ell \in L} & = \pechoice{\ell^{p_\ell'} : \tsum^{\m{R}}_{i \in \calI} p_i \cdot C_{\ell,i}}_{\ell \in L} \\
      & = \pechoice{\ell^{p_\ell'} : C_\ell}_{\ell \in L}, \\
      \tsum_{i \in \calI} p_i \cdot q_i & = \tsum_{i \in \calI} p_i \cdot \tsum_{\ell \in L} p_\ell' \cdot q_{\ell,i} \\
      & = \tsum_{\ell \in L} p_\ell' \cdot \tsum_{i \in \calI} p_i \cdot q_{\ell,i} \\
      & = \tsum_{\ell \in L} p_\ell' \cdot q_\ell \\
      & = q.
    \end{align*}
    \end{small}
    
    \item
    $\small
    \Rule{${\pechoiceop}L{=}$}
    { p_k' = 1 \\
      p_j' = 0\; (j \neq k) \\
      \Dl,(x:A_k') \entailpot{q} Q :: (z : C)
    }
    { \Dl,(x:\pechoice{\ell^{p_\ell'}:A_{\ell}'}_{\ell \in L}) \entailpot{q} (\esendlp{x}{k} \semi Q) :: (z : C) }
    $
    
    By the definition of sharing, we know that
    $A_i = \pechoice{\ell^{p_\ell'} : A_{\ell,i}}_{\ell \in L}$ for each $i \in \calI$ such that
    \begin{small}
    \begin{align*}
      & \pechoice{\ell^{p_\ell'} : A_\ell'}_{\ell \in L} = \tsum_{i \in \calI}^{\m{R}} p_i \cdot (\pechoice{\ell^{p_{\ell}'} : A_{\ell,i}}_{\ell \in L}) = \pechoice{\ell^{ p_\ell'} :  \tsum^{\m{R}}_{i \in \calI} p_i \cdot A_{\ell,i} }_{\ell \in L} \\
      \implies{} & \forall \ell \in L: A'_\ell = \tsum_{i \in \calI}^{\m{R}} p_i \cdot A_{\ell,i}.
    \end{align*}
    \end{small}
    
    By induction hypothesis, there exist $\{\Dl_i\}_{i \in \calI}$, $\{C_i\}_{i \in \calI}$, and $\{q_i\}_{i \in \calI}$ such that
    \begin{itemize}
      \item for all $i \in \calI$, $\Dl_i,(x:A_{k,i}) \entailpot{q_i} Q :: (z : C_i)$, and
      \item $\tsum^{\m{L}}_{i \in \calI} p_i \cdot \Dl_i = \Dl$,
      $\tsum^{\m{R}}_{i \in \calI} p_i \cdot C_i = C$,
      $\tsum_{i \in \calI} p_i \cdot q_i = q$.
    \end{itemize}
    
    Then for each $i \in \calI$, $\Dl_i,(x: \pechoice{\ell^{p_\ell'}: A_{\ell,i}}_{\ell \in L}) \entailpot{q_i} (\esendlp{x}{k} \semi Q) :: (z : C_i)$.
    
    \item
    $\small
    \Rule{${\pechoiceop}L{\neq}$}
    { p_k' = 1 \\
      p_j' = 0\; (j \neq k) \\
      \Dl,(x:A),(y:B_k) \entailpot{q} Q :: (z : C)
    }
    { \Dl,(x:A),(y:\pechoice{\ell^{p_\ell'}:B_{\ell}}_{\ell \in L}) \entailpot{q} (\esendlp{y}{k} \semi Q) :: (z : C) }
    $
    
    By induction hypothesis, there exist $\{\Dl_i\}_{i \in \calI}$, $\{B_{k,i}\}_{i \in \calI}$, $\{C_i\}_{i \in \calI}$, and $\{q_i\}_{i \in \calI}$ such that
    \begin{itemize}
      \item for all $i \in \calI$, $\Dl_i,(x:A_i),(y:B_{k,i}) \entailpot{q_i} Q :: (z : C_i)$, and
      \item $\tsum^{\m{L}}_{i \in \calI} p_i \cdot \Dl_i = \Dl$,
      $\tsum^{\m{L}}_{i \in \calI} p_i \cdot B_{k,i} = B_k$,
      $\tsum^{\m{R}}_{i \in \calI} p_i \cdot C_i = C$,
      $\tsum_{i \in \calI} p_i \cdot q_i = q$.
    \end{itemize}
    
    Then for each $i \in \calI$, $\Dl_i,(x:A_i),(y:\pechoice{\ell^{p_\ell'}:B_{\ell,i}}_{\ell \in L}) \entailpot{q_i} (\esendlp{y}{k} \semi Q) :: (z : C_i)$, where $B_{j,i} \defeq B_j$ for $j \neq k$.
    
    We conclude by the following:
    \begin{small}
    \begin{align*}
      \tsum_{i \in \calI}^{\m{L}} p_i \cdot \pechoice{\ell^{p_\ell'} : B_{\ell,i}}_{\ell \in L} & = \pechoice{\ell^{p_\ell'} : \tsum_{i \in \calI}^{\m{L}} p_i \cdot B_{\ell,i}}_{\ell \in L} \\
      & = \pichoice{\ell^{p_\ell'} : B_\ell}_{\ell \in L}.
    \end{align*}
    \end{small}
    
    \item
    $\small
    \Rule{${\one}L{=}$}
    { \Dl \entailpot{q} Q :: (z : C)
    }
    { \Dl,(x:\one) \entailpot{q} (\ewait{x} \semi Q) :: (z : C)  }
    $
    
    By the definition of sharing, we know that $A_i = \one$ for each $i \in \calI$.
    
    For each $i \in \calI$, define $\Dl_i \defeq \Dl$, $C_i \defeq C$, $q_i \defeq q$.
    
    \item
    $\small
    \Rule{${\one}L{\neq}$}
    { \Dl,(x:A) \entailpot{q} Q :: (z : C)
    }
    { \Dl,(x:A),(y:\one) \entailpot{q} (\ewait{y} \semi Q) :: (z : C) }
    $
    
    By induction hypothesis, there exist $\{\Dl_i\}_{i \in \calI}$, $\{C_i\}_{i \in \calI}$, and $\{q_i\}_{i \in \calI}$ such that
    \begin{itemize}
      \item for all $i \in \calI$, $\Dl_i,(x:A_i) \entailpot{q_i} Q :: (z : C_i)$, and
      \item $\tsum^{\m{L}}_{i \in \calI} p_i \cdot \Dl_i = \Dl$,
      $\tsum^{\m{R}}_{i \in \calI} p_i \cdot C_i = C$,
      $\tsum_{i \in \calI} p_i \cdot q_i = q$.
    \end{itemize}
    
    Then for each $i \in \calI$, $\Dl_i,(x:A_i),(y:\one) \entailpot{q_i} (\ewait{y} \semi Q) :: (z : C_i)$.
    
    \item
    $\small
    \Rule{$\m{id}$}
    {
    }
    { (x : A) \entailpot{0} (\fwd{z}{x}) :: (z : A) }
    $
    
    For each $i \in \calI$, we have $(x : A_i) \entailpot{0} (\fwd{z}{x}) :: (z : A_i)$.
    
    Then we conclude by the assumption that $A = \tsum^{\m{R}}_{i \in \calI} p_i \cdot A_i$.
    
    \item
    $\small
    \Rule{$\m{spawn}{\in}$}
    { (x':A),\many{y':B} \entailpot{q} f = P_f :: (t' : D) \in \Sg \\
      \Dl,(t:D) \entailpot{r} Q :: (z : C)
    }
    { \Dl,(x:A),\many{y:B} \entailpot{q+r} (\ecut{t}{f}{x\;\many{y}}{Q}) :: (z : C) }
    $
    
    By induction hypothesis, there exist $\many{\{B_i\}_{i \in \calI}}$, $\{D_i\}_{i \in \calI}$, and $\{q_i\}_{i \in \calI}$ such that
    \begin{itemize}
      \item for all $i \in \calI$, $(x':A_i),\many{y':B_i} \entailpot{q_i} P_f :: (t' : D_i)$, and
      \item $\tsum^{\m{L}}_{i \in \calI} p_i \cdot \many{B_i} = \many{B}$,
      $\tsum^{\m{R}}_{i \in \calI} p_i \cdot D_i = D$,
      $\tsum_{i \in \calI} p_i \cdot q_i = q$.
    \end{itemize}
    
    Again, by induction hypothesis, there exist $\{\Dl_i\}_{i \in \calI}$, $\{C_i\}_{i \in \calI}$, and $\{r_i\}_{i \in \calI}$ such that
    \begin{itemize}
      \item for all $i \in \calI$, $\Dl_i,(t:D_i) \entailpot{r_i} Q :: (z : C_i)$, and
      \item $\tsum^{\m{L}}_{i \in \calI} p_i \cdot \Dl_i = \Dl$,
      $\tsum^{\m{R}}_{i \in \calI} p_i \cdot C_i = C$,
      $\tsum_{i \in \calI} p_i \cdot r_i = r$.
    \end{itemize}
    
    Then for each $i \in \calI$, $\Dl_i,(x:A_i),\many{y:B_i} \entailpot{q_i+r_i} (\ecut{t}{f}{x\;\many{y}}{Q}) :: (z : C_i)$.
    
    We conclude by the following:
    \begin{small}
    \begin{align*}
      \tsum_{i \in \calI} p_i \cdot (q_i +r_i) & = \tsum_{i \in \calI} p_i \cdot q_i + \tsum_{i \in \calI} p_i \cdot r_i \\
      & = q + r.
    \end{align*}
    \end{small}
    
    \item
    $\small
    \Rule{$\m{spawn}{\not\in}$}
    { \many{y':B} \entailpot{q} f = P_f :: (t' : D) \in \Sg \\
      \Dl,(x:A),(t:D) \entailpot{r} Q :: (z : C)
    }
    { \Dl,(x:A),\many{y:B} \entailpot{q+r} (\ecut{t}{f}{\many{y}}{Q}) :: (z : C) }
    $
    
    By induction hypothesis, there exist $\{\Dl_i\}_{i \in \calI}$, $\{D_i\}_{i \in \calI}$, $\{C_i\}_{i \in \calI}$, and $\{r_i\}_{i \in \calI}$ such that
    \begin{itemize}
      \item for all $i \in \calI$, $\Dl_i,(x:A_i),(t:D_i) \entailpot{r_i} Q :: (z : C_i)$, and
      \item $\tsum^{\m{L}}_{i \in \calI} p_i \cdot \Dl_i = \Dl$,
      $\tsum^{\m{L}}_{i \in \calI} p_i \cdot D_i = D$,
      $\tsum^{\m{R}}_{i \in \calI} p_i \cdot C_i = C$,
      $\tsum_{i \in \calI} p_i \cdot r_i = r$.
    \end{itemize}
    
    By \cref{Lem:ProcPropagateRight}, there exist $\many{\{B_i\}_{i \in \calI}}$ and $\{q_i\}_{i \in \calI}$ such that
    \begin{itemize}
      \item for all $i \in \calI$, $\many{y' : B_i} \entailpot{q_i} P_f :: (t' : D_i)$, and
      \item $\tsum^{\m{L}}_{i \in \calI} p_i \cdot \many{B_i} = \many{B}$,
      $\tsum_{i \in \calI} p_i \cdot q_i = q$.
    \end{itemize}
    
    Then for each $i \in \calI$, $\Dl_i,(x:A_i),\many{y:B_i} \entailpot{q_i+r_i} (\ecut{t}{f}{\many{y}}{Q}) :: (z : C_i)$.
    
    We conclude by the following:
    \begin{small}
    \begin{align*}
      \tsum_{i \in \calI} p_i \cdot (q_i +r_i) & = \tsum_{i \in \calI} p_i \cdot q_i + \tsum_{i \in \calI} p_i \cdot r_i \\
      & = q + r.
    \end{align*}
    \end{small}
    
    \item
    $\small
    \Rule{${\paypot}R$}
    { \Dl,(x:A) \entailpot{q-r} P :: (z : C)
    }
    { \Dl,(x:A) \entailpot{q} (\epay{z}{r} \semi P) :: (z : \tpaypot{C}{r}) }
    $
    
    By induction hypothesis, there exist $\{\Dl_i\}_{i \in \calI}$, $\{C_i\}_{i \in \calI}$, and $\{q_i\}_{i \in \calI}$ such that
    \begin{itemize}
      \item for all $i \in \calI$, $\Dl_i,(x:A_i) \entailpot{q_i} P :: (z : C_i)$, and
      \item $\tsum^{\m{L}}_{i \in \calI} p_i \cdot \Dl_i = \Dl$,
      $\tsum^{\m{R}}_{i \in \calI} p_i \cdot C_i = C$,
      $\tsum_{i \in \calI} p_i \cdot q_i = q-r$.
    \end{itemize}
    
    Then for each $i \in \calI$, $\Dl_i,(x:A_i) \entailpot{q_i+r} (\epay{z}{r} \semi P) :: (z : \tpaypot{C_i}{r})$.
    
    We conclude by the following:
    \begin{small}
    \begin{align*}
      \tsum^{\m{R}}_{i \in \calI} p_i \cdot (\tpaypot{C_i}{r}) & = \tpaypot{(\tsum^{\m{R}}_{i \in \calI} p_i \cdot C_i)}{r} \\
      & = \tpaypot{C}{r}, \\
      \tsum_{i \in \calI} p_i \cdot (q_i+r) & = \tsum_{i \in \calI} p_i \cdot q_i + r \\
      & = (q-r)+r \\
      & = q.
    \end{align*}
    \end{small}
    
    \item
    $\small
    \Rule{${\paypot}L{=}$}
    { \Dl,(x : A) \entailpot{q+r} Q :: (z : C)
    }
    { \Dl,(x:\tpaypot{A}{r}) \entailpot{q} (\eget{x}{r} \semi Q) :: (z : C) }
    $
    
    By induction hypothesis, there exist $\{\Dl_i\}_{i \in \calI}$, $\{C_i\}_{i \in \calI}$, and $\{q_i\}_{i \in \calI}$ such that
    \begin{itemize} 
      \item for all $i \in \calI$, $\Dl_i,(x:A_i) \entailpot{q_i} Q :: (z : C_i)$, and
      \item $\tsum^{\m{L}}_{i \in \calI} p_i \cdot \Dl_i = \Dl$,
      $\tsum^{\m{R}}_{i \in \calI} p_i \cdot C_i = C$,
      $\tsum_{i \in \calI} p_i \cdot q_i = q+r$.
    \end{itemize}
    
    Then for each $i \in \calI$, $\Dl_i,(x:\tpaypot{A_i}{r}) \entailpot{q_i-r} (\eget{x}{r} \semi Q) :: (z : C_i)$.
    
    We conclude by the following:
    \begin{small}
    \begin{align*}
      \tsum_{i \in \calI} p_i \cdot (q_i-r) & = \tsum_{i \in \calI} p_i \cdot q_i - r \\
      & = (q+r)-r \\
      & = q.
    \end{align*}
    \end{small}
    
    \item
    $\small
    \Rule{${\paypot}L{\neq}$}
    { \Dl,(x:A),(y:B) \entailpot{q+r} Q :: (z : C)
    }
    { \Dl,(x:A),(y:\tpaypot{B}{r}) \entailpot{q} (\eget{y}{r} \semi Q) :: (z : C) }
    $
    
    By induction hypothesis, there exist $\{\Dl_i\}_{i \in \calI}$, $\{B_i\}_{i \in \calI}$, $\{C_i\}_{i \in \calI}$, and $\{q_i\}_{i \in \calI}$ such that
    \begin{itemize} 
      \item for all $i \in \calI$, $\Dl_i,(x:A_i),(y:B_i) \entailpot{q_i} Q :: (z : C_i)$, and
      \item $\tsum^{\m{L}}_{i \in \calI} p_i \cdot \Dl_i = \Dl$,
      $\tsum^{\m{L}}_{i \in \calI} p_i \cdot B_i = B$,
      $\tsum^{\m{R}}_{i \in \calI} p_i \cdot C_i = C$,
      $\tsum_{i \in \calI} p_i \cdot q_i = q+r$.
    \end{itemize}
    
    Then for each $i \in \calI$, $\Dl_i,(x:A_i),(y:\tpaypot{B_i}{r}) \entailpot{q_i-r} (\eget{x}{r} \semi Q) :: (z : C_i)$.
    
    We conclude by the following:
    \begin{small}
    \begin{align*}
      \tsum^{\m{L}}_{i \in \calI} p_i \cdot (\tpaypot{B_i}{r}) & = \tpaypot{(\tsum^{\m{L}}_{i \in \calI} p_i \cdot B_i)}{r} \\
      & = \tpaypot{B}{r}, \\
      \tsum_{i \in \calI} p_i \cdot (q_i-r) & = \tsum_{i \in \calI} p_i \cdot q_i - r \\
      & = (q+r)-r \\
      & = q.
    \end{align*}
    \end{small}
    
    \item
    $\small
    \Rule{${\getpot}R$}
    { \Dl,(x:A) \entailpot{q+r} P :: (z : C)
    }
    { \Dl,(x:A) \entailpot{q} (\eget{z}{r} \semi P) :: (z : \tgetpot{C}{r}) }
    $
    
    By induction hypothesis, there exist $\{\Dl_i\}_{i \in \calI}$, $\{C_i\}_{i \in \calI}$, and $\{q_i\}_{i \in \calI}$ such that
    \begin{itemize}
      \item for all $i \in \calI$, $\Dl_i,(x:A_i) \entailpot{q_i} P :: (z : C_i)$, and
      \item $\tsum^{\m{L}}_{i \in \calI} p_i \cdot \Dl_i = \Dl$,
      $\tsum^{\m{R}}_{i \in \calI} p_i \cdot C_i = C$,
      $\tsum_{i \in \calI} p_i \cdot q_i = q+r$.
    \end{itemize}
    
    Then for each $i \in \calI$, $\Dl_i,(x:A_i) \entailpot{q_i-r} (\eget{z}{r} \semi P) :: (z : \tgetpot{C_i}{r})$.
    
    We conclude by the following:
    \begin{small}
    \begin{align*}
      \tsum^{\m{R}}_{i \in \calI} p_i \cdot (\tgetpot{C_i}{r}) & = \tgetpot{(\tsum^{\m{R}}_{i \in \calI} p_i \cdot C_i)}{r} \\
      & = \tgetpot{C}{r}, \\
      \tsum_{i \in \calI} p_i \cdot (q_i-r) & = \tsum_{i \in \calI} p_i \cdot q_i - r \\
      & = (q+r)-r \\
      & = q.
    \end{align*}
    \end{small}
    
    \item
    $\small
    \Rule{${\getpot}L{=}$}
    { \Dl, (x:A) \entailpot{q-r} Q :: (z : C)
    }
    { \Dl, (x : \tgetpot{A}{r}) \entailpot{q} (\epay{x}{r} \semi Q) :: (z : C) }
    $
    
    By induction hypothesis, there exist $\{\Dl_i\}_{i \in \calI}$, $\{C_i\}_{i \in \calI}$, $\{q_i\}_{i \in \calI}$ such that
    \begin{itemize}
      \item for all $i \in \calI$, $\Dl_i,(x:A_i) \entailpot{q_i} Q :: (z : C_i)$, and
      \item $\tsum^{\m{L}}_{i \in \calI} p_i \cdot \Dl_i = \Dl$,
      $\tsum^{\m{R}}_{i \in \calI} p_i \cdot C_i = C$,
      $\tsum_{i \in \calI} p_i \cdot q_i = q - r$.
    \end{itemize}
    
    Then for each $i \in \calI$, $\Dl_i,(x:\tgetpot{A_i}{r}) \entailpot{q_i+r} (\epay{x}{r} \semi Q) :: (z : C)$.
    
    We conclude by the following:
    \begin{small}
    \begin{align*}
      \tsum_{i \in \calI} p_i \cdot (q_i+r) & = \tsum_{i \in \calI} p_i \cdot q_i + r \\
      & = (q-r)+r \\
      & = q.
    \end{align*}
    \end{small}
    
    \item
    $\small
    \Rule{${\getpot}L{\neq}$}
    { \Dl, (x:A),(y:B) \entailpot{q-r} Q :: (z : C)
    }
    { \Dl, (x:A), (y:\tgetpot{B}{r}) \entailpot{q} (\epay{y}{r} \semi Q) :: (z : C) }
    $
    
    By induction hypothesis, there exist $\{\Dl_i\}_{i \in \calI}$, $\{B_i\}_{i \in \calI}$, $\{C_i\}_{i \in \calI}$, $\{q_i\}_{i \in \calI}$ such that
    \begin{itemize}
      \item for all $i \in \calI$, $\Dl_i,(x:A_i),(y:B_i) \entailpot{q_i} Q :: (z : C_i)$, and
      \item $\tsum^{\m{L}}_{i \in \calI} p_i \cdot \Dl_i = \Dl$,
      $\tsum^{\m{L}}_{i \in \calI} p_i \cdot B_i = B$,
      $\tsum^{\m{R}}_{i \in \calI} p_i \cdot C_i = C$,
      $\tsum_{i \in \calI} p_i \cdot q_i = q - r$.
    \end{itemize}
    
    Then for each $i \in \calI$, $\Dl_i,(x:A_i),(y:\tgetpot{B_i}{r}) \entailpot{q_i+r} (\epay{y}{r} \semi Q) :: (z : C)$.
    
    We conclude by the following:
    \begin{small}
    \begin{align*}
      \tsum^{\m{L}}_{i \in \calI} p_i \cdot (\tgetpot{B_i}{r}) & = \tgetpot{(\tsum^{\m{L}}_{i \in \calI} p_i \cdot B_i)}{r} \\
      & = \tgetpot{B}{r}, \\
      \tsum_{i \in \calI} p_i \cdot (q_i+r) & = \tsum_{i \in \calI} p_i \cdot q_i + r \\
      & = (q-r)+r \\
      & = q.
    \end{align*}
    \end{small}
    
    \item
    $\small
    \Rule{$\m{work}$}
    { \Dl,(x:A) \entailpot{q-r} P :: (z : C)
    }
    { \Dl,(x:A) \entailpot{q} (\ework{r} \semi P) :: (z : C) }
    $
    
    By induction hypothesis, there exist $\{\Dl_i\}_{i \in \calI}$, $\{C_i\}_{i \in \calI}$, and $\{q_i\}_{i \in \calI}$ such that
    \begin{itemize}
      \item for all $i \in \calI$, $\Dl_i,(x:A_i) \entailpot{q_i} P :: (z : C_i)$, and
      \item $\tsum^{\m{L}}_{i \in \calI} p_i \cdot \Dl_i = \Dl$,
      $\tsum^{\m{R}}_{i \in \calI} p_i \cdot C_i = C$,
      $\tsum_{i \in \calI} p_i \cdot q_i = q-r$.
    \end{itemize}
    
    Then for each $i \in \calI$, $\Dl_i ,(x:A_i) \entailpot{q_i+r} (\ework{r} \semi P) :: (z : C_i)$.
    
    We conclude by the following:
    \begin{small}
    \begin{align*}
      \tsum_{i \in \calI} p_i \cdot (q_i+r) & = \tsum_{i \in \calI} p_i \cdot q_i + r \\
      & = (q-r)+r \\
      & = q.
    \end{align*}
    \end{small}
  \end{itemize}
\end{proof}

\begin{lemma}\label{Lem:ProcPropagateRight}
  If $\Dl \entailpot{q} P :: (x : A)$ and $A = \tsum_{i \in \calI}^{\m{L}} p_i \cdot A_i$,
  then there exist $\{\Dl_i\}_{i \in \calI}$ and $\{q_i\}_{i \in \calI}$ such that
  \begin{itemize}
    \item for all $i \in \calI$, $\Dl_i \entailpot{q_i} P :: (x : {A_i})$, and
    \item $\tsum_{i \in \calI}^{\m{L}} p_i \cdot \Dl_i = \Dl$, $\tsum_{i \in \calI} p_i \cdot q_i = q$.
  \end{itemize}
\end{lemma}
\begin{proof}
  Similar to the proof of \cref{Lem:ProcPropagateLeft}.
\end{proof}

Now we lift the type judgments from expressions to processes, configurations, and distributions, as well as extend \cref{Lem:ProcPropagateLeft,Lem:ProcPropagateRight} accordingly.
\begin{mathpar}\footnotesize
  \inferrule
  { \Dl \entailpot{q} P :: (c : A)
  }
  { \Dl \potconf{q+w} \proc{c}{w,P} :: (c : A) }
  \and
  \inferrule
  { x_1:A_1,x_2:A_2,\cdots,x_n:A_n \potconf{q} \calC :: (z : C) \\\\
    \Forall{i \in \{1,2\cdots,n\}}
    \Dl_i \potconf{q_i} \calC_i :: (x_i : A_i)
  }
  { \Dl_1,\Dl_2,\cdots,\Dl_n \potconf{q+\tsum_{i=1}^n q_i} (\calC\;\calC_1\;\calC_2\;\cdots\;\calC_n) :: (z : C) }
  \and
  \inferrule
  { \mu = \{ \calC_i : p_i \}_{i \in \calI} \\
    \Forall{i \in \calI} \Dl_i \potconf{q_i} \calC_i :: (x : A_i) \\
    \Dl = \tsum^{\m{L}}_{i \in \calI} p_i \cdot \Dl_i \\
    A = \tsum^{\m{R}}_{i \in \calI} p_i \cdot A_i \\
    q = \tsum_{i \in \calI} p_i \cdot q_i
  }
  { \Dl \potconf{q} \mu :: (x : A) }
\end{mathpar}

\begin{corollary}\label{Cor:ConfPropagateLeft}
  If $\Dl,(x:A) \potconf{q} \calC :: (z : C)$ and $A = \tsum_{i \in \calI}^{\m{R}} p_i \cdot A_i$,
  then there exist $\{\Dl_i\}_{i \in \calI}$, $\{C_i\}_{i \in \calI}$, and $\{q_i\}_{i \in \calI}$ such that
  \begin{itemize}
    \item for all $i \in \calI$, $\Dl_i,(x:{ A_i}) \potconf{q_i} \calC :: (z : C_i)$, and
    \item $\tsum_{i \in \calI}^{\m{L}} p_i \cdot \Dl_i = \Dl$, $\tsum_{i \in \calI}^{\m{R}} p_i \cdot C_i=  C$, $\tsum_{i \in \calI} p_i \cdot q_i = q$.
  \end{itemize}
\end{corollary}
\begin{proof}
  By induction on the derivation of $\Dl,(x:A) \potconf{q} \calC :: (z : C)$.
  
  \begin{itemize}
    \item
    $\small
    \inferrule
    { \Dl,(x:A) \entailpot{q} Q :: (c : C)
    }
    { \Dl,(x:A) \potconf{q+w} \proc{c}{w,Q} :: (c : C) }
    $
    
    Appeal to \cref{Lem:ProcPropagateLeft}.

    \item
    $\small
    \inferrule
    { x_1:B_1,x_2:B_2,\cdots,x_n:B_n \potconf{q_0} \calC_0 :: (z : C) \\
      \Forall{ j \in \{1,2,\cdots,n\} }
      \Dl_j \potconf{q_j} \calC_j :: (x_j : B_j)
    }
    { \Dl_1,\Dl_2,\cdots,\Dl_n \potconf{\tsum_{j=0}^n q_j} (\calC_0\;\calC_1\;\calC_2\;\cdots\;\calC_n) :: (z : C) }
    $
    
    WLOG let us assume $\Dl_1 = \Gm , (x : A)$.
    
    By induction hypothesis on $\Gm ,(x:A) \potconf{q_1} \calC_1 :: (x_1: B_1)$,
    there exist $\{\Gm_i\}_{i \in \calI}$, $\{B_{1,i}\}_{i \in \calI}$, and $\{q_{1,i}\}_{i \in \calI}$ such that
    \begin{itemize}
      \item for all $i \in \calI$, $\Gm_i,(x: {A_i}) \potconf{q_{1,i}} \calC_1 :: (x_1 : B_{1,i})$, and
      \item $\tsum^{\m{L}}_{i \in \calI} p_i \cdot \Gm_i = \Gm$, $\tsum_{i \in \calI}^{\m{R}} p_i \cdot B_{1,i} = B_1$, $\tsum_{i \in \calI} p_i \cdot q_{1,i} = q_1$.
    \end{itemize}
    
    By induction hypothesis on the typing judgment for $\calC_0$, there exist $\{B_{j,i}\}_{i \in \calI}$ for each $j \neq 1$, $\{C_i\}_{i \in \calI}$, and $\{q_{0,i}\}_{i \in \calI}$ such that
    \begin{itemize}
      \item for all $i \in \calI$, $x_1: { B_{1,i}},x_2:B_{2,i},\cdots,x_n:B_{n,i} \potconf{q_{0,i}} \calC_0 :: (z : C_i)$, and
      \item $\tsum^{\m{L}}_{i \in \calI} p_i \cdot B_{j,i} = B_j$ for each $j \neq 1$, $\tsum^{\m{R}}_{i \in \calI} p_i \cdot C_i = C$, $\tsum_{i \in \calI} p_i \cdot q_{0,i} = q_0$.
    \end{itemize}
    
    For each $j \neq 1$, by \cref{Cor:ConfPropagateRight}, there exist $\{\Dl_{j,i}\}_{i \in \calI}$ and $\{q_{j,i}\}_{i \in \calI}$ such that
    \begin{itemize}
      \item for all $i \in \calI$, $\Dl_{j,i} \potconf{q_{j,i}} \calC_j :: (x_j : { B_{j,i}})$, and
      \item $\tsum_{i \in \calI}^{\m{L}} p_i \cdot \Dl_{j,i} = \Dl_j$, $\tsum_{i \in \calI} p_i \cdot q_{j,i} = q_j$.
    \end{itemize}
    
    For each $i \in \calI$, by the fact that
    \begin{itemize}
      \item $\Gm_i,(x:{A_i}) \potconf{q_{1,i}} \calC_1 :: (x_1 : B_{1,i})$,
      \item $x_1:{B_{1,i}}, x_2:B_{2,i}, \cdots, x_n : B_{n,i} \potconf{q_{0,i}} \calC_0 :: (z : C_i)$, and
      \item for each $j \neq 1$, $\Dl_{j,i} \potconf{q_{j,i}} \calC_j :: (x_j : {B_{j,i}})$,
    \end{itemize}
    we derive
    \[
      \Gm_i,(x:{A_i}), \Dl_{2,i},\cdots,\Dl_{n,i}  \potconf{\tsum_{j=0}^n q_{j,i}} (\calC_0\;\calC_1\;\calC_2\;\cdots\;\calC_n) :: (z : C_i).
    \]
    
    We conclude by the following:
    \begin{small}
    \begin{align*}
      \tsum_{i \in \calI} p_i \cdot \tsum_{j=0}^n q_{j,i} & = \tsum_{j=0}^n \tsum_{i \in \calI} p_i \cdot q_{j,i} \\
      & = \tsum_{j=0}^n q_j.
    \end{align*}
    \end{small}
  \end{itemize}
\end{proof}

\begin{corollary}\label{Cor:ConfPropagateRight}
  If $\Dl \potconf{q} \calC :: (x : A)$ and $A = \tsum_{i \in \calI}^{\m{L}} p_i \cdot A_i$,
  then there exist $\{\Dl_i\}_{i \in \calI}$ and $\{q_i\}_{i \in \calI}$ such that
  \begin{itemize}
    \item for all $i \in \calI$, $\Dl_i \potconf{q_i} \calC :: (x : {A_i} )$, and
    \item $\tsum_{i \in \calI}^{\m{L}} p_i \cdot \Dl_i = \Dl$, $\tsum_{i \in \calI} p_i \cdot q_i = q$.
  \end{itemize}
\end{corollary}
\begin{proof}
  Similar to the proof of \cref{Cor:ConfPropagateLeft}, but appeal to \cref{Lem:ProcPropagateRight}.
\end{proof}

Finally, we can formulate and prove \emph{preservation} of this type system.

\begin{theorem}\label{The:PreservationForDetKernel}
  If $\Dl \potconf{q} \calC :: (x : A)$ and $\calC \ostep \calC'$, then $\Dl \potconf{q} \calC' :: (x : A)$.
\end{theorem}
\begin{proof}
  Appeal to the preservation of resource-aware session types~\cite{Das18RAST}.
\end{proof}

\begin{theorem}\label{The:PreservationForKernel}
  If $\Dl \potconf{q} \calC :: (x : A)$ and $\calC \pstep \mu'$, then $\Dl \potconf{q} \mu' :: (x : A)$.
\end{theorem}
\begin{proof}
  The deterministic case appeals to \cref{The:PreservationForDetKernel}.
  For the probabilistic case, i.e., where one process in $\calC$ evaluates a flip expression, we proceed by induction on the derivation of $\Dl \potconf{q} \calC :: (x : A)$.
  The intuition is to use \cref{Cor:ConfPropagateLeft,Cor:ConfPropagateRight} to propagate the type adjustment on the tree.
\end{proof}

\begin{theorem}[Preservation]\label{The:Preservation}
  If $\Dl \potconf{q} \mu :: (x : A)$ and $\mu \dstep \mu'$, then $\Dl \potconf{q} \mu' :: (x : A)$.
\end{theorem}
\begin{proof}
  Appeal to \cref{The:PreservationForKernel}.
\end{proof}

Note that in this approach, it is \emph{unnecessary} to prove \emph{global progress} of the type system.
Instead, in the distribution-to-distribution semantics, we cannot make a step on a distribution if and only if all the configurations in the support of the distribution are poised.
Then we can directly apply the global progress of resource-aware session types~\cite{Das18RAST}.

\end{document}